\documentclass[12pt]{article}
\usepackage{amsfonts,latexsym,amsthm,amssymb,amsmath,amscd,euscript}
\usepackage{framed}
\usepackage{array}
\usepackage{subcaption}
\usepackage{natbib}
\usepackage{pgfplots}
\pgfplotsset{compat = 1.18}
\usepgfplotslibrary{colormaps}
\usepackage{standalone}
\usepackage{fullpage}
\usepackage{comment}
\usepackage{algorithm}
\usepackage{algpseudocode}
\usepackage{rotating}
\usepackage{soul}
\usepackage{xcolor}
\usepackage{graphicx}

\usepackage{caption}
\usepackage{stmaryrd}
\usepackage{hyperref}
    \hypersetup{colorlinks=true,citecolor=blue,urlcolor =black,linkbordercolor={1 0 0}}
\usepackage[margin=1in]{geometry}
\usepackage{dsfont}
\usepackage{ stmaryrd }
\usepackage{listings}
\usepackage{setspace}
\usepackage{float}
\usepackage[shortlabels]{enumitem}
\usepackage{dcolumn}
\usepackage{refcount}
\makeatletter
\newcommand{\primecounter}[2]{%
  \expandafter\renewcommand\csname the#1\endcsname{\getrefnumber{#2}$'$}%
  \expandafter\def\csname theH#1\endcsname{#2prime}%
}
\newcommand{\dprimecounter}[2]{%
  \expandafter\renewcommand\csname the#1\endcsname{\getrefnumber{#2}$''$}%
  \expandafter\def\csname theH#1\endcsname{#2dprime}%
}
\makeatother

\allowdisplaybreaks[1]
\numberwithin{equation}{section}
\usepackage{tabularx}

\newtheorem{theorem}{Theorem}

\newtheorem{proposition}{Proposition}
\newtheorem{lemma}{Lemma}
\newtheorem{corollary}{Corollary}

\theoremstyle{definition}
\newtheorem{defn}{Definition}
\newtheorem{intuition}{Intuition}
\newtheorem{assumption}{Assumption}
\newtheorem{example}{Example}
\newtheorem{runningex}{Example}

\newtheorem{remark}{Remark}
\theoremstyle{remark}

\usepackage{bm}
\let\Lstroke\L
\renewcommand{\L}{\mathcal{L}}
\newcommand{\Y}{\mathcal{Y}}
\newcommand{\EE}{\mathcal{E}}
\newcommand{\B}{\mathcal{B}}
\newcommand{\Z}{\mathbb{Z}}
\newcommand{\U}{\mathcal{U}}
\newcommand{\R}{\mathbb{R}}
\newcommand{\N}{\mathbb{N}}
\newcommand{\C}{\mathbb{C}}
\newcommand{\PP}{\mathcal{P}}
\newcommand{\Leb}{\operatorname{Leb}}
\newcommand{\F}{\mathcal{F}}
\newcommand{\M}{\mathcal{M}}
\newcommand{\E}{\mathbb{E}}
\newcommand{\V}{\mathcal{V}}
\newcommand{\G}{\mathcal{G}}
\newcommand{\T}{\mathcal{T}}
\renewcommand{\P}{\mathds{P}}
\newcommand{\1}{\mathds{1}}
\newcommand{\dist}{\operatorname{dist}}
\renewcommand{\d}{\delta}
\newcommand{\Var}{\operatorname{Var}}
\newcommand{\g}{\gamma}
\renewcommand{\b}{\beta}
\renewcommand{\a}{\alpha}
\newcommand{\e}{\epsilon}
\renewcommand{\t}{\theta}

\title{Singular Equilibrium and Selection of Narratives}
\author{David J. Jin\thanks{\protect\setstretch{1}Department of Economics, Yale University. I am indebted to Kai Hao Yang for his invaluable mentorship. I thank Marina Halac, Navin Kartik, Tianhao Liu, Stephen Morris, Philipp Strack, and Daniel Windisch for helpful conversations. This project would not have been possible without Sumio Watanabe's textbook \textit{Algebraic Geometry and Statistical Learning Theory}. This material is based upon work supported by the National Science Foundation Graduate Research Fellowship Program under Grant No. DGE-2139841. Any opinions, findings, and conclusions or recommendations expressed in this material are those of the author and do not necessarily reflect the views of the National Science Foundation.}}
\date{\today}
\begin{document}
\maketitle

\begin{abstract}
\noindent An agent learns from data that many parameters of her model explain equally well. Berk's theorem states that her belief concentrates on the parameters that best fit the data, but does not select between them. We show that, with a continuum of parameters, the posterior concentrates on the best-fitting parameters with the smallest \textit{local learning coefficient}. We propose a definition of \textit{narratives}, subsets of the best-fitting parameters on which the learning coefficient is constant. Narratives are characterized by the number of coincidences they require (\textit{complexity}) and how exactly each must hold (\textit{robustness}). When the agent's actions affect her data, every long-run action is a \textit{singular equilibrium}, a best reply to a belief supported on the least complex and most robust narratives. Singular equilibrium refines Berk-Nash and self-confirming equilibrium by restricting overly complex or knife-edge off-path beliefs, and its strict uniform version characterizes uniformly stable actions.
 
\end{abstract}

\newpage

\section{Introduction}

\subsection{Motivation and Overview}

We study a Bayesian agent who uses a subjective model to learn from data. In the exogenous learning framework, the agent observes outcomes drawn from an objective distribution and updates her belief over parameters, each of which specifies a subjective distribution of outcomes. In the endogenous learning framework, the agent also takes actions, which govern the objective distribution of the outcome. We accommodate both \textit{correctly specified} models, in which some parameter induces the true distribution of outcomes under every action, and \textit{misspecified} models, in which none does.

A large body of work studies the case in which the agent has finitely many parameters, each of which induces a distinct subjective distribution. By contrast, we assume the agent's parameter set is a continuum and several parameters may correspond to the same subjective distribution. This allows us to model agents whose data cannot distinguish economically distinct explanations. Examples include a monopolist whose sales cannot distinguish demand curves that agree at its price \citep{mclennan}, a policymaker whose data cannot distinguish natural from unnatural rate theories \citep{sargent1976}, and a worker who cannot tell whether her poor outcomes reflect her ability or her circumstances.

We propose a definition of \textit{narratives}, which formalize these economically distinct explanations. A narrative is a set of parameters that best fit the data for the same reason, described by coincidences, restrictions on the parameters that must hold exactly. We characterize narratives by their \textit{complexity}, the number of coincidences they require, and their \textit{robustness}, how exactly each must hold. For example, a worker may attribute her poor outcome solely to her lack of ability (one coincidence), or she may maintain that she is productive but her output is neutralized by her circumstances (two coincidences). 

In the exogenous learning framework, Berk's theorem implies that the agent's beliefs concentrate on the parameters that best fit her data, but does not say which of them she comes to believe \citep{berk}. Our first main result, Theorem \ref{main1}, refines Berk's theorem in the setting with a continuum of parameters and shows that her beliefs concentrate on the narratives with the smallest \textit{local learning coefficient}, a quantity from singular learning theory \citep{watanabe}. The worker comes to attribute her poor outcome solely to her lack of ability, although the second narrative fits her data equally well.

We show that the set of best-fitting parameters admits a \textit{narrative decomposition}, a partition into finitely many narratives, on each of which the local learning coefficient is constant. Our second main result, Theorem \ref{twosided}, then characterizes the local learning coefficient of narratives by their complexity and robustness. The complexity of a narrative is the number of coincidences it requires, and its robustness measures how exactly each coincidence must hold for nearby parameters to fit almost as well. Higher complexity raises the local learning coefficient and higher robustness lowers it, so the posterior favors simpler and more robust narratives. This is the sense of complexity in the principle of parsimony known as \textit{Occam's razor}, as in the first of Newton's rules of reasoning: ``No more causes of natural things should be admitted than are both true and sufficient to explain their phenomena" \citep[Book III, Rule I]{newton}. While most previous work in economics assumes complexity aversion as a preference of the agent, we show that it naturally emerges from Bayesian updating.

We then define an equilibrium solution concept, \textit{singular equilibrium}, an action that is a myopic best reply to a belief supported on the narratives with the smallest local learning coefficient. Singular equilibrium refines Berk-Nash equilibrium \citep{bne} and self-confirming equilibrium \citep{sce} by requiring that the agent's belief be supported on the least complex and most robust narratives. Our third main result, Theorem \ref{main2}, which refines \citet{fls} with a continuum of parameters, shows that every long-run action sustained by a forward-looking agent must be a singular equilibrium. Conversely, our fourth main result, Theorem \ref{main3}, shows that learning converges to an action from every prior that favors its selected narratives if and only if the action is a strict uniform singular equilibrium, extending the stability result of \citet{fls}. Thus, singular equilibrium suffices to predict long-run behavior.

A singular equilibrium need not exist, which implies the existence of \textit{cycles}. Surprisingly, cycles may persist even in correctly specified settings. \citet{ek} show that a correctly specified Bayesian agent's beliefs converge for almost every true parameter under her prior. Under correct specification, some parameter best fits the data under every action, and if it also has the smallest local learning coefficient under every action, a singular equilibrium exists. Singular equilibrium can therefore fail to exist only if, under some action, another best-fitting parameter has a smaller local learning coefficient. This happens only for a prior-measure zero set of parameters, so it may be contained in the null set of \citet{ek}. To our knowledge, ours is the first paper to analyze how cycles emerge under correct specification in the endogenous learning framework with a finite-dimensional parameter space and a smooth prior.

We apply our results in four settings in which several parameters fit the agent's data equally well for economically distinct reasons. First, in models of unobserved heterogeneity, an agent who allows for more types than the population has can explain her data by absent types or by identical types. Second, in the style of \citet{hks}, we analyze the behavior of an underconfident agent who is also ``hesitant." Third, an arbitrarily small misspecification can change which narrative the agent adopts. Fourth, singular equilibrium describes the behavior of an agent who takes actions based on predictions of a neural network, such as a large language model.

In the Supplemental Appendix, we extend our results to general outcome spaces, our first and third main results to settings in which several subjective distributions fit the data equally well, and our second main result to narratives in which the fit worsens only when several coincidences fail simultaneously. We bracket the smallest set to which the posterior converges without a rate. This set can break ties between narratives with the same learning coefficient. Lastly, we show that an agent who switches between models, as in \citet{ba}, settles on the model with the smaller learning coefficient.

\subsection{Related Literature}

\paragraph{Learning and misspecified learning.}
Our paper builds on the literature that derives equilibrium from learning, including \citet{sce}, \citet{fk1993,fk1995}, \citet{fl1993b}, \citet{dfl2004}, \citet{fl2006}, and \citet{bcmm2015}.. Its misspecified branch, which has grown since \citet{bne} defined Berk-Nash equilibrium, includes \citet{frs2017}, \citet{hks,hks2021}, \citet{ep2021}, \citet{epy2021}, \citet{bohrenhauser}, \citet{fii2023}, \citet{fl2023}, \citet{helibgober}, and \citet{ba}. Our paper is most similar to \citet{fls}, who show that every limit action is a uniform Berk-Nash equilibrium and that an action is uniformly stable if and only if it is a uniformly strict Berk-Nash equilibrium. Their equilibrium concepts restrict beliefs only to the parameters that minimize the \textit{Kullback-Leibler divergence} from the true distribution.

\paragraph{Singular learning theory.} Singular learning theory is the study of statistical learning in singular models, in which the map from parameters to distributions is not one-to-one or the Fisher information matrix is degenerate.
To our knowledge, ours is the first paper to use singular learning theory to study learning and equilibrium in economics. Watanabe developed singular learning theory in a series of papers and compiled them in a textbook \citep{watanabe}, which uses the algebro-geometric resolution of singularities \citep{hironaka} to derive asymptotic expansions of the marginal likelihood and the generalization error in singular models. Much of this theory assumes correct specification (realizability), and \citet{renorm} extends it to the misspecified setting (unrealizability), under a renormalizability condition, which is our Assumption \ref{fc2}(iv). \citet{watanabe} defines the key global quantity, the learning coefficient, which equals the \textit{real log canonical threshold} of algebraic geometry, and \citet{lau2025local} define the local learning coefficient around a fixed minimum, assuming no nearby minimizer is more degenerate. Because the learning coefficient replaces half the number of parameters in these expansions, it is mainly used as a criterion for selecting between singular models by their marginal likelihood or their predictive accuracy. We instead use the local learning coefficient to select among parameters of a single model that explain the data equally well. \citet{almost} discusses how almost all models used in machine learning have distinct parameters that induce the same distribution and the local learning coefficient varies over them. Singular models include layered neural networks, large language models, normal mixtures, binomial mixtures, Bayes networks, hidden Markov models, Boltzmann machines, stochastic context-free grammars, and reduced-rank regressions.

\paragraph{Cycles.}
\citet{nyarko} shows that a misspecified agent's actions may cycle forever, because the subjective distributions that best fit the data generated by each action make another action optimal. \citet{frs2017} show that a patient agent's beliefs may fail to converge even when a myopic agent's converge, \citet{epy2021} characterize the asymptotic behavior of the frequency of actions, and \citet{lanzani2025} shows that an endogenous concern for misspecification induces cycles. \citet{fls} obtain Nyarko's nonconvergence as a corollary of their necessity result, as a uniform Berk-Nash equilibrium need not exist. Unlike the cycles in past work, ours are induced by parameters, instead of by subjective distributions, and can persist even in correctly specified settings.

\paragraph{Narratives.}
\citet{shiller} describes narratives as simple stories that explain events and argues that they drive economic fluctuations. In \citet{es2020} and \citet{egs2025}, a narrative is a causal model fit to data, and the narratives that prevail are those that maximize anticipatory utility or political mobilization. In \citet{ss2021}, persuaders propose models, and receivers favor those that better explain the data. \citet{flynn} analyze narratives in macroeconomics. \citet{fm2026} show that penalizing complex models can eliminate the cycles in \citet{lanzani2025}. To our knowledge, ours is the first paper to show that Bayesian updating alone selects among narratives that explain the data equally well, by their complexity and robustness and independently of the prior.

\paragraph{Complexity.}

Economists have modeled complexity with tools from computer science, from finite-automaton strategies \citep{rubinstein1986,ar1988} to the computational hardness of finding regularities in data \citep{agps2005}. A recent literature studies how people respond to complex decisions \citep{oprea2020,eg2023,gg2024,puri2025}, and \citet{gabaix2025} models complexity aversion as dread from anticipated mistakes, which, when strong, makes optimal prices, taxes, and contracts simple. Closest to us is work on the complexity of models. Inductive inference without a subjective criterion such as simplicity overfits \citep{gs2012}, agents with simpler models are the most confident when data are scarce \citep{moopp2022}, and \citet{fm2026} add a penalty on complex worst-case models to an endogenous concern for misspecification. In most of this work, complexity is a cost or preference the modeler assumes. While the advantage of simplicity emerges from estimation \citep{moopp2022}, it vanishes as data accumulate. In our paper, complexity and robustness emerge from Bayesian updating and govern beliefs asymptotically. The Bayesian Occam's razor \citep{jb1992,mackay1992} selects between models. Ours selects among the parameters of one model, and robustness can override complexity.

\section{Model}\label{model}

\subsection{Endogenous Learning}

In each time period $t\in \{1,2,\hdots\}$, an agent chooses an action $a_t$ from a finite action set $A$, which induces an objective probability distribution $P^{a_t}_*$ over the finite set of possible outcomes $Y$. The chosen action $a_t$ and the realized outcome $y_t$ drawn from $P^{a_t}_*$ together determine the agent's payoff for that period via her utility function $u:A\times Y\to \R$.

The agent is uncertain about the objective distributions $\{P^a_*\}_{a\in A}$, but possesses a \textit{model}, which is a triple $(\Theta, \{Q^a_\t\}_{\t\in \Theta, a\in A}, \mu_0)$ consisting of her parameter space $\Theta\subseteq \R^d$, subjective distributions $Q^a_\t\in \Delta(Y)$ for all $a\in A$ and $\t\in \Theta$, and a prior belief over parameters $\mu_0\in \Delta(\Theta)$. We refer to the agent's model and the objective distributions together as the \textit{environment} $(\Theta, \{Q^a_\t\}_{\t\in \Theta, a\in A}, \mu_0, \{P^a_*\}_{a\in A})$. We introduce the first of our two assumptions maintained throughout.

\begin{assumption}\label{fc1}
    The environment $(\Theta, \{Q^a_\t\}_{\t\in \Theta, a\in A}, \mu_0, \{P^a_*\}_{a\in A})$ satisfies
    \begin{enumerate}[(i)]
        \item (Common support): $\{y\in Y : P^a_*(y)>0\} = \{y\in Y : Q^a_\t(y)>0\}$ for all $\t\in \Theta$ and $a\in A$.
        \item (Compact parameter space): $\Theta$ is compact and convex, and $\Theta=\{\t\in \Theta^{(R)}: q_1(\t) \geq 0, \hdots, q_k(\t) \geq 0\}$ where $q_1,\hdots, q_k \not\equiv 0$ are real analytic on some connected open set $\Theta^{(R)}\subseteq \R^d$.
        \item (Positive smooth prior): $\mu_0$ is absolutely continuous with respect to the Lebesgue measure on $\R^d$, with a density that is $C^\infty$ and strictly positive on $\Theta^{(R)}$.
    \end{enumerate}
\end{assumption}

Assumption \ref{fc1}(i) requires that the supports of the agent's subjective distributions be the same as the supports of the objective distributions---this assumption is identical to Assumption 1(i) in \citet{fls}. Assumption \ref{fc1}(ii) allows any closed ball, box, simplex, or convex polytope as the parameter space. Assumption \ref{fc1}(iii) restricts attention to smooth, strictly positive prior beliefs. In particular, the agent entertains a continuum of parameters and $\{Q^a_\t\}_{\t\in \Theta}$ is infinite (uncountable) for all $a\in A$.

The agent updates her prior belief $\mu_0$ over parameters using Bayes' rule. Let $\mu_t(\cdot\mid (a^t,y^t))$ be the agent's posterior belief over her parameter set $\Theta$ at the end of period $t\in \N$ after choosing actions $a^t=(a_s)_{s\leq t}$ and observing outcomes $y^t = (y_s)_{s\leq t}$. That is, her posterior belief assigns a probability to $U\subseteq \Theta$:
\[
\mu_t(U \mid (a^t,y^t)) = \frac{\int_U\left(\prod_{s\leq t} Q^{a_s}_\t(y_s)\right)\mu_0(\t) d\t}{\int_\Theta\left(\prod_{s\leq t} Q^{a_s}_\t(y_s)\right)\mu_0(\t) d\t}.
\]
Next, we establish the action dynamics. A \textit{policy} is a mapping from histories to actions $\sigma: \bigcup_{t\geq 0} (A\times Y)^t\to A$. The agent's objective is to choose an \textit{optimal} policy $\sigma$, i.e., one that maximizes her expected discounted value of per-period utility with discount factor
$\b \in [0,1)$,
\[
V(\sigma, \mu_t) = (1-\b)\E_{\sigma,\mu_t}\left[\sum_{s\geq 1} \b^{s-1} u(a_s,y_s)\right],
\]
at every $t \geq 0$ along every history for which $\mu_t$ is well defined. Standard dynamic programming arguments (Lemma \ref{dp}) guarantee that an optimal policy exists and depends on the history only through the belief $\mu_t$, so we restrict attention to belief-based policies $\pi : \Delta(\Theta) \to A$ when analyzing optimal play. Given a belief-based policy, the action process is defined recursively as $a_{t+1} = \pi(\mu_t)$ and $y_{t+1}\sim P^{a_{t+1}}_*$ with $\mu_{t+1}$ the Bayes update of $\mu_t$ after $(a_{t+1},y_{t+1})$. The objective distributions $\{P^a_*\}_{a\in A}$, the prior $\mu_0$, and a belief-based policy $\pi$ induce a probability measure $\P^{\mu_0}_\pi$ on the space of infinite histories $(a_\tau, y_\tau)_{\tau\geq 1}$.

The \textit{exogenous learning framework} is the special case with only one possible action $A=\{a\}$, in which case the agent observes i.i.d. draws $y_1,y_2,\hdots \sim P_*$, so only her belief-updating process matters. All objects defined in this paper for the endogenous learning framework have exogenous counterparts, which we denote by suppressing the action $a$. For example, $P^a_*$ is the objective distribution associated with action $a$ in the endogenous setting, while $P_*$ is the objective distribution in the exogenous setting.

\subsection{Kullback-Leibler Divergence}

The Kullback-Leibler (KL) divergence from objective distributions to the agent's subjective distributions is $D_{KL}(P\|Q) := \sum_y P(y) \log(P(y)/Q(y))$. For each $a\in A$, we define the set of minimizers $\Theta^*(a) \subseteq \Theta$ by,
\begin{equation}\label{eq:klminimizers}
\Theta^*(a) := \arg\min_{\t' \in \Theta} D_{KL}(P^a_*\|Q^a_{\t'}).
\end{equation}
In the exogenous learning framework, we denote the KL-minimizing set by $\Theta^*$. The environment is \textbf{correctly specified} if there exists a $\t\in \Theta$ with $Q^a_\t(y) = P^a_*(y)$ for all $y\in Y$ and $a\in A$, and \textbf{misspecified} otherwise. In words, the agent's model is correctly specified when a single parameter in her model exactly matches the objective distribution of every action she can take. All results in this paper hold for both the correctly specified and misspecified settings, unless stated otherwise. We therefore work with the \textit{recentered KL divergence}
\[
K^a(\t) := D_{KL}(P^a_*\|Q^a_\t) - \min_{\t'\in \Theta} D_{KL}(P^a_* \|Q^a_{\t'}),
\]
the minimum of which is always zero. We abuse terminology and refer to $K^a(\t)$ simply as the KL divergence. Note that $\Theta^*(a) = \arg\min_{\t'\in \Theta}K^a(\t')$.

A central theorem that has been applied in the economic learning literature is Berk's theorem \citep{berk}, which asserts that, in the exogenous learning framework, the posterior belief over $\Theta$ asymptotically concentrates on the KL minimizers $\Theta^*$. We introduce our central running example, which we enrich throughout the paper and in which an agent learns about her ability and her circumstances.

\begin{runningex}[Ability and circumstances]\label{staticexo}
    Consider the exogenous setting. The agent observes an unfavorable objective distribution $P_* = N(0,1)\otimes N(0,1)$, e.g., an inert career. The agent aims to explain this using a model with parameter space $\Theta = [0,1]^2$ with typical element $\t= (q,c) \in \Theta$ and subjective distributions $Q_\t = N(f_1(\t),1) \otimes N(f_2(\t),1)$ with
    \[
    f_1(q,c) = q(c-d),\quad f_2(q,c) = q(q-c)
    \]
    for some constant $d\in (0,1)$. Here, $q$ represents the agent's ability, and $c$ represents her circumstances, e.g., managerial support, quality of her team/firm, or fit between her skills and her role. She aims to explain the first coordinate of the poor outcome according to her output $f_1(q,c)$, which combines her ability and circumstances. She explains the second coordinate of the outcome according to her output net of the circumstances $f_2(q,c)$. One can compute the KL divergence:
    \[
    K(\t) = K(q,c) = \frac{1}{2}q^2[(c-d)^2 + (q-c)^2].
    \]
    Notice that the environment is correctly specified, and $\min_{\t\in \Theta}K(\t) = 0$ is attained on
    \[
    \Theta^* = \underbrace{\{(0,c) : c\in [0,1]\}}_{I_1} \cup \underbrace{\{(d,d)\}}_{I_2}.
    \]
    We call $I_1$ and $I_2$ \textbf{narratives}, which intuitively are explanations for \textit{why} the KL divergence is minimized. Throughout, we refer to narratives of the type $I_1$ as \textit{self-blame}: $I_1$ attributes her unsuccessful career entirely to having zero ability $q=0$, while the effects of her circumstances, whether favorable or unfavorable, are overshadowed by her inability. We refer to narratives of the type $I_2$ as \textit{circumstance-blame}: $I_2$ maintains that her ability is positive but exactly canceled out by unfavorable circumstances. In Section \ref{narrative}, we formalize narratives and show that $\Theta^*$ always admits a \textit{narrative decomposition}.
\end{runningex}

We now introduce the second of our two assumptions that we maintain throughout the paper.

\begin{assumption}\label{fc2}
    The environment $(\Theta, \{Q^a_\t\}_{\t\in \Theta, a\in A}, \mu_0, \{P^a_*\}_{a\in A})$ satisfies
    \begin{enumerate}[(i)]
        \item (Weak identification): For all $a\in A$, if $K^a(\t) = K^a(\t') = 0$ for $\t,\t'\in \Theta$, then $Q^a_\t(y) = Q^a_{\t'}(y)$ for all $y\in Y$.
        \item (Richness): KL divergence is not identically 0, i.e.,  $K^a\not\equiv 0$ on $\Theta^{(R)}$ for all $a\in A$.
        \item (Analyticity): For each $a\in A$, let $f^a(y,\t) := \log(Q^a_{\t_0}(y) / Q^a_\t(y))$ where $\t_0 \in \Theta^*(a).$ There exists an open set $\Theta^{(C)}\subseteq \C^d$ such that
        \begin{enumerate}[(a)]
            \item $\Theta^{(R)}\subseteq \Theta^{(C)}$.
            \item $\t\mapsto f^a(\cdot,\t)$ admits an extension $\Theta^{(C)} \ni \t\mapsto f^a(\cdot,\t)$ that is an $L^2(P^a_*)$-valued complex analytic function.
            \item $M^a(y) := \sup_{\t\in \Theta^{(C)}} |f^a(y,\t)|$ belongs to $L^2(P^a_*).$ 
        \end{enumerate}
        \item (Renormalizability): There exists $\e>0$ such that for all $a\in A$, there exists an open set $V_a\subseteq \Theta^{(R)}$ with $\Theta^*(a)\subseteq V_a$ and a constant $C>0$ such that
        \[
        \E_{P^a_*}[f^a(y,\t)] \geq C\E_{P^a_*}[f^a(y,\t)^2]
        \]
        for all $\t\in V_a\cap \{\t\in \Theta : K^a(\t) \leq \e\}$ where $K^a(\t) := \E_{P^a_*}[f^a(y,\t)]$ on $\Theta^{(R)}$.
    \end{enumerate}
\end{assumption}

Assumption \ref{fc2}(i) focuses our attention on settings where, if two distinct parameters minimize KL divergence, then they correspond to the same subjective distribution. We emphasize that the selection in our paper is among best-fitting parameters (narratives), instead of best-fitting subjective distributions. Supplemental Appendix \ref{supp:weakid} relaxes weak identification and extends Theorems \ref{main1} and \ref{main2}. Assumption \ref{fc2}(ii) rules out trivial models where every parameter is a best fit. Assumption \ref{fc2}(iii) is mainly technical and guarantees that $\t\mapsto Q^a_\t(y)$ is analytic, and is satisfied by a wide class of standard specifications in economics (linear/probit/logit signal structures, Gaussian signals, exponential families, etc.). Assumption \ref{fc2}(iv) is drawn from \citet{renorm} and intuitively requires that the KL divergence vanish no faster than the distribution moves as $\t$ converges to a KL minimizer in $\Theta^*(a)$. Weak identification and renormalizability are automatically satisfied in correctly specified environments.

\begin{remark}
    In the main text of this paper, we assume that the outcome space $Y$ is finite. With suitable adjustments to Assumptions \ref{fc1} and \ref{fc2} and an integrable utility, all results also hold when the outcome space is a standard Borel space with a $\sigma$-finite measure (Supplemental Appendix \ref{supp:general}), which covers the Gaussian signals in our examples.
\end{remark}

\section{Narratives and Exogenous Convergence}\label{narrsec}

This section develops a way to select between parameters within the KL-minimizing set, in the exogenous learning framework. Section \ref{llc} defines the \textit{local learning coefficient} (LLC) $\lambda(\t)$ of each parameter $\t \in \Theta^*$, and our first main result, Theorem \ref{main1}, shows that the agent's posterior belief over parameters distinguishes parameters in $\Theta^*$ at polynomial rates only through their LLCs. We then propose a formalization of \textit{narratives} as the connected pieces of $\Theta^*$ on which $\lambda$ is constant. Our second main result, Theorem \ref{twosided}, characterizes $\lambda$ by the number of \textit{coincidences} a narrative requires (its \textit{complexity}) and how exactly each coincidence holds (its \textit{robustness}).

\subsection{Local Learning Coefficient}\label{llc}

\begin{defn}[Near-explanations]\label{nearexp}
    Fix $a\in A$ and $\t\in \Theta^*(a)$. For $r>0$ and $\e>0$, the \textbf{near-explanations} of $\t$ at tolerance $\e$ are
    \[
    N^a_\t(r,\e) := \{\t'\in \Theta\cap B(\t,r) : K^a(\t')\leq \e\},
    \]
    and $V^a_\t(r,\e) := \Leb(N^a_\t(r,\e))$ is their volume, where $\Leb$ denotes the Lebesgue measure on $\R^d$.
\end{defn}

\begin{defn}[Local learning coefficient and concentration set]\label{concentrationset}
    Fix $a\in A$. The \textbf{local learning coefficient} of $\t\in \Theta^*(a)$ is
    \[
    \lambda(\t\mid a) := \lim_{r\downarrow 0} \liminf_{\e\downarrow 0} \frac{\log V^a_\t(r,\e)}{\log \e}.
    \]
    The \textbf{concentration set} $\Theta^{**}(a)\subseteq \Theta^*(a)$ is defined by
    \[
    \Theta^{**}(a) := \left\{\t\in \Theta^*(a) : \lambda(\t\mid a) = \underline{\lambda}(a)\right\}
    \]
    where $\underline{\lambda}(a) := \min_{\t\in \Theta^*(a)}\lambda(\t\mid a)$. For $U\subseteq \Theta^*(a)$, we define $\lambda(U\mid a) := \min_{\t\in U}\lambda(\t\mid a)$. Lemma \ref{lambdaprops} shows that the limit exists, that the minimum is attained, and that $\Theta^{**}(a)$ is compact.
\end{defn}

The concentration set is the subset of KL-minimizing parameters that moreover minimize the local learning coefficient $\lambda(\t\mid a)$. We defer the full economic analysis of $\lambda(\t\mid a)$ to Section \ref{candr} and paraphrase the core intuition of Definition \ref{concentrationset} in the following.

\begin{intuition}[Local learning coefficient]\label{rlctintuition}
    The local learning coefficient measures the concentration of parameters \textit{near} $\t$ with similar KL divergence. It is determined by \textbf{complexity} and \textbf{robustness}, which we define in Definitions \ref{simpdef} and \ref{robustdef}.

    Complexity is the codimension of the set of perfect explanations onto which the near-explanations shrink. As $\e\downarrow 0$, the near-explanations shrink onto the parameters near $\t$ that explain the data perfectly, and the set's dimension counts how many coordinates this set does not restrict. In Example \ref{staticexo}, the perfect explanations near a point of $I_1$ form the line $q=0$, so that ability must be zero and circumstances are unrestricted. The perfect explanations near $(d,d)$ form the single point $(d,d)$, so that ability and circumstances must both be $d$. Each coordinate that is fixed is a coincidence that must hold, so a lower dimension of perfect explanations requires more coincidences, which corresponds to higher complexity and a larger LLC.

    Robustness is the size of the set of near-explanations. In Example \ref{staticexo}, the fit at $(d,d)$ is perfect and worsens with the square of the distance from $(d,d)$, so the near-explanations of $(d,d)$ form a neighborhood of radius $\sqrt{\e}$. If instead the fit worsened with the fourth power of the distance, the neighborhood would have radius $\e^{1/4}$, which is larger for $\e<1$, and would shrink more slowly as $\e\downarrow 0$. A more robust parameter corresponds to a smaller LLC.
\end{intuition}

\begin{proposition}[Invariance]\label{invariance}
    Fix $a\in A$. The local learning coefficient $\lambda(\cdot\mid a)$ and the concentration set $\Theta^{**}(a)$ do not depend on the prior $\mu_0$, and are unchanged if the parameters are relabeled by a real analytic isomorphism $\varphi$, i.e., if $\t\mapsto Q^a_\t$ is replaced by $\t\mapsto Q^a_{\varphi(\t)}$.
\end{proposition}
\begin{proof}
    See Appendix \ref{app:lambda}.
\end{proof}

This proposition states that the local learning coefficient does not depend on the prior, and it is determined by the parameter space and the KL divergence alone. Relabeling the parameters as in Proposition \ref{invariance} changes the coordinates of every KL-minimizing parameter, but not its learning coefficient, so if two agents' models differ only by such a relabeling, the relabeling maps one agent's concentration set onto the other's.\footnote{If the prior density in Assumption \ref{fc1}(iii) is only the product of a strictly positive smooth function and a nonnegative real analytic function $\mu_2$ (Definition \ref{fc}(iii)), all results hold with the prior-weighted learning coefficient in place of $\lambda(\cdot\mid a)$, except that Theorem \ref{twosided}, Corollary \ref{occam}, and Theorem \ref{jointtwosided} require $\mu_2>0$ on the narratives involved. The first claim of Proposition \ref{invariance} weakens to: the prior-weighted learning coefficient and the associated concentration set are unchanged if $\mu_0$ is multiplied by a strictly positive smooth function. The reparametrization statement is unchanged.}

\subsection{Exogenous Convergence}\label{exosec}

We now present our first main result of convergence in the exogenous learning framework, which is a refinement of Berk's theorem \citep{berk} in the setting with a continuum of parameters. Berk's theorem asserts that the posterior concentrates on the KL-minimizing set $\Theta^*$. The theorem says nothing further about which parameters within $\Theta^*$ the agent comes to believe. Our result shows that the posterior concentrates on the subset $\Theta^{**}\subseteq \Theta^*$, and that this is tight up to a polynomial rate. To avoid confusion with endogenous convergence, we write $\Theta^{**}(a) = \Theta^{**}$ for the following exogenous convergence results.

\begin{theorem}[Exogenous convergence]\label{main1}
    Consider the exogenous learning framework.
    \begin{enumerate}[(i)]
        \item For every open set $U\supseteq \Theta^{**}$, there is $\e(U) > 0$ such that almost surely $\mu_t(\Theta\setminus U) \leq t^{-\e(U)}$ for all large $t$. In particular, $\mu_t(U)\to 1$ almost surely as $t\to \infty.$
        \item If $F\subseteq \Theta$ is closed and, for every open $U\supseteq F$, there is $\e>0$ such that almost surely $\mu_t(\Theta\setminus U)\leq t^{-\e}$ for all large $t$, then $F\supseteq \Theta^{**}.$
    \end{enumerate}
\end{theorem}
\begin{proof}
    See Supplemental Appendix \ref{supp:main1}.
\end{proof}

Theorem \ref{main1} establishes almost sure concentration on $\Theta^{**}$, and, outside any neighborhood, the posterior mass decays polynomially in $t$. Thus, no limit point of the posterior places any mass on parameters that minimize KL divergence from the objective distribution but do not minimize the local learning coefficient $\lambda(\t)$. Part (ii) guarantees that $\Theta^{**}$ is the smallest closed set that the posterior concentrates on at a polynomial rate.

\begin{proof}[Proof sketch of Theorem \ref{main1}]
    Write $\mu_t(V) = Z_t(V) / Z_t(\Theta)$ where $Z_t(V) := \int_V e^{-t K_t(\t)}\mu_0(\t) d\t$ is the ``evidence" of a set $V$ and $K_t(\t) := t^{-1}\sum_{s\leq t}\log(Q_{\t_0}(y_s)/Q_\t(y_s))$, $\t_0\in \Theta^*$, is the empirical KL divergence. A resolution of singularities \citep{hironaka,watanabe} covers the parameters with small KL divergence with finitely many charts $[0,b]^d$ on which the KL divergence becomes a monomial $u^{2k}$ and the prior times the Jacobian becomes $\psi(u) u^h$ with $\psi$ bounded above and away from zero (Theorem \ref{partitiontheta2}). On each chart, the empirical KL divergence takes the \textit{standard form} $K_t = u^{2k} - t^{-1/2} u^k \xi_t(u)$ of \citet{watanabe}, where $\xi_t$ is an empirical process whose supremum grows no faster than $\sqrt{\log \log t}$ by a compact law of iterated logarithm from \citet{ledouxtalagrand}. Integrating $e^{-tu^{2k}}u^h$ over a chart gives $t^{-\lambda}$ up to logarithmic factors, where $\lambda$ is the smallest ratio $(h_j+ 1)/2k_j$, and Proposition \ref{volumeequiv} shows that these ratios compute the local learning coefficient of Definition \ref{concentrationset}. Hence, the evidence of a small ball around a point of $\Theta^{**}$ is at least $t^{-\underline{\lambda}}$ up to logarithmic factors. On every chart, the coordinates that attain $\underline{\lambda}$ vanish only on the preimage of $\Theta^{**}$, so they are bounded away from zero on the preimage of $\Theta\setminus U$, and the evidence of $\Theta\setminus U$ is at most $t^{-\lambda'}$ up to logarithmic factors for some $\lambda' > \underline{\lambda}$. Their ratio gives part (i). For part (ii), the evidence of any ball around a point of $\Theta^{**}$ is a fraction of $Z_t(\Theta)$ that decays at most logarithmically, so no closed set that excludes such a point attracts the posterior at a polynomial rate.
\end{proof}

\begin{intuition}\label{intuition1}
    The posterior belief concentrates on parameters that are the least complex and most robust, as intuitively defined in Intuition \ref{rlctintuition}. After $t$ observations, the likelihood of a parameter $\t'$ is approximately $e^{-tK(\t')}$, so the posterior effectively ignores parameters with $K(\t')$ much larger than $1/t$ and spreads its mass over the near-explanations $N_\t(r,1/t)$. The posterior mass near $\t$ is therefore of the order of $V_\t(r,1/t)$, which is $t^{-\lambda(\t)}$ up to logarithmic factors.

    \begin{enumerate}[(1)]
        \item \textbf{Complexity.} Consider two KL minimizers, one lying on a line of KL minimizers, as $I_1$ in Example \ref{staticexo}, and one isolated, as $I_2$. At time $t$, the posterior mass near the first is spread over a strip of width $t^{-1/2}$, and the mass near the second is spread over a disc of radius $t^{-1/2}$. The strip loses mass like $t^{-1/2}$ and the disc like $t^{-1}$, so the posterior odds of the first against the second grow like $t^{1/2}$. In general, a parameter on a set of perfect explanations of dimension $d-c$ loses mass in $c$ directions, at rate $t^{-c/2}$ when the fit is quadratic in each. The posterior favors the parameters that require fewer coincidences.
        \item \textbf{Robustness.} Now hold the dimension fixed and consider two isolated KL minimizers, one around which the fit deteriorates with the square of the distance and one around which it deteriorates with the fourth power. At time $t$, the near-explanations extend to distance $t^{-1/2}$ around the first and $t^{-1/4}$ around the second, so the second retains more mass. The posterior favors more robust parameters.
    \end{enumerate}

    The local learning coefficient combines the complexity and robustness into the single rate $t^{-\lambda(\t)}$. The posterior odds between two KL minimizers with coefficients $\lambda_1<\lambda_2$ are of the order $t^{\lambda_2-\lambda_1}$ up to logarithmic factors, which is the polynomial rate in Theorem \ref{main1}. The posterior odds depend on the prior only through the ratio of its densities at the two parameters, which is a constant.
\end{intuition}

We continue Example \ref{staticexo} and compute the local learning coefficients to find the parameters that are limits of exogenous learning.

\begin{runningex}[Computing the local learning coefficient]\label{resex}
    We compute the local learning coefficients in Example \ref{staticexo}. Since $\lambda(\t)$ is the limit of $\log V_\t(r,\e)/\log\e$, constant factors do not affect it, so $V_\t(r,\e)\asymp \e^s$ for small $r$ implies $\lambda(\t) = s$.\footnote{For nonnegative $f$ and $g$, we write $f \asymp g$ on $U$ if there are constants $0 < c \leq C$ such that $cg(u) \leq f(u) \leq Cg(u)$ on $U$, or for all sufficiently large $t$ or small $\e$.} Near a point $(0,c_0)$ of $I_1$, the bracket $(c-d)^2+(q-c)^2$ in $K$ is bounded above and away from zero, so $K\asymp q^2$ on $B((0,c_0),r)$ for small $r$. Then, the near-explanations $N_{(0,c_0)}(r,\e)$ contain and are contained in strips $\{0\leq q\leq c_1\sqrt{\e}\}$ and $\{0\leq q\leq C\sqrt{\e}\}$ inside $B((0,c_0),r)\cap\Theta$, for constants $0<c_1\leq C$. A strip of width $w$ and fixed length has volume $\asymp w$, so $V_{(0,c_0)}(r,\e)\asymp \sqrt{\e}$, and $\lambda((0,c_0)) = 1/2$.

    Near $(d,d)$, the factor $q^2$ is bounded above and away from zero, so $K\asymp (c-d)^2+(q-c)^2$ on $B((d,d),r)$ for small $r$. Then, the near-explanations $N_{(d,d)}(r,\e)$ contain and are contained in the discs of radius of order $\sqrt{\e}$. A disc of radius $\varrho$ has volume $\asymp \varrho^2$, so $V_{(d,d)}(r,\e)\asymp \e$, and $\lambda((d,d)) = 1$. The concentration set is therefore $\Theta^{**} = I_1\subseteq \Theta^*$. By Theorem \ref{main1}, the posterior concentrates on $I_1$, not on $I_2$, even though $I_2$ also minimizes KL divergence.
\end{runningex}

Narrative $I_1$ in Examples \ref{staticexo} and \ref{resex} explains the poor outcome through zero ability (self-blame), while $I_2$ explains it through ability and circumstances coinciding at a knife-edge. Neither is more accurate, as both lie in $\Theta^*$ and perfectly match the data. However, more parameters nearly explain the data near $I_1$ than near $I_2$. Narrative $I_1$ is a line within $\Theta$, and $I_2$ is a point, so the near-explanations of $I_1$ form a strip while those of $I_2$ form a disc, and the strip shrinks more slowly as the data accumulate. The narrative $I_1$ requires only one coincidence on ability, while $I_2$ requires two separate coincidences on ability and circumstances.

\subsection{Narratives}\label{narrative}

The narrative structure we saw in Example \ref{staticexo} holds in general. In every environment, the set of best-fitting parameters is a finite union of connected analytic manifolds, on each of which the local learning coefficient is constant. We call these manifolds \textit{narratives}.

\begin{defn}[Narrative]\label{narrdef}
    Fix $a\in A$. A \textbf{narrative} for $a$ is a connected analytic manifold $I\subseteq\Theta^*(a)$ on which $\lambda(\cdot\mid a)$ is constant. We write $\lambda(I\mid a)$ for the common value of $\lambda(\cdot\mid a)$ on $I$.\footnote{Precisely, what we call a connected analytic manifold contained in $\Theta^*(a)$ need not itself be a manifold. It is a set of the form $M\cap \Theta$, where $M$ is a connected real analytic submanifold of $\Theta^{(R)}$, defined in Assumption \ref{fc1}, with $M\subseteq \{\t\in \Theta^{(R)} : K^a(\t) = 0\}$, such that $M\cap\Theta$ is connected and is the closure, relative to $M$, of its interior relative to $M$. We define it in this way because $\Theta$ has a boundary, so a narrative may fail to be a manifold at $\partial \Theta$.}
\end{defn}
 
\begin{proposition}[Narrative decomposition]\label{narrdecomp}
    Fix $a\in A$ and, for $\ell\in\R$, let $L_\ell:=\{\t\in\Theta^*(a):\lambda(\t\mid a)=\ell\}$.
    \begin{enumerate}[(a)]
        \item (Existence) There exist finitely many pairwise disjoint narratives $I_1,\hdots,I_n$ for $a$ such that $\Theta^*(a)=I_1\cup\cdots\cup I_n$ and, for every $S\subseteq\{1,\hdots,n\}$ with $|S|\geq 2$, the set $\bigcup_{i\in S}I_i$ is not a narrative. Consequently, $\Theta^{**}(a)$ is the union of the $I_i$ with $\lambda(I_i\mid a)=\underline{\lambda}(a)$.
        \item (Uniqueness) Suppose that any two narratives $I,J$ for action $a$ with $\lambda(I\mid a)=\lambda(J\mid a)$ and $\overline{I}\cap\overline{J}\neq\emptyset$ are contained in a common narrative. Then the collection $\{I_1,\hdots,I_n\}$ in (a) is unique.
    \end{enumerate}
\end{proposition}
\begin{proof}
    See Supplemental Appendix \ref{supp:narr}.
\end{proof}

Proposition \ref{narrdecomp}(a) states that $\Theta^*$ can be finitely partitioned into \textit{maximal} narratives, i.e., no union of two or more of these narratives is a narrative. Maximality intuitively rules out splitting a narrative into separate pieces. For instance, a piece of the self-blame line $I_1$ in Example \ref{staticexo}, such as $\{(0,c) : c\leq 1/2\}$, adds the claim that circumstances are not too good, but every point of $I_1$ induces the same distribution and has the same learning coefficient, so the posterior never favors the piece over the whole $I_1$. Proposition \ref{narrdecomp}(b) states that the partition is unique unless two narratives with the same learning coefficient intersect.

As exemplified by the self-blame narrative $I_1 = \{q=0\}$ and the circumstance-blame narrative $I_2 = \{(d,d)\}$ in Example \ref{staticexo}, narratives in general are described by finitely many \textit{coincidences}. That is, near any of its points, each narrative takes the form $I= \{\t\in \Theta : \phi_1(\t) = \cdots = \phi_c(\t) = 0\}$ for distinct analytic functions $\phi_1,\hdots,\phi_c$. Each equation $\phi_\ell(\t)=0$ is a \textit{coincidence}: a restriction on the parameters that must hold exactly for the narrative to explain the data. In Example \ref{staticexo}, narrative $I_1$ requires one coincidence, $q=0$, and narrative $I_2$ requires two coincidences, $c=d$ and $q=c$.

\begin{defn}[Coincidences]\label{coincdef}
    Fix $a\in A$, and let $I$ be a narrative for $a$ with $c:=d-\dim I$. Let $\t_0\in I$, and let $\phi_1,\hdots,\phi_c$ be real analytic functions with linearly independent gradients at $\t_0$ such that, near $\t_0$, $I=\{\t\in \Theta : \phi_1(\t)=\cdots=\phi_c(\t)=0\}.$ We call the equations $\phi_1(\t)=0,\hdots,\phi_c(\t)=0$ the \textbf{coincidences} of $I$ at $\t_0$.
\end{defn}

The condition that the gradients are linearly independent rules out redundant coincidences, like $\phi_1 = 2\phi_2$, and degenerate ones, like $q^2=0$ in place of $q=0$. Note that the functions $\phi$ may differ at distinct points of a single narrative. Near each of its points, a narrative restricts some coordinates and leaves the others free. Different points of the same narrative can restrict different coordinates, so the coincidences depend on the point. For example, consider a narrative that requires one coincidence, that ability times circumstances equals a threshold, $qc = 1/4$. Near $(1/4,1)$, the coincidence restricts ability to about $1/4$, and near $(1,1/4)$, the coincidence restricts circumstances to about $1/4$. While the coincidence function may differ along a narrative, the number of coincidences is always constant.

\subsection{Complexity and Robustness}\label{candr}

In this subsection, we define complexity and robustness as properties of a narrative, in both the correctly specified and misspecified settings. Our second main result shows that complexity and robustness together determine the local learning coefficient, and therefore the limits of exogenous learning.

\begin{defn}[Complexity]\label{simpdef}
    The \textbf{complexity} of a narrative $I$ is its codimension $c(I) := d - \dim I$, the number of its coincidences. A narrative is more complex the larger $c(I)$ is.
\end{defn}

A narrative is more complex if it requires more coincidences to explain the data, exactly in the sense of Occam's razor \citep[Book III, Rule I]{newton}. Next, we define robustness.

\begin{defn}[Robustness]\label{robustdef}
    Fix $a\in A$ and suppress it. Let $I$ be a narrative with coincidences $\phi_1(\t)=0,\hdots,\phi_c(\t)=0$ at $\t_0\in I$. For $\rho=(\rho_1,\hdots,\rho_c)>0$ and $\d,r>0$, let
    \[
    N^r_\d(\rho):=\left\{\t\in B(\t_0,r) : |\phi_\ell(\t)|\leq \d^{1/\rho_\ell}\text{ for all }\ell\right\}.
    \]
    \begin{enumerate}[(a)]
    \item The vector $\rho$ is a \textbf{lower robustness profile} of $I$ if, for some such $\t_0$ and $\phi$, (i) $K\leq C\sum_{\ell=1}^c|\phi_\ell|^{2\rho_\ell}$ on $\Theta$ near $\t_0$ for some $C>0$, and (ii) $\liminf_{\d\downarrow 0}\Leb(\Theta\cap N^r_\d(\rho))/\Leb(N^r_\d(\rho))>0$ for all sufficiently small $r>0$.
    \item The vector $\rho$ is an \textbf{upper robustness profile} of $I$ if, for some such $\t_0$ and $\phi$, $K\geq c_0\sum_{\ell=1}^c|\phi_\ell|^{2\rho_\ell}$ on $\Theta$ near $\t_0$ for some $c_0>0$. 
    \end{enumerate}
    If $\rho$ is both a lower and an upper profile with respect to the same coincidences, then $I$ is \textbf{separably robust} and $\rho_\ell$ is the \textbf{robustness} of coincidence $\ell$.
\end{defn}

A robustness profile measures how precisely each coincidence of $I$ must hold for a parameter to nearly explain the data. The entry $\rho_\ell$ is half the order at which the KL divergence deteriorates when coincidence $\ell$ is violated. When $\rho_\ell$ is small, the fit deteriorates quickly, and the narrative requires that coincidence $\ell$ be maintained sharply. When $\rho_\ell$ is large, the fit deteriorates slowly, and coincidence $\ell$ need only hold loosely. A lower robustness profile bounds the deterioration from above, so each coincidence is at least as robust as $\rho_\ell$ indicates. An upper robustness profile bounds it from below, so each coincidence is at most as robust. Definition \ref{robustdef}(a)(ii) ensures that $N^r_\d$ does not end up almost entirely outside $\Theta$ as $\d\downarrow 0$.

\begin{theorem}\label{twosided}
    Fix $a\in A$ and suppress it. Let $I$ be a narrative with complexity $c$. For any lower robustness profile $\underline{\rho}$ and any upper robustness profile $\bar\rho$ of $I$,
    \[
    \frac{1}{2}\sum_{\ell=1}^c\frac{1}{\bar\rho_\ell} \leq \lambda(I) \leq \frac{1}{2}\sum_{\ell=1}^c\frac{1}{\underline{\rho}_\ell}.
    \]
    In particular, if $(\kappa,\hdots,\kappa)$ is a lower and $(\a,\hdots,\a)$ an upper robustness profile of $I$, then $c/2\a\leq \lambda(I)\leq c/2\kappa$, and if $I$ is separably robust with robustness $\rho$, then $\lambda(I)=\frac{1}{2}\sum_{\ell=1}^c 1/\rho_\ell$.
\end{theorem}
\begin{proof}
    See Appendix \ref{app:twosided}.
\end{proof}

Theorem \ref{twosided} shows that the local learning coefficient is bounded below and above by robustness and complexity. The bounds depend on complexity through the number of summands, one summand per coincidence. Each summand is the reciprocal $1/\rho_\ell$ of robustness. Theorem \ref{twosided} is the rigorous statement of Intuition \ref{rlctintuition}, in which $\lambda$ is jointly determined by complexity and robustness. We isolate each in the following corollary.

\begin{corollary}[Complexity and robustness]\label{occam}
    Fix $a\in A$ and suppress it. Let $I$ and $J$ be narratives, let $\underline{\rho}$ be a lower robustness profile of $I$, and let $\bar\rho$ be an upper robustness profile of $J$.
    \begin{enumerate}[(a)]
        \item (Complexity) If $\underline{\rho} = (\rho,\hdots, \rho) \in \R^{c(I)}$ and $\bar\rho = (\rho,\hdots, \rho) \in \R^{c(J)}$ for some $\rho>0$, then $c(I)<c(J)$ implies $\lambda(I)<\lambda(J)$.
        \item (Robustness) If $c(I)=c(J)$ and $\underline{\rho}_\ell\geq \bar\rho_\ell$ for all $\ell$, then $\underline{\rho}\neq \bar\rho$ implies $\lambda(I)<\lambda(J)$.
    \end{enumerate}
    If $I$ and $J$ are separably robust with robustness $\underline{\rho}$ and $\bar\rho$, then both implications are equivalences.
\end{corollary}
\begin{proof}
    By Theorem \ref{twosided}, $\lambda(I)\leq \frac{1}{2}\sum_{\ell=1}^{c(I)}1/\underline{\rho}_\ell$ and $\lambda(J)\geq \frac{1}{2}\sum_{\ell=1}^{c(J)}1/\bar\rho_\ell$. In (a), this gives $\lambda(I)\leq c(I)/2\rho<c(J)/2\rho\leq \lambda(J)$. In (b), $1/\underline{\rho}_\ell\leq 1/\bar\rho_\ell$ for all $\ell$ with strict inequality for some $\ell$, and both sums have $c(I)=c(J)$ terms, so $\lambda(I)<\lambda(J)$. If $I$ and $J$ are separably robust, then Theorem \ref{twosided} gives equalities in both bounds, so $\lambda(I) < \lambda(J)$ if and only if $c(I) < c(J)$ in (a), and if and only if $\underline{\rho}\neq \bar\rho$ in (b).
\end{proof}

Corollary \ref{occam}(a) states that, holding robustness fixed, the less complex narrative has the smaller learning coefficient. Theorem \ref{main1} therefore suggests that Occam's razor emerges in Bayesian behavior automatically, in the sense that agents converge to the simplest narratives. Corollary \ref{occam}(b) states that, holding complexity fixed, the more robust narrative has the smaller learning coefficient. This can be viewed as a counterpart to Occam's razor: if two narratives are equally complex, then pick the one whose coincidences are less of a ``knife-edge."

\begin{remark}[Existence of profiles]\label{profileremark}
    A lower robustness profile always exists. The KL divergence is real analytic and vanishes on $I$, so violating a coincidence of $I$ increases it at most linearly, and $(\frac{1}{2},\hdots, \frac{1}{2})$ is a lower robustness profile. An upper robustness profile exists if and only if some point of $I$ has a neighborhood that intersects no other narrative of the decomposition in Proposition \ref{narrdecomp}. If every point of $I$ is arbitrarily close to a different narrative, a parameter can violate the coincidences of $I$ by an arbitrarily small amount and still explain the data perfectly inside the other narrative (Proposition \ref{profileexist}).
\end{remark}

Next, we state an example that illustrates the tradeoff between complexity and robustness. In particular, a complex narrative may have a smaller LLC than a simple narrative if it is sufficiently robust.

\begin{runningex}[Robustness]\label{robustex}
    Consider the exogenous setting of Example \ref{staticexo}, with the threshold entering at order $k\in \N$:
    \[
    Q_{(q,c)} = N\left(q(c-d)^k,1\right)\otimes N\left(q(q-c)^k,1\right),\quad P_* = N(0,1)\otimes N(0,1),
    \]
    on $\Theta = [0,1]^2$ with $d\in (0,1)$. The case $k=1$ is Example \ref{staticexo}. When $k$ is large, the agent models small shortcomings as less critical and substantial ones as more critical. For every $k$,
    \[
    K(q,c) = \frac{1}{2}q^2\left[(c-d)^{2k} + (q-c)^{2k}\right],\quad \Theta^* = \underbrace{\{(0,c) : c\in [0,1]\}}_{I_1}\cup \underbrace{\{(d,d)\}}_{I_2},
    \]
    so the two narratives, self-blame and circumstance-blame, and their complexities $c(I_1) = 1$ and $c(I_2) = 2$ do not depend on $k$. The robustness of $I_2$ depends on $k$, however. Near $(0,c_0)\in I_1$, the bracketed factor is bounded above and away from zero, so $K\asymp q^2$. The coincidence $q=0$ has robustness $1$, and Theorem \ref{twosided} gives $\lambda(I_1) = 1/2$. Near $(d,d)$, the factor $q^2$ is bounded away from zero, so $K\asymp (c-d)^{2k} + (q-c)^{2k}$. The coincidences $c=d$ and $q=c$ each have robustness $k$, $I_2$ is separably robust with robustness $(k,k)$, and Theorem \ref{twosided} gives $\lambda(I_2) = \frac{1}{2}\left(\frac{1}{k}+\frac{1}{k}\right) = \frac{1}{k}$. Hence, by Theorem \ref{main1}:
    \begin{itemize}
        \item $k=1$: $\lambda(I_2) = 1 > 1/2 = \lambda(I_1)$, and $\Theta^{**} = I_1$. Self-blame is selected by complexity (Example \ref{resex}).
        \item $k=2$: $\lambda(I_2) = 1/2 = \lambda(I_1)$, and $\Theta^{**} = \Theta^*$. The posterior concentrates on the union of the two narratives as they have the same learning coefficient.
        \item $k\geq 3$: $\lambda(I_2) = 1/k < 1/2 = \lambda(I_1)$, and $\Theta^{**} = I_2$. Circumstance-blame is selected by robustness, even though it requires two coincidences and self-blame requires one.
    \end{itemize}

    The order $k$ measures how insensitive output is to small departures from the threshold. When $k$ is large, circumstances slightly below $d$ and ability slightly above $c$ produce output that is nearly zero, so the coincidences $c=d$ and $q=c$ need to hold only approximately for the data to be nearly perfectly explained.
\end{runningex}

Definition \ref{robustdef} measures how fast the fit worsens when a single coincidence is violated, but violating a single coincidence need not worsen the fit, as when output is $(c-d)(q-c)$. If only one of $c=d$ and $q=c$ fails, the other factor is still zero, so output stays at zero and the fit stays perfect. Supplemental Appendix \ref{supp:joint} extends robustness and Theorem \ref{twosided} to narratives in which the fit worsens only when several coincidences fail simultaneously.

\section{Singular Equilibrium and Endogenous Convergence}\label{equilibrium}

In this section, we define \textit{singular equilibrium}, which is a refinement of Berk-Nash equilibrium \citep{bne} in the misspecified case, and of self-confirming equilibrium \citep{sce} in the correctly specified case. \citet{fls} show that all limits of learning are ``uniform" Berk-Nash equilibria. Our third main result refines \citet{fls} in the setting with a continuum of parameters by showing that all limits of learning are moreover \textit{singular equilibria}. In parallel with \citet{fls}, our fourth main result characterizes the actions to which learning converges from every prior that favors their concentration set.

\subsection{Singular Equilibrium}

Here, we define singular equilibrium, which is an action that is a best reply to a belief supported on the concentration set $\Theta^{**}(a)\subseteq \Theta^*(a)$. We say that an action $a$ is \textit{myopically optimal} given a belief $\mu\in \Delta(\Theta)$ if $a\in BR(\mu) := \arg\max_{a'\in A}\bar u(a',\mu)$ where $\bar u(a,\mu) := \int_\Theta \int_Y u(a,y) dQ^a_\t(y)d\mu(\t)$ is the flow payoff. We write $\d_\t$ for the point mass at $\t$ and $\bar u(a,\t) := \bar u(a,\d_\t).$

\begin{defn}[Singular equilibrium]\label{se}
\text{}
    \begin{enumerate}[(i)]
        \item An action $a\in A$ is a \textbf{singular equilibrium} (SE) if for some belief $\mu\in \Delta(\Theta^{**}(a))$, $a$ is myopically optimal given $\mu$, i.e., $a\in BR(\mu)$.
        \item An action $a\in A$ is a \textbf{strict uniform singular equilibrium} (SUSE) if $BR(\d_\t) = \{a\}$ for every $\t\in \Theta^{**}(a).$ 
    \end{enumerate}
    The sets of equilibria nest as $\mathrm{SUSE}\subseteq\mathrm{SE}$. 
\end{defn}

\begin{remark}[Comparison to other solution concepts]
    Singular equilibrium is a refinement of Berk-Nash equilibrium in misspecified environments, and of self-confirming equilibrium in correctly specified environments. SE requires the belief to be supported on the concentration set $\Theta^{**}(a)$, while Berk-Nash equilibria and self-confirming equilibria only require support on the KL-minimizing set $\Theta^*(a)$. The refinement restricts the agent's beliefs about the actions she does not play (off-path beliefs). The parameters in $\Theta^*(a)$ induce the same distribution of outcomes under $a$ but can disagree about other actions. Since she plays only $a$, she never observes outcomes under other actions, so Berk-Nash and self-confirming equilibria allow any belief supported on $\Theta^*(a)$. Singular equilibrium allows only beliefs over parameters that learning selects by complexity and robustness (Example \ref{seex}). Strict uniform singular equilibrium is a coarsening of ``uniformly strict Berk-Nash equilibrium" of \citet{fls}, i.e., every uniformly strict Berk-Nash equilibrium is a strict uniform singular equilibrium. Under weak identification, the uniformity of SUSE is over the \textit{parameters} that induce the single KL-minimizing subjective distribution and minimize the LLC.
\end{remark}

\subsection{Necessity}\label{endo}

We now present our third main result, which restricts the possible limit actions of the endogenous learning framework. A limit action is an action that has a positive probability of being the only action played after a point in time, along some optimal policy. 

\begin{defn}[Limit action]\label{limitaction}
    An action $a\in A$ is a \textbf{limit action} if there exists an optimal belief-based policy $\pi: \Delta(\Theta) \to A$ with $\P^{\mu_0}_\pi(\exists T: a_t = a\text{ for all }t>T) > 0$.
\end{defn}

\begin{theorem}\label{main2}
If $a\in A$ is a limit action, then it is a singular equilibrium.
\end{theorem}

\begin{proof}
    See Appendix \ref{app:main2}.
\end{proof}

Theorem \ref{main2} states that, if an agent in the endogenous learning framework converges with positive probability to an action $a$, then that action must be a singular equilibrium. Even though all KL-minimizing parameters are associated with the same subjective distribution under weak identification, the support of the belief must be constrained to the subset of parameters that minimize the learning coefficient $\lambda(\t\mid a).$ Thus, only actions that are best replies to beliefs supported on the least complex and most robust narratives can be sustained in the long run.

\begin{intuition}
    Theorem \ref{main2} says that any action the agent settles on must be justified by a belief supported on its own concentration set. There are three steps.

    (1) \textbf{Berk's theorem.} Once the agent plays $a$ forever, every observation $y$ is drawn from $P^a_*$, and Bayes' rule weights each parameter $\t$ by $Q^a_\t(y)$. By Berk's theorem, the posterior eventually places zero mass on every parameter whose subjective distribution $Q^a_\t$ does not minimize the KL divergence from $P^a_*$, so beliefs place all mass on $\Theta^*(a)$. Since she never plays other actions after some time $T$, a parameter's distributions under them can be arbitrarily wrong without affecting her posterior. This is the self-confirming logic of \citet{sce} and \citet{bne}.

    (2) \textbf{Refinement.} Theorem \ref{main1} refines almost sure concentration to the parameters with the smallest learning coefficient. Two parameters that are observationally identical after action $a$ therefore receive different posterior weight, because one lies on a narrative that requires fewer coincidences, or coincidences that need only hold approximately.

    (3) \textbf{Optimality.} At every date an optimal policy chooses an action whose Bellman loss (Lemma \ref{dp}), the value lost by choosing it now instead of an optimal action, is zero. Bellman loss is continuous in the belief, so $a$ remains dynamically optimal at the limiting belief $\mu$. Weak identification then removes the option value of experimenting with $a$. That is, on $\Theta^*(a)$ the action is uninformative, so the continuation value of playing $a$ is just its flow payoff, and dynamic optimality implies myopic optimality.
\end{intuition}

We now provide an endogenous version of Examples \ref{staticexo} and \ref{resex} and exhibit a Berk-Nash equilibrium that is not a singular equilibrium (a self-confirming equilibrium in this case, as the minimum KL divergence is zero).

\begin{runningex}[Singular equilibrium]\label{seex}
     We add a second action to Example \ref{staticexo}. Let $A=\{a,b\}$, where $a$ is \textit{persist} (stay in the current environment) and $b$ is \textit{pivot} (move to a new one at a cost $\kappa>0$). The agent models each environment as in Example \ref{staticexo}, with its own threshold:
    \[
    Q^a_{(q,c)}=N\left(q(c-d),1\right)\otimes N\left(q(q-c),1\right),\quad
    Q^b_{(q,c)}=N\left(q(c-d'),1\right)\otimes N\left(q(q-c),1\right)
    \]
    for $0<d'<d<1$. The agent subjectively believes the new environment is less demanding ($d'<d$), in that a lower level of circumstances suffices for her ability to yield output. Suppose that both environments are objectively unfavorable, $P^a_*=P^b_*=N(0,1)\otimes N(0,1)$. Utility is output net of the pivot cost, $u(a,y)=y_1+y_2$ and $u(b,y)=y_1+y_2-\kappa$, and we take $\kappa<d'(d-d')$. Then $\bar u(a,(q,c))=q(q-d)$ and $\bar u(b,(q,c))=q(q-d')-\kappa.$ Both KL divergences have the form of Example \ref{staticexo}, so by Example \ref{resex} (with $d'$ in place of $d$ for action $b$),
    \[
    \Theta^*(a)=I_1\cup\{(d,d)\},\quad \Theta^{**}(a)=I_1,\quad\quad
    \Theta^*(b)=I_1\cup\{(d',d')\},\quad \Theta^{**}(b)=I_1,
    \]
    where $I_1=\{q=0\}$ is self-blame. Each action has its own circumstance-blame narrative, $(d,d)$ for $a$ and $(d',d')$ for $b$.

    Under self-blame, $\bar u(a,(0,c))=0>-\kappa=\bar u(b,(0,c))$ for every $c$. That is, with no ability, pivoting only incurs a cost. Under $b$'s circumstance-blame narrative $(d',d')$, we have $\bar u(b,(d',d'))=-\kappa > d'(d'-d)=\bar u(a,(d',d'))$ by the choice of $\kappa$. The belief $\d_{(d',d')}$ is supported on $\Theta^*(b)$ and makes $b$ a best reply, so $b$ is a Berk-Nash equilibrium \citep{bne}. The agent can rationalize a costly, unproductive pivot by attributing her poor outcome to an overly demanding environment. But every belief on $\Theta^{**}(b)=I_1$ makes $a$ strictly better, so $b$ is not a singular equilibrium.
\end{runningex}

Berk-Nash and self-confirming equilibrium only require that the agent's action be justified by \textit{some} narrative. In the example, pivoting is justified by the complex circumstance-blame narrative. By contrast, singular equilibrium requires that the action be justified by the least complex narrative, so persisting is the unique singular equilibrium justified by the simple self-blame narrative. The two narratives induce the same distribution of outcomes under pivoting and different distributions under persisting, so we can see how singular equilibrium restricts off-path beliefs and thus leads to a different prediction.

\subsection{Sufficiency and Stability}

Theorem \ref{main2} says that \textit{if} the learning process converges on an action, then it is a singular equilibrium. We now establish a converse and show that an action is a strict uniform singular equilibrium if and only if the learning process converges to it from every prior concentrated near its concentration set, under every optimal policy, with high probability. We call this convergence \textit{uniform stability}, after the notion of the same name in \citet{fls}. Stability in \citet{fls} quantifies over all priors that place probability at least $1-\e$ on an $\e$-neighborhood of the KL minimizers. Ours quantifies over all priors in an $\e$-basin of $\Theta^{**}(a)$: smooth tilts of $\mu_0$ whose weight away from $\Theta^{**}(a)$ is at most a fraction $\e$ of its weight near it, together with the point masses near it.

\begin{defn}[$\e$-basin of $\Theta^{**}(a)$]\label{basindef}
    A \textit{smooth tilt} of the prior $\mu_0$ is a prior $\phi\mu_0$ with $\phi:\Theta\to (0,\infty)$ of class $C^\infty$. Let $\PP$ be the set of smooth tilts of $\mu_0$. For $r>0$, let $N_r := \{\t\in \Theta : \dist(\t,\Theta^{**}(a)) < r\}$. For $0<r<R$ and $\e>0$, the \textbf{$\e$-basin of $\Theta^{**}(a)$} is defined as
    \[
    \B_{r,R}(\e; a):=\left\{\phi\mu_0\in\PP: \sup_{\t\notin N_R}\phi(\t)\leq \e\inf_{\t\in N_r}\phi(\t)\right\} \cup \{\d_\t:\t\in N_R\},
    \]
    with the convention that the supremum over the empty set is $0$.
\end{defn}

The priors comprising the first set are smooth priors that favor $\Theta^{**}(a)$: the smoothing cannot shift more than a factor $\e$ of weight from $\mu_0$ far from $\Theta^{**}(a)$. The second set consists of the priors that are point masses sufficiently near $\Theta^{**}(a).$\footnote{The point masses in $\B_{r,R}(\e;a)$ do not satisfy Assumption \ref{fc1}(iii), but the argument does not need it. From $\d_\t$, the posterior equals $\d_\t$ at every date, so a belief-based policy $\pi$ plays the constant action $\pi(\d_\t)$, and we say $\pi$ is optimal for $\d_\t$ if $\pi(\d_\t) \in BR(\d_\t)$.}

\begin{defn}[Uniformly stable]\label{behnotions}
An action $a$ is \textbf{uniformly stable} if for every $\kappa\in(0,1)$ there are $0<r<R$ and $\e>0$ such that, for every $\nu \in \B_{r,R}(\e;a)$, $\P^\nu_\pi(\exists T : a_t = a\text{ for all }t>T) > 1-\kappa$ for every policy $\pi$ optimal for $\nu$.
\end{defn}

\begin{theorem}\label{main3}
An action $a$ is uniformly stable if and only if $a$ is a strict uniform singular equilibrium.
\end{theorem}
\begin{proof}
    See Appendix \ref{app:main3}.
\end{proof}

Theorem \ref{main3} establishes an equivalence between uniform stability, which is a dynamic property, and strict uniform singular equilibrium, which is a static concept, so the static use of SUSE is supported by learning foundations. For necessity, if some $\t_0 \in \Theta^{**}(a)$ makes another action a best reply, the policy that plays that action forever is optimal for $\d_{\t_0}$. For sufficiency, playing $a$ repeatedly concentrates beliefs on $\Theta^{**}(a)$ by Theorem \ref{main1}, and since $a$ is the unique best reply to beliefs near $\Theta^{**}(a)$, the agent continues to play $a$, which leads to further concentration of beliefs. The construction of the $\e$-basin guarantees that the agent enters this ``loop."

\begin{runningex}[Uniform stability]\label{stabex}
    In Example \ref{seex}, persisting is a strict uniform singular equilibrium. On $\Theta^{**}(a)=I_1$, $\bar u(a,(0,c))=0>-\kappa=\bar u(b,(0,c))$ for every $c$. It is not a uniformly strict Berk-Nash equilibrium in the sense of \citet{fls}. At $(d,d)\in\Theta^*(a)$, $\bar u(b,(d,d))=d(d-d')-\kappa>0=\bar u(a,(d,d))$ by the choice $\kappa<d'(d-d')$. So the stability result of \citet{fls} does not apply to persisting, while Theorem \ref{main3} implies that it is uniformly stable.
\end{runningex}

\subsection{Existence and Cycles}

In general, singular equilibria need not exist, which causes behavior to \textit{cycle} in both the misspecified and correctly specified settings. \citet{nyarko} shows that behavior may cycle in misspecified environments because the beliefs that best fit the data generated by each action make another action optimal. Our cycles instead arise because learning selects among the parameters that induce the unique best-fitting subjective distribution. Surprisingly, this logic persists in correctly specified settings, which appears to be contrary to \citet{ek}. We explain in Remark \ref{ekremark} how our results coexist.

\begin{runningex}[Cycle]\label{nonexist}
    We modify the pivot action of Example \ref{seex} so that no singular equilibrium exists. Keep $a$, persist, exactly as before, so that $\Theta^*(a)=I_1\cup\{(d,d)\}$ and $\Theta^{**}(a)=I_1=\{q=0\}$. Now let $b$ be a pivot to a role in which the agent believes her ability is counterproductive, e.g., she believes the better a scientist she is, the worse a manager she will be.
    \[
    Q^b_{(q,c)}=N\left(\g(d-q),1\right)\otimes N\left(c-d',1\right),\quad P^b_*=N(0,1)\otimes N(0,1),
    \]
    with $\g>1$ and $d'\in(0,1)$. The first coordinate is output net of the interference from her ability, at rate $\g$. The second is her circumstances relative to the new role's threshold $d'$. Objectively, the new role is also unfavorable. Utility is output net of the pivot cost, so $\bar u(a,(q,c))=q(q-d)$ and $\bar u(b,(q,c))=\g(d-q)+(c-d')-\kappa,$ and we take $0<\kappa<\g d-d'$, assuming $\g d>d'$. Under $b$, the KL divergence is minimized only at $(d,d')$, so $\Theta^*(b)=\{(d,d')\}$. The environment is correctly specified if $d'=d$, since $(d,d)$ then explains the data of both actions, and misspecified if $d'\neq d$.

    The following argument applies to both regimes. On $\Theta^{**}(a)=I_1$, persisting yields $0$ utility, while $\bar u(b,(0,c))=\g d-d'-\kappa+c \geq \g d-d'-\kappa > 0$ for every $c\in [0,1]$. A belief on $I_1$ means that the agent believes she has no ability as a scientist, so the role as a manager looks strictly better. So $a$ is not a best reply to any belief on $\Theta^{**}(a)$. Under $b$, $\bar u(b,(d,d'))=-\kappa<0=\bar u(a,(d,d'))$: if her ability is exactly $d$ and her circumstances exactly at the new threshold, the new role only incurs a cost. So $b$ is not a best reply to any belief on $\Theta^{**}(b)$ either. No singular equilibrium exists, and by Theorem \ref{main2} no action is a limit action under any prior. She persists until self-blame takes hold, pivots to $b$, learns $(d,d')$, returns to $a$, and self-blame takes hold again. In the correctly specified regime, the cycle persists even though the agent learns the truth $(d,d)$ every time she pivots. Under $a$, the truth is circumstance-blame, which requires more coincidences than self-blame.
\end{runningex}

\begin{remark}\label{ekremark}
    \citet{ek} consider an agent whose objective distributions are $P^a_*=Q^a_{\t_0}$ for a true parameter $\t_0$ drawn from her prior $\mu_0$, and show that her beliefs converge for $\mu_0$-almost every $\t_0$. Our setting has no true parameter, since the objective distributions are fixed, but we can extend it to theirs by letting each $\t_0\in\Theta$ define the correctly specified environment with objective distributions $P^a_*=Q^a_{\t_0}$ for every $a$. By Proposition \ref{existence}(a) below, in the environment with true parameter $\t_0$, a singular equilibrium fails to exist only if some $a$ and some $\t\in\Theta^*(a)$ satisfy $\lambda(\t\mid a)<\lambda(\t_0\mid a)$. In Example \ref{nonexist} with $d'=d$, the set of such $\t_0$ is $\{(d,d)\}$, which is contained in the null set of \citet{ek}, since the posterior does not converge at $(d,d)$. Thus, analyzing the nonexistence of singular equilibrium allows us to identify true parameters at which behavior cycles, about which the result of \citet{ek} is silent.
\end{remark}

We now give two sufficient conditions for the existence of singular equilibria, the latter of which relies on Tarski's fixed point theorem.

\begin{proposition}[Existence]\label{existence}
\text{}
    \begin{enumerate}[(a)]
        \item If $\bigcap_{a\in A}\Theta^{**}(a)\neq \emptyset$, a singular equilibrium exists.
        \item Let $A = \{a_1<\cdots <a_n\}\subseteq \R$ and let $\Theta$ be partially ordered. If (i) $\bar u$ has increasing differences, i.e., for $b'>b$ and $\t'\geq \t$, $\bar u(b',\t') - \bar u(b', \t)\geq \bar u(b,\t') - \bar u(b,\t)$, and (ii) there is a nondecreasing selection $\t(a) \in \Theta^{**}(a)$, then a singular equilibrium exists.
    \end{enumerate}
\end{proposition}

\begin{proof}
    See Appendix \ref{app:existence}.
\end{proof}

Part (a) requires that there be a parameter that is contained in every concentration set. That is, there exists a narrative that survives learning regardless of the action taken. Part (b) considers the map that sends each action to the largest best reply to the point mass at $\t(a)$. Under (ii), higher actions select higher parameters, and under (i), higher parameters favor higher actions, so this map is nondecreasing and thus has a fixed point.

\section{Applications}

We apply our results to four settings. (1) With unobserved heterogeneity, an agent who allows for more types than the population has learns the number of types if and only if types can differ in at least three outcomes. (2) An underconfident agent, in the style of \citet{hks}, who is also \textit{hesitant} reverts to a default action if she must make three or more decisions. (3) An arbitrarily small misspecification can change which narrative the agent comes to believe. (4) Singular equilibrium summarizes the behavior of an agent who uses neural networks to make predictions.

\subsection{Unobserved Heterogeneity}

Since \citet{heckmansinger}, economists have often modeled unobserved heterogeneity with finitely many types: the population consists of types in unknown shares, and each type has its own skills or preferences. For instance, types represent the skills of individuals in each occupation and their tastes for school and home in models of career choice \citep{keanewolpin}, the severity of patients' illnesses in models of pharmaceutical demand \citep{crawfordshum}, the preferences of households in the evaluation of a school subsidy \citep{toddwolpin}, and the unobserved earning ability of workers in matched employer-employee data \citep{blm}. The number of types is a choice of the modeler, and a model may allow more types than the population has. The parameters of such a model are not identified. For instance, a model with two types induces the distribution of a single type either when one type is absent, whatever the skills of the absent type, or when the two types are identical, whatever their shares.

Theorem \ref{main1} tells us which of these narratives an agent comes to believe when she models more types than the population has, whether she is an employer who thinks of workers in terms of types or an economist who estimates a model with types by Bayesian methods. In the following example, she learns the number of types if and only if types can differ in at least three outcomes.

\begin{example}[Too many types]\label{mixex}
    Consider the exogenous setting. An agent models the population as consisting of $m$ types in shares $\pi=(\pi_1,\hdots,\pi_m)$, where type $j$ has mean output $\t_j\in\R^d$, so $\t=(\pi,\t_1,\hdots,\t_m)\in\Theta=\Delta^{m-1}\times\Theta_*^m$ with $\Theta_*=[-1,1]^d$,\footnote{We parametrize the simplex $\Delta^{m-1}$ of shares by $(\pi_1,\hdots,\pi_{m-1})$, so $\Theta\subseteq\R^{m-1+md}$.} and
    \[
    Q_\t=\sum_{j=1}^m\pi_jN(\t_j,I_d),\quad P_*=\sum_{i=1}^n p_iN(\t^*_i,I_d),
    \]
    where $n<m$, each $p_i>0$, and $\t^*_1,\hdots,\t^*_n$ are distinct points in the interior of $\Theta_*$. This is the specification of unobserved types in \citet{keanewolpin}, abstracting from choices.\footnote{We truncate $N(\t,I_d)$ to a bounded box $Y\subseteq\R^d$, so $\t$ is the mean output before truncation. Without truncation, Assumption \ref{fc2}(iii) fails at absent types, where a slightly negative share can make $Q_\t$ negative for large $y$. We expect that Theorem \ref{main1} still holds, since in a correctly specified environment its proof can run on the likelihood ratio $Q_\t/P_*$, which remains analytic, instead of its logarithm, but we do not prove it here.} Objectively, there are $n$ types, so the agent allows for $m-n$ types too many.

    The environment is correctly specified, and $Q_\t=P_*$ exactly when every type of the agent is either absent or has the mean output of an objective type, and the shares of the types with the mean output of objective type $i$ sum to $p_i$. For a set $A\subseteq\{1,\hdots,m\}$ of absent types and a surjection $\sigma:\{1,\hdots,m\}\setminus A\to\{1,\hdots,n\}$ that assigns each present type to an objective type, let
    \[
    I(A,\sigma):=\left\{\t: \pi_j=0\text{ for }j\in A, \pi_j>0\text{ and }\t_j=\t^*_{\sigma(j)}\text{ for }j\notin A, \sum_{j\in\sigma^{-1}(i)}\pi_j=p_i\text{ for all }i\right\}.
    \]
    The narratives $I(A,\sigma)$ partition $\Theta^*$.\footnote{When $d=1$, the learning coefficient is smaller at the points of $I(A,\sigma)$ at which an absent type has the mean output of an objective type. Away from these points, $I(A,\sigma)$ is a union of narratives with the same learning coefficient, and we refer to $I(A,\sigma)$ as a narrative with that learning coefficient.} The narrative $I(A,\sigma)$ says that the types in $A$ are absent and that each group $\sigma^{-1}(i)$ of remaining types consists of copies of objective type $i$, and does not restrict the mean outputs of the absent types or the shares within each group. Let $a:=|A|$ be the number of absent types and $e:=m-n-a$ the number of \textit{duplicates}, the copies beyond one in each group. The shares and mean outputs of the $n$ groups require $n-1+nd$ coincidences. Each absent type $j$ adds the coincidence $\pi_j=0$, and each duplicate $j$ adds the $d$ coincidences $\t_j=\t^*_{\sigma(j)}$, so
    \[
    c(I(A,\sigma))=n-1+nd+a+ed.
    \]

    We analyze robustness. For group $i$, let $W_i:=\sum_{j\in\sigma^{-1}(i)}\pi_j$ be its share, $\mu_i:=\sum_{j\in\sigma^{-1}(i)}\pi_j(\t_j-\t^*_i)$ its first moment around $\t^*_i$, and $s_i:=\sum_{j\in\sigma^{-1}(i)}\pi_j\|\t_j-\bar\t_i\|_2^2$ its spread around its mean $\bar\t_i:=W_i^{-1}\sum_{j\in\sigma^{-1}(i)}\pi_j\t_j$. Near a point of $I(A,\sigma)$ at which the absent types have mean outputs different from $\t^*_1,\hdots,\t^*_n$, we can compute
    \[
    K\asymp\sum_{j\in A}\pi_j^2+\sum_{i=1}^n\left[(W_i-p_i)^2+\|\mu_i\|_2^2+s_i^2\right].
    \]
    The third and higher moments of each group are of smaller order than $s_i$. We have coincidences $\pi_j=0$ for $j\in A$, $W_i=p_i$, and $\mu_i=0$, each with robustness $1$. The spread $s_i$ is a positive definite quadratic form in the differences between the mean outputs within group $i$, so $s_i^2$ is of the order of the sum of the fourth powers of these differences, and the $ed$ coincidences that equate mean outputs within groups each have robustness $2$. Hence $I(A,\sigma)$ is separably robust, and Theorem \ref{twosided} gives
    \[
    \lambda(I(A,\sigma))=\frac{n(d+1)-1}{2}+\frac{a}{2}+\frac{ed}{4}.
    \]
    The first term is half the dimension of a model with $n$ types, each absent type adds $1/2$, and each duplicate adds $d/4$. The concentration set thus depends on $d$, and there are three cases:
        \begin{itemize}
        \item $d=1$. A duplicate adds $1/4$, less than an absent type, so $\Theta^{**}$ is the closure of the union of the narratives with $a=0$, in which every extra type is a duplicate. The agent comes to believe that all $m$ types are present and form $n$ groups of identical types, each group with the share and mean output of an objective type. She believes there are more types than there actually are, and is agnostic about the composition of groups and shares within each group.
        \item $d=2$. An absent type and a duplicate each add $1/2$ to the learning coefficient, so every narrative has the same learning coefficient and $\Theta^{**}=\Theta^*$. Theorem \ref{main1} does not select among the narratives.
        \item $d\geq 3$. A duplicate adds $d/4>1/2$, more than an absent type, so $\Theta^{**}$ is the union of the narratives with $a=m-n$, in which every extra type is absent. The agent comes to believe that exactly $n$ types are present, with the objective shares and mean outputs, and that the remaining $m-n$ types are absent. She learns the number of types, and is agnostic only about the mean outputs of the absent types.
    \end{itemize}
    With $d=1$, all narratives have the same complexity, and robustness selects duplicates, as in Corollary \ref{occam}(b). With $d\geq3$, the additional coincidences of duplicates outweigh their robustness. Each coincidence of a duplicate need only hold approximately, but a duplicate requires $d$ of them, one per outcome, so the more outcomes in which types can differ, the more learning favors absent types.
\end{example}

The number of types that an agent learns therefore depends on the number of outcomes in which types can differ. When types can differ only in earnings in one occupation, she comes to believe that every type she allows for is present. When types can differ in three or more outcomes, such as earnings in the three occupations of \citet{keanewolpin}, she comes to believe that the extra types are absent.

\begin{remark}
    \citet{rm} obtain the same conclusions by choosing a Dirichlet prior on the shares with hyperparameter one, the uniform prior. Their conclusions depend on the hyperparameter, while ours hold for every smooth, strictly positive prior.
\end{remark}

\subsection{Underconfidence and Hesitancy}

In this subsection, we study how an underconfident and ``hesitant" agent learns about a correct decision. She is underconfident, as in \citet{hks}, in that she dogmatically believes her ability to be lower than it actually is. We model hesitancy through the agent's uncertainty about the cost of wrong decisions and the option to take a safe default action that requires no effort. When she makes decisions other than her default, she observes the marginal returns to her decisions. 

We illustrate how the number of decisions she must make determines whether the optimal, correct decision or the safe, suboptimal default can persist in singular equilibrium. The narrative that her decisions are correct requires a coincidence for each decision, while the narrative that incorrect decisions cost nothing requires only one. If she makes a single decision, the two narratives are equally complex, robustness selects the first, and she can persist in the optimal decision. If she must make three or more decisions, the narrative that her decisions are correct requires too many coincidences, so she instead believes that incorrect decisions cost nothing, which prescribes the incorrect yet effortless default. Hesitancy thus induces the agent to believe that her decisions do not matter, and underconfidence keeps her ``complacent."

\begin{example}[Underconfidence and hesitancy]\label{hksex}
    An agent makes $d\geq 1$ decisions $e=(e_1,\hdots,e_d)$, e.g., one on each of $d$ projects, from a finite grid $A\subseteq[0,1]^d$ that contains her default $0$. Her parameter is $\t=(s,\phi)\in\Theta=[0,1]\times[0,1]^d$, where $\phi$ is the vector of right decisions (the fundamental) and $s$ is the stakes, how much wrong decisions cost her:
    \[
    Q^e_\t=\begin{cases}
        N\left(\tilde a-sL(e-\phi),1\right)\otimes N\left(s\g(\phi-e),I_d\right) & \text{for }e\neq 0\\
        N\left(\tilde a-sL(\phi),1\right)\otimes N\left(0,I_d\right) & \text{for }e=0
    \end{cases}
    \]
    where $L(x)=\frac{\g}{2}\|x\|_2^2$ with $\g>0$. The first coordinate is her output. The second is the marginal return to her decisions, $-\nabla_e[sL(e-\phi)]$, which she observes for nondefault decisions. Here, $\tilde a$ is her dogmatic belief about ability, as in \citet{hks}. Objectively, the stakes are $s^*=1$, the right decisions are $\phi^*\in A$, and her ability is $a$:
    \[
    P^e_*= \begin{cases}
        N\left(a-L(e-\phi^*),1\right)\otimes N\left(\g(\phi^*-e),I_d\right)  & \text{for }e\neq 0\\
        N\left(a-L(\phi^*),1\right)\otimes N\left(0,I_d\right) & \text{for }e=0
    \end{cases}.
    \]
    Note that the agent correctly models her decisions as only lowering output, except at the correct decision $e = \phi^*$. She is underconfident, $0<\Delta:=a-\tilde a\leq\g$. Decisions other than her default require effort, $u(e,y)=y_1-\kappa(e)$ with $\kappa(0)=0$ and $\kappa(e)=\bar\kappa>0$ for $e\neq 0$, so
    \[
    \bar u(e',\t)=\tilde a-sL(e'-\phi)-\kappa(e'),
    \]
    and we take $\bar\kappa<L(\phi^*)\leq\Delta$. So $\phi^*$ is the optimal decision, and her output at the default, $a-L(\phi^*)$, is at least $\tilde a$. For $e\in A$, let $\omega_e:=\Delta-L(e-\phi^*)$, so that her mean output at $e$ is $\tilde a+\omega_e$. For $e\neq 0$, with $x:=sL(e-\phi)\geq 0$, we have
    \begin{equation}\label{eq:hkskl}
    D_{KL}(P^e_*\|Q^e_\t)=\frac{1}{2}(\omega_e+x)^2+\frac{\g^2}{2}\left\|s(\phi-e)-(\phi^*-e)\right\|_2^2.
    \end{equation}
    At $\phi^*$, the marginal returns are objectively zero, $\omega_{\phi^*}=\Delta$, and the second term is $\frac{\g^2}{2}s^2\|\phi-\phi^*\|_2^2=\g sx$. So the minimum is attained exactly when $x=0$, and
    \[
    K^{\phi^*}(\t)=\Delta x+\frac{1}{2}x^2+\g sx,\quad \Theta^*(\phi^*)=\underbrace{\{s=0\}}_{S_0}\cup\underbrace{\{\phi=\phi^*\}}_{I_{\phi^*}}.
    \]
    Both narratives say that $\phi^*$ yields the optimal output and that the marginal returns are zero. The narrative $S_0$ says that there are zero stakes, so wrong decisions do not lower output, and requires one coincidence. The narrative $I_{\phi^*}$ says that her decisions are exactly right, leaving the stakes unrestricted, and requires $d$ coincidences, one for each decision. The truth $(1,\phi^*)$ lies in $I_{\phi^*}$.

    We analyze robustness at $\phi^*$. Near $(0,\phi_0)\in S_0$ with $\phi_0\neq\phi^*$, $L(\phi^*-\phi)$ is bounded away from zero, so $K^{\phi^*}\asymp s$. The coincidence $s=0$ therefore has robustness $1/2$, and Theorem \ref{twosided} gives $\lambda(S_0\setminus I_{\phi^*}\mid\phi^*)=1$. Near $(s_0,\phi^*)\in I_{\phi^*}$ with $s_0>0$, $K^{\phi^*}\asymp\|\phi-\phi^*\|_2^2$. Each of the $d$ coincidences $\phi_i=\phi^*_i$ therefore has robustness $1$, and $\lambda(I_{\phi^*}\setminus S_0\mid\phi^*)=d/2$. At the crossing $(0,\phi^*)$, $K^{\phi^*}\asymp s\|\phi-\phi^*\|_2^2$, so the near-explanations have volume $\asymp\int_0^r\min\{r^d,(\e/u)^{d/2}\}du$, which is $\asymp\e^{1/2}$ if $d=1$, $\asymp\e\log\frac{1}{\e}$ if $d=2$, and $\asymp\e$ if $d\geq 3$, so the crossing has the smaller of the two learning coefficients, $\min\{1,d/2\}$. Hence:
    \begin{itemize}
        \item $d=1$: $\lambda(I_{\phi^*}\mid\phi^*)=1/2<1=\lambda(S_0\setminus I_{\phi^*}\mid\phi^*)$, and $\Theta^{**}(\phi^*)=I_{\phi^*}$. Robustness selects the narrative that her decision is exactly right.
        \item $d=2$: both narratives have learning coefficient $1$, and $\Theta^{**}(\phi^*)=\Theta^*(\phi^*)$.
        \item $d\geq 3$: $\lambda(I_{\phi^*}\setminus S_0\mid\phi^*)=d/2>1=\lambda(S_0\mid\phi^*)$, and $\Theta^{**}(\phi^*)=S_0$. Complexity selects the narrative that wrong decisions cost nothing, although the true stakes are $s^*=1$.
    \end{itemize}
    As in Example \ref{mixex}, robustness selects with a single decision, and complexity selects with three or more. Each correct decision is a coincidence that needs to hold only approximately, but zero stakes is a single coincidence that explains every decision.

    Next, we consider her decisions. For $\mu\in\Delta(\Theta^*(\phi^*))$, let $t_\mu:=\int_\Theta sd\mu(\t)\in[0,1]$ be her expected stakes. Since every parameter in $\Theta^*(\phi^*)$ has $s=0$ or $\phi=\phi^*$,
    \[
    \bar u(\phi^*,\mu)-\bar u(e',\mu)=t_\mu L(e'-\phi^*)-\left[\bar\kappa-\kappa(e')\right],
    \]
    so $\phi^*$ is a best reply to $\mu$ if and only if the expected loss from switching to the default, $t_\mu L(\phi^*)$, is at least $\bar\kappa$. At $e\notin\{0,\phi^*\}$, the marginal returns are objectively nonzero, and \eqref{eq:hkskl} has a unique minimizer:
    \[
    \Theta^*(e)=\{(1,\hat\phi_e)\},\quad \hat\phi_e=e+\hat\rho_e\frac{\phi^*-e}{\|\phi^*-e\|_2},\quad \frac{1}{2}\|\phi^*-e\|_2<\hat\rho_e<\|\phi^*-e\|_2.
    \]
    Then $\bar u(\phi^*,(1,\hat\phi_e))-\bar u(e,(1,\hat\phi_e))=\frac{\g}{2}\|\phi^*-e\|_2(2\hat\rho_e-\|\phi^*-e\|_2)>0$, so $e$ is not a Berk-Nash equilibrium. At the default, she observes only her output, which is at least $\tilde a$. So $\Theta^*(0)=S_0\cup I_0$ with $I_0 := \{\phi = 0\}$, and for $\mu\in\Delta(\Theta^*(0))$ and $e'\neq 0$, $\bar u(0,\mu)-\bar u(e',\mu)=t_\mu L(e')+\bar\kappa>0$. Thus, for $d\geq 3$:
    \begin{enumerate}[(i)]
        \item The optimal decision $\phi^*$ is a best reply to the truth $\d_{(1,\phi^*)}\in\Delta(\Theta^*(\phi^*))$, since $L(\phi^*)>\bar\kappa$, so it is a Berk-Nash equilibrium. But $\phi^*$ is not a singular equilibrium, as it is not supported by the narrative that there are zero stakes, under which her only best reply is the default.
        \item No other decision $e\neq 0$ is a Berk-Nash equilibrium.
        \item The default is a strict uniform singular equilibrium, since $BR(\d_\t)=\{0\}$ for every $\t\in\Theta^*(0)$.
    \end{enumerate}
    The Berk-Nash equilibria are $0$ and $\phi^*$, and the only singular equilibrium is $0$, where she loses $L(\phi^*)-\bar\kappa$, which is less than her underconfidence. By Theorem \ref{main2}, $\phi^*$ cannot be a limit action. If $d\leq 2$, then $I_{\phi^*}\subseteq\Theta^{**}(\phi^*)$, and $\phi^*$ is also a singular equilibrium.
\end{example}

\begin{remark}
    In \citet{hks}, the fundamental is an external condition that affects her output directly and also determines her optimal action. In Example \ref{hksex}, the fundamental is just the vector of right decisions, so her output depends only on how far her decisions are from it. Once she makes the right decisions, the narrative that there are no stakes explains her data as well as the narrative that her decisions are correct. With three or more decisions, she comes to believe that there are zero stakes and excludes the narrative that her decisions are correct. Her loss is less than her underconfidence, as in \citet{hks}, but she settles for her default rather than a decision near the optimal one.
\end{remark}

\subsection{Sensitivity to Misspecification}

In this subsection, we study whether predictions are robust to small misspecification, that is, whether an agent whose model is nearly correct comes to hold beliefs close to those of an agent whose model is correct. \citet{fii} show that in social learning among agents with heterogeneous types, arbitrarily small misperceptions of the type distribution can lead all agents to assign probability one to an incorrect state, because agents come to rely on others' actions, which depend on the state only through others' beliefs. In our setting, a single agent learns from exogenous data whose distribution always depends on the truth, and small misspecification can still change her long-run beliefs discontinuously, through the robustness of narratives. In the following example, the set of KL minimizers does not depend on the degree of misspecification, but an arbitrarily small misspecification lowers the robustness of a coincidence and changes the concentration set.\footnote{There is an econometric literature on sensitivity to local misspecification, which studies estimators rather than beliefs \citep{ags,armstrongkolesar,bonhommeweidner}.}

\begin{runningex}[Misspecification]\label{misspecex}
    Consider the exogenous setting of Example \ref{robustex} with threshold order $k\in \N$, and add a third signal:
    \begin{gather*}
    Q_{(q,c)} = N\left(q(c-d)^k,1\right)\otimes N\left(q(q-c)^k,1\right)\otimes N\left(-q^2(q-c)^2,1\right)\\
    P_* = N(0,1)\otimes N(0,1)\otimes N(\tau,1),
    \end{gather*}
    on $\Theta=[0,1]^2$ with $d\in (0,1)$ and $\tau\geq 0$. The third signal is a penalty for the mismatch between her ability and her circumstances. Objectively, the penalty has mean $\tau$. The case $\tau=0$ is correctly specified. When $\tau>0$, the agent is rewarded, whereas her model only allows a penalty. Writing $m(q,c)$ for the vector of subjective means, one can compute
    \[
    D_{KL}(P_*\|Q_{(q,c)}) = \frac{1}{2}\tau^2 + \frac{1}{2}\|m(q,c)\|_2^2 + \tau q^2(q-c)^2.
    \]
    The last two terms are nonnegative and vanish exactly when $m(q,c)=0$, so the minimum is $D_0 = \tau^2/2$, and
    \begin{gather*}
    \Theta^* = \underbrace{\{(0,c) : c\in [0,1]\}}_{I_1}\cup \underbrace{\{(d,d)\}}_{I_2}\\
    K(q,c) = \frac{1}{2}q^2\left[(c-d)^{2k}+(q-c)^{2k}\right] + \frac{1}{2}q^4(q-c)^4 + \tau q^2(q-c)^2.
    \end{gather*}
    As in Example \ref{robustex}, the narratives and their complexities do not depend on $\tau$. However, $\tau$ influences the robustness of $I_2$. When $\tau>0$, the coincidence $q=c$ of circumstance-blame becomes less robust, while self-blame is unaffected. Near $(0,c_0)\in I_1$, every term of $K$ carries a factor $q^2$, and the remaining factor is bounded above and away from zero, so $K\asymp q^2$. The coincidence $q=0$ has robustness $1$, and Theorem \ref{twosided} gives $\lambda(I_1) = 1/2$ for every $\tau\geq 0$. Near $(d,d)$, the factor $q^2$ is bounded away from zero, and for $k\geq 2$
    \[
    K\asymp (c-d)^{2k} + (q-c)^4 + \tau(q-c)^2.
    \]
    When $\tau = 0$, the coincidences $c=d$ and $q=c$ have robustness $k$ and $2$, and Theorem \ref{twosided} gives $\lambda(I_2) = \frac{1}{2}\left(\frac{1}{k}+\frac{1}{2}\right)$. When $\tau>0$, the coincidence $q=c$ has robustness $1$ instead because the misspecification adds the term $\tau q^2(q-c)^2$ to the KL divergence, and near $(d,d)$ this term grows faster in $q-c$ than the quartic term that governs the KL divergence when $\tau = 0$. Thus, $\lambda(I_2) = \frac{1}{2}\left(\frac{1}{k}+1\right)$. Hence, by Theorem \ref{main1}, for every $k\geq 3$:
    \begin{itemize}
        \item $\tau = 0$: $\lambda(I_2) = \frac{k+2}{4k}<\frac{1}{2} = \lambda(I_1)$, and $\Theta^{**} = I_2$. Circumstance-blame is selected by robustness, as in Example \ref{robustex}.
        \item $\tau>0$: $\lambda(I_2) = \frac{k+1}{2k}>\frac{1}{2} = \lambda(I_1)$, and $\Theta^{**} = I_1$. Self-blame is selected, however small $\tau$ is.
    \end{itemize}
\end{runningex}

\begin{remark}
    In \citet{fii}, belief updating decouples from the state because agents rely on the actions of others. In Example \ref{misspecex}, a single agent observes data drawn from the objective distribution in every period. The misspecification does not change the set of best-fitting parameters, but changes the robustness of coincidence $q=c$, from the fourth power to the square of the violation, which raises the learning coefficient of circumstance-blame above that of self-blame. When $\tau=0$, she observes what she expects under circumstance-blame, so believing that ability and circumstances are slightly mismatched adds almost no surprise, and $q=c$ need hold only roughly. When $\tau>0$, she observes a reward, which she does not expect under circumstance-blame, so with a slight mismatch, this surprise is compounded and $q=c$ must hold almost exactly.
\end{remark}

\subsection{Neural Networks}

Firms and policymakers increasingly act on the predictions of neural networks, in which different values of the weights, the network's parameters, can produce the same predictions \citep{almost}. Our local learning coefficient extends the one that \citet{lau2025local} define for neural networks, which has been estimated during the training of transformer language models \citep{hoogland2025,wang2025}. By Theorem \ref{main2}, an action that is not a singular equilibrium cannot persist when an agent acts on predictions of a neural network. In the following example, we consider a bank that uses a neural network to predict its profit on a loan from the applicant's characteristics, and chooses between lending to every applicant and paying to lend only to two kinds of applicants. Restricting lending can be a Berk-Nash equilibrium but is never a singular equilibrium.

\begin{example}[Neural networks]\label{nnex}
    A bank predicts its profit $y\in\R$ on a loan from the characteristics $x\in[0,1]^p$ of the applicant, $p\geq 3$, using a neural network with one hidden unit. Its parameter is $\t=(v,w)\in\Theta=[-1,1]\times[-1,1]^p$, and the network predicts $v\tanh(w^\top x)$. The hidden unit transforms the weighted sum $w^\top x$ of the characteristics, and the output weight $v$ scales the result. Under action $a$, the bank lends to every applicant, and $x$ is drawn from a known distribution $G^a$ with positive density on $[0,1]^p$. Under action $b$, the bank pays $\kappa>0$ to lend only to two kinds of applicants, with linearly independent characteristics $\bar x_1$ and $\bar x_2$, and $x$ is drawn from the distribution $G^b$ that places probability $1/2$ on each. Under each action $a'$, the bank observes $(x,y)$, its model draws $x$ from $G^{a'}$ and then $y$ from $N(v\tanh(w^\top x),1)$, and objectively $y$ is drawn from $N(0,1)$. Utility is profit net of the cost of restricting, so
    \[
    \bar u(a,\t)=v\E_{G^a}\left[\tanh(w^\top x)\right],\quad\bar u(b,\t)=\frac{v}{2}\left[\tanh(w^\top\bar x_1)+\tanh(w^\top\bar x_2)\right]-\kappa.
    \]
    The environment is correctly specified, and $K^{a'}(\t)=\frac{1}{2}v^2\E_{G^{a'}}[\tanh(w^\top x)^2]$.

    Under $a$, $K^a$ vanishes exactly when $v=0$ or $w=0$, so $\Theta^*(a)=I_1\cup I_2$ with $I_1:=\{v=0\}$ and $I_2:=\{w=0, v\neq 0\}$. The narrative $I_1$ says that the hidden unit does not affect the prediction and requires one coincidence. The narrative $I_2$ says that the hidden unit does not respond to the characteristics and requires $p$ coincidences. Near $(0,w_0)$ with $w_0\neq 0$, $K^a\asymp v^2$, and near $(v_0,0)$ with $v_0\neq 0$, $K^a\asymp\|w\|_2^2$, so every coincidence has robustness $1$, and $\Theta^{**}(a)=I_1$ by Corollary \ref{occam}(a). Under $b$, $K^b$ vanishes exactly when $v=0$ or $w^\top\bar x_1=w^\top\bar x_2=0$, so $\Theta^*(b)=I_1\cup S$ with $S:=\{w^\top\bar x_1=w^\top\bar x_2=0, v\neq 0\}$. The narrative $S$ says that the hidden unit is zero for both kinds of applicants, regardless of its value for other applicants, and requires two coincidences with robustness $1$. Near points of $I_1$ with $w^\top\bar x_1\neq 0$, $K^b\asymp v^2$, so Theorem \ref{twosided} gives $\lambda(I_1\mid b)=1/2<1=\lambda(S\mid b)$, and $\Theta^{**}(b)=I_1$.

    On $I_1$, the predicted profit is zero under both actions, so $\bar u(a,\t)=0>-\kappa=\bar u(b,\t)$. Hence, $a$ is a strict uniform singular equilibrium, and $b$ is not a singular equilibrium. However, $b$ can be a Berk-Nash equilibrium. Suppose the mean characteristics $\E_{G^a}[x]$ do not lie in the span of $\bar x_1$ and $\bar x_2$, which is possible since $p\geq 3$. Then some $w$ with $w^\top\bar x_1=w^\top\bar x_2=0$ has $\E_{G^a}[\tanh(w^\top x)]\neq 0$. Choose the sign of $v$ so that $\bar u(a,(v,w))<0$, and suppose $\kappa<-\bar u(a,(v,w))$. The belief $\d_{(v,w)}$ is supported on $S\subseteq\Theta^*(b)$, and $\bar u(b,(v,w))=-\kappa>\bar u(a,(v,w))$, so $b$ is a Berk-Nash equilibrium. Under this belief, the bank expects to lose money on the applicants it turns away, and the profits it observes have mean zero, as the belief predicts. By Theorem \ref{main2}, $b$ is not a limit action. If the bank restricts its lending, its posterior concentrates on $I_1$, under which the characteristics of applicants do not affect profit and restricting only incurs a cost. The narrative that justifies restricting requires a coincidence for each kind of applicant.
\end{example}
\section{Concluding Remarks}

Our paper provides a way to identify economic narratives naturally arising from Bayesian updating, and a solution concept, justified by learning foundations, that selects the least complex and most robust narratives. Our machinery may extend beyond learning environments, as the local learning coefficient can be defined at the minimizers of any nonnegative analytic function $f$ and governs every distribution proportional to $e^{-tf}$ as $t\to\infty$, not only Bayesian posteriors. Such distributions include the quasi-posteriors of moment-based estimators \citep{ch2003}, which concentrate on the minimizers of a GMM criterion, and the stationary distributions of logit learning in potential games \citep{blume1993}, which concentrate on the maximizers of the potential as noise vanishes. With a continuum of parameters or strategies, the local learning coefficient would determine which of these points survive. Extending our results to dependent data would also allow us to apply them to the Bayesian estimation of DSGE models, in which distinct parameters can be observationally equivalent \citep{komunjerng2011}. The learning coefficient also suggests a role for information design, where a designer can change which narrative the agent learns, and hence which actions are singular equilibria.

Several extensions remain open. First, one could define and study singular equilibria in games with multiple players and with mixed actions. Second, Supplemental Appendix \ref{supp:competing} only considers the exogenous and correctly specified setting, so one could extend it to the endogenous and misspecified setting. Third, one could endogenize a concern for complexity, in the style of \citet{lanzani2025}. Lastly, Supplemental Appendix \ref{supp:tightest} brackets the tightest set to which the posterior concentrates, but we were unable to characterize it in general. Whether a parameter belongs to it depends on how the fluctuations of the log-likelihood compete with the logarithmic factors from multiplicity, and we conjecture that a characterization requires the cluster set of the normalized empirical process, which, by the compact law of the iterated logarithm, is the unit ball of the reproducing kernel Hilbert space of its covariance.

\pagebreak

\appendix
\renewcommand{\theequation}{A.\arabic{equation}}
\setcounter{equation}{0}
\section{Proofs}\label{app}
\renewcommand{\thelemma}{A.\arabic{lemma}}
\setcounter{lemma}{0}

Throughout this appendix, the outcome space $Y$ is finite, as in the main text. Supplemental Appendix \ref{supp:general} explains why the same arguments apply to general outcome spaces. We use Theorem \ref{main1}, which is proved in Supplemental Appendix \ref{supp:main1}, as a black box.

\subsection{The Local Learning Coefficient}\label{app:lambda}

For $\t\in \Theta^*(a)$ and $r>0$, let $\lambda_r(\t\mid a) := \liminf_{\e\downarrow 0}\log V^a_\t (r,\e) / \log \e$.

\begin{lemma}\label{lambdaprops}
    Fix $a\in A$.
    \begin{enumerate}[(a)]
        \item For every $\t\in\Theta^*(a)$, $r\mapsto\lambda_r(\t\mid a)$ is nonincreasing, so the limit defining $\lambda(\t\mid a)$ exists and equals $\sup_{r>0}\lambda_r(\t\mid a)$. Moreover, $0\leq\lambda(\t\mid a)\leq d$.
        \item $\lambda(\cdot\mid a)$ is lower semicontinuous on $\Theta^*(a)$.
        \item $\Theta^{**}(a)$ is nonempty and compact.
    \end{enumerate}
\end{lemma}

\begin{proof}
    Suppress $a$. For (a), if $r<r'$, then $N_\t(r,\e)\subseteq N_\t(r',\e)$, so $V_\t(r,\e)\leq V_\t(r',\e)\leq\Leb(\Theta)$. Taking log and dividing by $\log \e$ gives $\lambda_r(\t)\geq\lambda_{r'}(\t)\geq 0$. For the upper bound, $K$ is real analytic and vanishes at $\t\in \Theta^*$, so $K(\t') \leq C_\t\|\t'-\t\|_2$ near $\t$. Since $\Theta$ is convex with nonempty interior, there is $c>0$ with $\Leb(\Theta\cap B(\t,s))\geq cs^d$ for all $s\in(0,1]$. Hence $V_\t(r,\e)\geq c(\e/C_\t)^d$ for small $\e$, and $\lambda_r(\t)\leq d$.

    For (b), let $\t\in\Theta^*$ and $r>0$. If $\t'\in\Theta^*\cap B(\t,r/2)$, then $B(\t',r/2)\subseteq B(\t,r)$, so $V_{\t'}(r/2,\e)\leq V_\t(r,\e)$ and $\lambda(\t')\geq\lambda_{r/2}(\t')\geq\lambda_r(\t)$. Hence $\liminf_{\t'\to\t}\lambda(\t')\geq\lambda_r(\t)$ for every $r>0$, and letting $r\downarrow 0$ gives $\liminf_{\t'\to\t}\lambda(\t')\geq\lambda(\t)$.
 
    For (c), $\Theta^*$ is nonempty and compact because $K$ is continuous on the compact set $\Theta$. A lower semicontinuous function with values in $[0,d]$ attains its minimum on a nonempty compact set, so $\Theta^{**}\neq\emptyset$. If $\t_n\in\Theta^{**}$ and $\t_n\to\t$, then $\t\in\Theta^*$ and, by (b), $\underline{\lambda}\leq\lambda(\t)\leq\liminf_n\lambda(\t_n)=\underline{\lambda}$, so $\t\in\Theta^{**}$.
\end{proof}

\begin{proof}[Proof of Proposition \ref{invariance}]
    The local learning coefficient is defined through Lebesgue volumes and does not involve $\mu_0$, so neither $\lambda(\cdot\mid a)$ nor $\Theta^{**}(a)$ depends on the prior. For the second claim, suppress $a$, let $\varphi:\Xi^{(R)}\to\Theta^{(R)}$ be a real analytic isomorphism, and let $\Xi:=\varphi^{-1}(\Theta)$. The KL divergence of the relabeled model is $K\circ\varphi$, so its KL-minimizing set is $\varphi^{-1}(\Theta^*)$. Fix $\xi$ in this set and let $\t:=\varphi(\xi)$. Near $\xi$ and $\t$, the maps $\varphi$ and $\varphi^{-1}$ are $L$-Lipschitz and $J^{-1}\leq|\det D\varphi|\leq J$ for some $L,J\geq 1$, so for all small $r$, $B(\t,r/L)\subseteq\varphi(B(\xi,r))\subseteq B(\t,Lr)$. The near-explanations of $\xi$ at radius $r$ are $\varphi^{-1}(T)$, where $T:=\{\t'\in\Theta \cap \varphi(B(\xi,r)) : K(\t')\leq\e\}$. By the change of variables formula, $\Leb(T)=\int_{\varphi^{-1}(T)}|\det D\varphi(\xi')|d\xi'$, so $J^{-1}\Leb(T) \leq V_\xi(r,\e)\leq J\Leb(T)$. Since $N_\t(r/L,\e)\subseteq T\subseteq N_\t(Lr,\e)$, this gives
    \[
    J^{-1}V_\t(r/L,\e)\leq V_\xi(r,\e)\leq JV_\t(Lr,\e).
    \]
    Taking log, dividing by $\log\e<0$, and letting $\e\downarrow 0$ and then $r\downarrow 0$, Lemma \ref{lambdaprops}(a) gives $\lambda(\xi)=\lambda(\t)$. Hence the relabeled concentration set is $\varphi^{-1}(\Theta^{**})$.
\end{proof}

\subsection{Proof of Theorem \ref{twosided}}\label{app:twosided}

Suppress $a$. Since $\lambda$ is constant on $I$, it suffices to bound $\lambda(\t_0)$ at the points $\t_0$ of Definition \ref{robustdef}. Given coincidences $\phi = (\phi_1,\hdots, \phi_c)$ at $\t_0$, choose an affine map $\psi : \R^d\to \R^{d-c}$ with $\psi(\t_0) = 0$ such that the differentials of $\phi_1,\hdots, \phi_c, \psi_1,\hdots, \psi_{d-c}$ are linearly independent at $\t_0$. By the inverse function theorem, $\chi := (\phi,\psi)$ is a real analytic isomorphism of a neighborhood $O$ of $\t_0$ onto an open set containing $[-\eta_0, \eta_0]^d$ with $J^{-1}\leq |\det D\chi^{-1}| \leq J$ on $[-\eta_0,\eta_0]^d$ for some $J\geq 1$. Write $\sigma(\rho) := \frac{1}{2}\sum_{\ell=1}^c 1/\rho_\ell.$

For the lower bound, let $\bar\rho$ be an upper robustness profile with $\t_0$, $\phi$, and $c_0$ as in Definition \ref{robustdef}(b). For small $r>0$, $B(\t_0,r)\subseteq O$, $\chi(B(\t_0,r))\subseteq[-\eta_0,\eta_0]^d$, and $K\geq c_0\sum_\ell|\phi_\ell|^{2\bar\rho_\ell}$ on $\Theta\cap B(\t_0,r)$. Every $\t'\in N_{\t_0}(r,\e)$ then satisfies $|\phi_\ell(\t')|\leq(\e/c_0)^{1/2\bar\rho_\ell}$ for all $\ell$, so taking $\t'\mapsto\chi(\t')$ gives
\[
V_{\t_0}(r,\e)\leq J(2\eta_0)^{d-c}\prod_{\ell=1}^c 2\left(\frac{\e}{c_0}\right)^{1/2\bar\rho_\ell}=C_1\e^{\sigma(\bar\rho)}.
\]
Dividing $\log V_{\t_0}(r,\e)\leq\log C_1+\sigma(\bar\rho)\log\e$ by $\log\e<0$ and letting $\e\downarrow 0$ gives $\lambda_r(\t_0)\geq\sigma(\bar\rho)$, so $\lambda(I)=\lambda(\t_0)\geq\sigma(\bar\rho)$.

For the upper bound, let $\underline{\rho}$ be a lower robustness profile with $\t_0$, $\phi$, and $C$ as in Definition \ref{robustdef}(a). Fix $r>0$ small enough that $B(\t_0,r)\subseteq O$, that $K\leq C\sum_\ell|\phi_\ell|^{2\underline{\rho}_\ell}$ on $\Theta\cap B(\t_0,r)$, and that Definition \ref{robustdef}(a)(ii) holds, and choose $\eta\in(0,\eta_0]$ with $\chi^{-1}([-\eta,\eta]^d)\subseteq B(\t_0,r)$. For $\d\in(0,1]$ with $\d^{1/\underline{\rho}_\ell}\leq\eta$ for all $\ell$,
\[
N^r_\d(\underline{\rho})\supseteq\chi^{-1}\left(\prod_{\ell=1}^c[-\d^{1/\underline{\rho}_\ell},\d^{1/\underline{\rho}_\ell}]\times[-\eta,\eta]^{d-c}\right),
\]
so $\Leb(N^r_\d(\underline{\rho}))\geq J^{-1}2^d\eta^{d-c}\d^{2\sigma(\underline{\rho})}$. Moreover, $K\leq Cc\d^2$ on $\Theta\cap N^r_\d(\underline{\rho})$, since $|\phi_\ell|^{2\underline{\rho}_\ell}\leq\d^2$ there. By Definition \ref{robustdef}(a)(ii), there exists $s_0>0$ with $\Leb(\Theta\cap N^r_\d(\underline{\rho}))\geq s_0\Leb(N^r_\d(\underline{\rho}))$ for all small $\d$. Taking $\d:=(\e/Cc)^{1/2}$ gives, for all small $\e$,
\[
V_{\t_0}(r,\e)\geq \Leb(\Theta\cap N^r_\d(\underline{\rho}))\geq s_0J^{-1}2^d\eta^{d-c}\left(\frac{\e}{Cc}\right)^{\sigma(\underline{\rho})},
\]
so $\lambda_r(\t_0)\leq\sigma(\underline{\rho})$. Since this holds for all small $r$, Lemma \ref{lambdaprops}(a) gives $\lambda(I)=\lambda(\t_0)\leq\sigma(\underline{\rho})$. The remaining claims follow by substituting constant profiles, and because a separable robustness profile is both a lower and an upper profile.

\subsection{Dynamic Programming and Coupling}\label{app:dp}

For $\mu\in\Delta(\Theta)$ and $a\in A$, let $Q^a_\mu:=\int_\Theta Q^a_\t d\mu(\t)$, and let $\mu_{a,y}$ denote the Bayes update of $\mu$ after $(a,y)$.
 
\begin{lemma}\label{dp}
    There is a bounded and weakly continuous function $W:\Delta(\Theta)\to\R$ such that
    \[
    W(\mu)=\max_{a\in A}B(a,\mu),\quad B(a,\mu):=(1-\b)\bar u(a,\mu)+\b\int_Y W(\mu_{a,y})dQ^a_\mu(y),
    \]
    and the following hold.
    \begin{enumerate}[(a)]
        \item $W(\mu)=\sup_\sigma V(\sigma,\mu)$ for every $\mu$, each $V(\sigma,\cdot)$ is affine with $V(\sigma,\mu)=\int_\Theta V(\sigma,\d_\t) d\mu(\t)$, and there is a belief-based policy that is optimal at every $\mu$.
        \item $B(a,\cdot)$ and the \textit{Bellman loss} $G_a:=W-B(a,\cdot)\geq 0$ are weakly continuous for every $a\in A$.
        \item If a belief-based policy $\pi$ is optimal for $\mu_0$, then $G_{\pi(\mu_t)}(\mu_t)=0$ for every $t\geq 0$, $\P^{\mu_0}_\pi$-almost surely.
    \end{enumerate}
\end{lemma}
\begin{proof}
    See Supplemental Appendix \ref{supp:dp}.
\end{proof}
 
\begin{lemma}[Coupling]\label{coupling}
    Fix $a\in A$, a prior $\mu$, and a belief-based policy $\pi$. On a probability space $(\Omega,\G,\P)$, let $z_1,z_2,\hdots$ be i.i.d. with law $P^a_*$ and, independently, let $w_1,w_2,\hdots$ be i.i.d. with uniform law on $[0,1]$. For $b\neq a$, fix a measurable $\varphi_b$ such that $\varphi_b(w_1)$ has law $P^b_*$. Let $\mu^c_0:=\mu$ and recursively define
    \[
    a_{t+1}:=\pi(\mu^c_t),\quad y_{t+1}:=
    \begin{cases}z_{t+1}&\text{if }a_{t+1}=a\\
    \varphi_{a_{t+1}}(w_{t+1})&\text{otherwise}
    \end{cases},
    \quad \mu^c_{t+1}:=(\mu^c_t)_{a_{t+1},y_{t+1}}.
    \]
    Then $(a_t,y_t)_{t\geq 1}$ has law $\P^\mu_\pi$. For $T\geq 0$, let $\G_T:=\sigma(z_s,w_s:s\leq T)$, and let $\hat\mu^T_s$ be the exogenous posterior that starts from $\mu^c_T$ and updates on $(a,z_{T+1}),\hdots,(a,z_{T+s})$. Then $\mu^c_T$ is $\G_T$-measurable, $(z_{T+s})_{s\geq 1}$ is independent of $\G_T$, and $\mu^c_{T+s}=\hat\mu^T_s$ for all $s\geq 0$ on the event $\{a_t=a\text{ for all }t>T\}$.
\end{lemma}
\begin{proof}
    The action $a_{t+1}$ is a function of $(z_s,w_s)_{s\leq t}$, and $(z_{t+1},w_{t+1})$ is independent of these, so the conditional law of $y_{t+1}$ given the history is $P^{a_{t+1}}_*$. By the uniqueness part of the Ionescu-Tulcea theorem, $(a_t,y_t)_{t\geq1}$ has law $\P^\mu_\pi$. The measurability and independence claims hold by construction. On $\{a_t=a\text{ for all }t>T\}$, we have $y_{T+s}=z_{T+s}$ for all $s\geq 1$, so $\mu^c_{T+s}$ and $\hat\mu^T_s$ are obtained from $\mu^c_T$ by the same Bayes updates.
\end{proof}
 
\begin{lemma}\label{bridge}
    If $\mu\in\Delta(\Theta^*(a))$ and $G_a(\mu)=0$, then $a\in BR(\mu)$.
\end{lemma}
\begin{proof}
    By Assumption \ref{fc2}(i), $Q^a_\t=Q^a_{\t'}$ for all $\t,\t'\in\Theta^*(a)$, so $Q^a_\mu=Q^a_\t$ for any $\t\in\Theta^*(a)$ and $\mu_{a,y}=\mu$ for $Q^a_\mu$-almost every $y$. Hence $B(a,\mu)=(1-\b)\bar u(a,\mu)+\b W(\mu)$, and $G_a(\mu)=0$ with $\b<1$ gives $W(\mu)=\bar u(a,\mu)$. Fix $b\in A$. The posteriors $\mu_{b,y}$ average to $\mu$ under $Q^b_\mu$, and $W$ is a supremum of the affine functions $V(\sigma,\cdot)$ (Lemma \ref{dp}(a)), so $\int_Y W(\mu_{b,y})dQ^b_\mu(y)\geq\sup_\sigma V(\sigma,\mu)=W(\mu)$. Hence $W(\mu)\geq B(b,\mu)\geq(1-\b)\bar u(b,\mu)+\b W(\mu)$, so $\bar u(a,\mu)=W(\mu)\geq\bar u(b,\mu)$.
\end{proof}

\subsection{Proof of Theorem \ref{main2}}\label{app:main2}
 
Let $\pi$ be an optimal belief-based policy with $\P^{\mu_0}_\pi(\exists T:a_t=a\text{ for all }t>T)>0$, and construct the history as in Lemma \ref{coupling} with $\mu=\mu_0$. Then there is $T$ with $\P(C_T)>0$, where $C_T:=\{a_t=a\text{ for all }t>T\}$. For $m\geq 1$, let $U_m:=\{\t\in\Theta:\dist(\t,\Theta^{**}(a))<1/m\}$.
 
The prior $\mu^c_T$ has density proportional to $\t\mapsto\mu_0(\t)\prod_{s\leq T}Q^{a_s}_\t(y_s)$, which is smooth and strictly positive almost surely by Assumptions \ref{fc1} and \ref{fc2}(iii), so the exogenous environment $(\Theta,\{Q^a_\t\}_{\t\in\Theta}, \mu^c_T,P^a_*)$ satisfies Assumptions \ref{fc1} and \ref{fc2}, and its concentration set is $\Theta^{**}(a)$ by Proposition \ref{invariance}. Conditional on $\G_T$, the process $(\hat\mu^T_s)_{s\geq 0}$ is the posterior of this environment under i.i.d. data, so Theorem \ref{main1} gives $\P(\hat\mu^T_s(U_m)\to 1\text{ for all }m\mid\G_T)=1$ almost surely. By Lemma \ref{coupling}, $\mu^c_t(U_m)\to 1$ for all $m$, almost surely on $C_T$. Next, note that, by Lemmas \ref{coupling} and \ref{dp}(c), $G_{a_{t+1}}(\mu^c_t)=0$ for all $t$ almost surely, so $G_a(\mu^c_t)=0$ for all $t\geq T$ almost surely on $C_T$. Finally, fix an outcome in $C_T$ at which both properties hold. Since $\Delta(\Theta)$ is weakly compact, $(\mu^c_t)$ has a weak limit point $\mu$. By the portmanteau theorem, $\mu(\overline{U_m})=1$ for every $m$, and since $\Theta^{**}(a)$ is closed (Lemma \ref{lambdaprops}), $\mu\in\Delta(\Theta^{**}(a))$. By Lemma \ref{dp}(b), $G_a(\mu)=0$, so $a\in BR(\mu)$ by Lemma \ref{bridge}, and $a$ is a singular equilibrium.

\subsection{Proof of Theorem \ref{main3}}\label{app:main3}
 
\begin{lemma}\label{trap}
    Suppose $a$ is a strict uniform singular equilibrium. There exist an open set $N\supseteq\Theta^{**}(a)$ and $\bar\eta\in(0,1/2]$ such that (a) $BR(\d_\t)=\{a\}$ for every $\t\in N$, and (b) for every $\mu\in\Delta(\Theta)$ with $\mu(N)\geq 1-\bar\eta$ and every $b\neq a$, $B(b,\mu)<W(\mu)=B(a,\mu)$.
\end{lemma}
\begin{proof}
    If $A=\{a\}$, take $N=\Theta$. Otherwise, since $\t\mapsto\bar u(b,\t)$ is continuous and $\Theta^{**}(a)$ is compact, $g:=\min_{\t\in\Theta^{**}(a)}\min_{b\neq a}[\bar u(a,\t)-\bar u(b,\t)]>0$. Let $N:=\{\t:\min_{b\neq a}[\bar u(a,\t)-\bar u(b,\t)]>g/2\}$, which gives (a). Since updating does not change $\d_\t$, we have $W(\d_\t)=\max_c\bar u(c,\t)$, and by Lemma \ref{dp}(a), $W(\mu)\leq\int_\Theta W(\d_{\t'})d\mu(\t')$ for every $\mu$. Applying this at each $\mu_{b,y}$ and integrating gives $\int_Y W(\mu_{b,y})dQ^b_\mu(y)\leq\int_\Theta W(\d_{\t'})d\mu(\t')$. Let $\|\bar u\|_\infty:=\max_{c\in A,\t\in\Theta}|\bar u(c,\t)|$. If $\mu(N)\geq 1-\eta$ with $\eta\leq 1/2$, then $\int_\Theta W(\d_{\t'})d\mu(\t')\leq\bar u(a,\mu)+2\|\bar u\|_\infty\eta$ and $\bar u(a,\mu)-\bar u(b,\mu)\geq g/4-2\|\bar u\|_\infty\eta$, so, since $W(\mu)\geq\bar u(a,\mu)$,
    \[
    W(\mu)-B(b,\mu)\geq\bar u(a,\mu)-(1-\b)\bar u(b,\mu)-\b\left(\bar u(a,\mu)+2\|\bar u\|_\infty\eta\right)\geq(1-\b)\frac{g}{4}-2\|\bar u\|_\infty\eta,
    \]
    which is positive for $\eta\leq\bar\eta:=\min\{1/2,(1-\b)g/(16\|\bar u\|_\infty)\}$. Since $W=\max_cB(c,\cdot)$, the maximum is then attained only at $a$.
\end{proof}
 
First, we prove sufficiency. Suppose $a$ is a strict uniform singular equilibrium, and let $N$ and $\bar\eta$ be as in Lemma \ref{trap}. Since $\Theta^{**}(a)$ is compact, there exists $R>0$ with $N_R\subseteq N$. Fix $r\in(0,R)$ and let $V:=\Theta\setminus N_R$. Let $z_1,z_2,\hdots$ be i.i.d. with law $P^a_*$, let $L_t(\t):=\prod_{s\leq t}Q^a_\t(z_s)$, and let $\hat\nu_t$ denote the exogenous posterior of a prior $\nu$ after $(a,z_1),\hdots,(a,z_t)$. For a smooth tilt $\nu=\phi\mu_0$,
\[
\hat\nu_t(V)=\frac{\int_VL_t(\t)\phi(\t)\mu_0(\t)d\t}{\int_\Theta L_t(\t)\phi(\t)\mu_0(\t)d\t}\leq\frac{\sup_V\phi}{\inf_{N_r}\phi} A_t,\quad A_t:=\frac{\hat\mu_{0,t}(V)}{\hat\mu_{0,t}(\overline{N_r})}.
\]
Each $A_t$ is finite, since $L_t$ is bounded above and away from zero on $\Theta$ and $\mu_0(N_r)>0$, and $A_t\to 0$ almost surely by Theorem \ref{main1}. Hence $\sup_tA_t<\infty$ almost surely, and for $\kappa\in(0,1)$ there exists $q<\infty$ with $\P(\sup_tA_t\leq q)\geq 1-\kappa/2$. Set $\e:=\bar\eta/q$. Let $\nu\in\B_{r,R}(\e;a)$ and let $\pi$ be optimal for $\nu$. If $\nu=\d_\t$ with $\t\in N_R\subseteq N$, then $\pi(\d_\t)=a$ by Lemma \ref{trap}(a), and the posterior does not change, so $a$ is played forever. If $\nu=\phi\mu_0$, construct the history as in Lemma \ref{coupling} with $\mu=\nu$, using the same $z_t$, and let $E:=\{\sup_t A_t\leq q\}\cap\{G_{a_{t+1}}(\mu^c_t)=0\text{ for all }t\}$, so $\P(E)\geq 1-\kappa/2$ by Lemma \ref{dp}(c). We show by induction that on $E$, $\mu^c_t=\hat\nu_t$ and $a_{t+1}=a$ for all $t\geq 0$. Clearly $\mu^c_0=\nu=\hat\nu_0$. If $\mu^c_t=\hat\nu_t$, then $\mu^c_t(V)\leq\e A_t\leq\bar\eta$, so $\mu^c_t(N)\geq 1-\bar\eta$, and by Lemma \ref{trap}(b), $a$ is the only action with zero Bellman loss at $\mu^c_t$. Hence $a_{t+1}=a$, $y_{t+1}=z_{t+1}$, and $\mu^c_{t+1}=\hat\nu_{t+1}$. Therefore $\P^\nu_\pi(a_t=a\text{ for all }t\geq 1)\geq\P(E)>1-\kappa$, and $a$ is uniformly stable.
 
Next, we prove necessity. If $a$ is not a strict uniform singular equilibrium, there exist $\t_0\in\Theta^{**}(a)$ and $b\neq a$ with $b\in BR(\d_{\t_0})$. For any $0<r<R$ and $\e>0$, $\d_{\t_0}\in\B_{r,R}(\e;a)$, the constant policy $\pi_b\equiv b$ is optimal for $\d_{\t_0}$, and under it $a$ is never played. Hence, $a$ is not uniformly stable.

\subsection{Proof of Proposition \ref{existence}}\label{app:existence}
 
For (a), let $\t\in\bigcap_{a\in A}\Theta^{**}(a)$ and pick $b\in BR(\d_\t)$. Then $\d_\t\in\Delta(\Theta^{**}(b))$, so $b$ is a singular equilibrium. For (b), define $\phi:A\to A$ by $\phi(a):=\max BR(\d_{\t(a)})$. Let $a<a'$, so $\t(a)\leq\t(a')$, and write $b:=\phi(a)$ and $b':=\phi(a')$. If $b'<b$, then $\bar u(b,\t(a))\geq \bar u(b',\t(a))$ and, by increasing differences, $\bar u(b,\t(a'))-\bar u(b',\t(a'))\geq\bar u(b,\t(a))-\bar u(b',\t(a))\geq 0$, so $b\in BR(\d_{\t(a')})$ and $\phi(a')\geq b>b'=\phi(a')$, a contradiction. Hence $\phi$ is nondecreasing on the finite chain $A$. Let $\bar a:=\max\{a\in A:\phi(a)\geq a\}$, which exists since $\phi(a_1)\geq a_1$. Then $\phi(\phi(\bar a))\geq\phi(\bar a)$, so $\phi(\bar a)\leq\bar a$ by maximality, and $\phi(\bar a)=\bar a$. Thus $\bar a\in BR(\d_{\t(\bar a)})$ with $\t(\bar a)\in\Theta^{**}(\bar a)$, and $\bar a$ is a singular equilibrium.

\begin{singlespace}
\bibliographystyle{aea}
\bibliography{bib_2}
\end{singlespace}

\clearpage
\setcounter{page}{1}
\renewcommand{\thepage}{S\arabic{page}}
\setcounter{section}{0}
\renewcommand{\thesection}{S.\arabic{section}}\renewcommand{\theHsection}{supp.\arabic{section}}
\counterwithout*{equation}{section}
\renewcommand{\theequation}{S\arabic{equation}}
\setcounter{equation}{0}
\renewcommand{\thetheorem}{S.\arabic{theorem}}\setcounter{theorem}{0}
\renewcommand{\theproposition}{S.\arabic{proposition}}\setcounter{proposition}{0}
\renewcommand{\thelemma}{S.\arabic{lemma}}\setcounter{lemma}{0}
\renewcommand{\thecorollary}{S.\arabic{corollary}}\setcounter{corollary}{0}
\renewcommand{\thedefn}{S.\arabic{defn}}\setcounter{defn}{0}
\renewcommand{\theexample}{S.\arabic{example}}\setcounter{example}{0}
\renewcommand{\theremark}{S.\arabic{remark}}\setcounter{remark}{0}
\renewcommand{\theHtheorem}{supp.theorem.\arabic{theorem}}
\renewcommand{\theHproposition}{supp.proposition.\arabic{proposition}}
\renewcommand{\theHlemma}{supp.lemma.\arabic{lemma}}
\renewcommand{\theHcorollary}{supp.corollary.\arabic{corollary}}
\renewcommand{\theHdefn}{supp.defn.\arabic{defn}}
\renewcommand{\theHexample}{supp.example.\arabic{example}}
\renewcommand{\theHremark}{supp.remark.\arabic{remark}}
\renewcommand{\theHequation}{supp.equation.\arabic{equation}}
 
\begin{center}
{\Large Supplemental Appendix to ``Singular Equilibrium and Selection of Narratives"}
 
\medskip
David J. Jin
\end{center}

This supplement contains the proof of Theorem \ref{main1} and the remaining proofs. Section \ref{supp:general} states the conditions for general outcome spaces. Sections \ref{supp:resolution}--\ref{supp:main1} prove Theorem \ref{main1}: Section \ref{supp:resolution} recalls the resolution of singularities, Section \ref{supp:volume} relates the local learning coefficient to the exponents of a resolution, and Section \ref{supp:main1} contains the proof. Section \ref{supp:narr} proves the results on narratives, and Section \ref{supp:dp} proves Lemma \ref{dp}. We provide four extensions: Section \ref{supp:weakid} relaxes weak identification, Section \ref{supp:tightest} brackets the smallest set that attracts the posterior, Section \ref{supp:joint} extends robustness to coincidences that are tested jointly, and Section \ref{supp:competing} studies an agent who switches between two models, as in \citet{ba}.

\section{General Outcome Spaces}\label{supp:general}
 
Let $(Y,\Y)$ be a standard Borel space with a $\sigma$-finite measure $\nu$, and identify each distribution on $Y$ with its density with respect to $\nu$. The utility function $u$ is measurable, $\sup_{a\in A,\t\in\Theta}\int_Y|u(a,y)|dQ^a_\t(y)<\infty$, and $\t\mapsto\bar u(a,\t):=\int_Y u(a,y)dQ^a_\t(y)$ is continuous on $\Theta$ for every $a\in A$. Fix $a\in A$ and suppress it.
 
\begin{defn}[Fundamental conditions]\label{fc}
    The environment $(\Theta,\{Q_\t\}_{\t\in\Theta},\mu_0,P_*)$ satisfies the \textbf{fundamental conditions} if
    \begin{enumerate}[(i)]
        \item (Common support): $\{P_*>0\}=\{Q_\t>0\}$ up to a $\nu$-null set, for every $\t\in\Theta$.
        \item (Parameter space): $\Theta$ is as in Assumption \ref{fc1}(ii).
        \item (Nearly analytic prior): $\mu_0=\mu_1\mu_2$ on $\Theta^{(R)}$, where $\mu_1>0$ is $C^\infty$ and $\mu_2\geq 0$ is real analytic.
        \item (Best fit): $\min_{\t\in\Theta}D_{KL}(P_*\|Q_\t)$ is finite and attained.
        \item Assumption \ref{fc2}(i)--(iv) holds, with expectations taken under $P_*$.
    \end{enumerate}
\end{defn}
 
These adapt the fundamental conditions in \citet{watanabe}. With a finite outcome space, Assumptions \ref{fc1} and \ref{fc2} imply the conditions of Definition \ref{fc} with $\mu_2\equiv 1$. When $\mu_2$ vanishes somewhere, the learning coefficient is defined with $\mu_0$ in place of the Lebesgue measure in Definition \ref{nearexp}. Under Assumption \ref{fc1}(iii), $\mu_0$ is bounded above and away from zero on $\Theta$, so the two definitions give the same coefficient.
 
\begin{lemma}[Regularity]\label{regularity}
    Assume the conditions of Definition \ref{fc}. There are a $\nu$-null set $N\subseteq Y$ and versions $\tilde Q_\t$ of $Q_\t$, $\t\in\Theta^{(R)}$, such that:
    \begin{enumerate}[(a)]
        \item for $y\notin N$ with $P_*(y)>0$, $\t\mapsto\tilde Q_\t(y)$ is strictly positive and real analytic on $\Theta^{(R)}$, and $\tilde Q_\t(y)=0$ for all $\t$ otherwise;
        \item $(y,\t)\mapsto\tilde Q_\t(y)$ is jointly measurable;
        \item $\t\mapsto Q_\t$ is continuous on $\Theta$ in total variation, and $\mu\mapsto\bar u(a,\mu)$ is weakly continuous;
        \item $\tilde Q_\t(y)=\tilde Q_{\t'}(y)$ for all $\t,\t'\in\Theta^*$ and $y\notin N$.
    \end{enumerate}
\end{lemma}
\begin{proof}
    On a polydisc $D(q,r)$ whose closure lies in $\Theta^{(C)}$, the $L^2(P_*)$-valued analytic map $\t\mapsto f(\cdot,\t)$ has an expansion $\sum_\a b_\a(\t-q)^\a$ with $\sum_\a\|b_\a\|_{L^2(P_*)}r^{|\a|}<\infty$. Since $\|b_\a\|_{L^1(P_*)}\leq\|b_\a\|_{L^2(P_*)}$, Tonelli's theorem gives $\sum_\a|b_\a(y)|r^{|\a|}<\infty$ for $P_*$-almost every $y$, so $F_{q,r}(y,\t):=\sum_\a b_\a(y)(\t-q)^\a$ is analytic in $\t$ on $D(q,r)$ and, for each $\t$, equals $f(y,\t)$ for $P_*$-almost every $y$. Countably many polydiscs with rational centers and radii cover $\Theta^{(R)}$. For intersections, the two expansions agree at countably many points of a dense set for almost every $y$, hence everywhere by continuity. This gives a single function $F(y,\t)$, analytic in $\t$ on a neighborhood of $\Theta^{(R)}$ for $y$ outside a $P_*$-null set, with $F(\cdot,\t)=f(\cdot,\t)$ almost everywhere for each $\t$. Enlarge this null set by the $\nu$-null set on which $\{P_*>0\}$ and $\{Q_{\t_0}>0\}$ differ, call the result $N$, and set $\tilde Q_\t(y):=Q_{\t_0}(y)e^{-\operatorname{Re}F(y,\t)}$ if $y\notin N$ and $P_*(y)>0$, and $\tilde Q_\t(y):=0$ otherwise. Parts (a) and (b) follow. For (c), the densities $\tilde Q_{\t_n}$ converge pointwise to $\tilde Q_\t$ when $\t_n\to\t$, so they converge in $L^1(\nu)$ by Scheff\'e's lemma, and $\bar u(a,\mu)=\int_\Theta\bar u(a,\t)d\mu(\t)$ integrates a continuous function of $\t$ over the compact set $\Theta$. For (d), weak identification gives $\tilde Q_\t=\tilde Q_{\t'}$ almost everywhere for $\t,\t'$ in a countable dense subset of $\Theta^*$. Removing the corresponding null sets and using continuity in $\t$ extends the equality to all of $\Theta^*$.
\end{proof}
 
\begin{remark}\label{appAgeneral}
    With the versions of Lemma \ref{regularity}, the proofs in Appendix \ref{app} apply to general outcome spaces. After finitely many observations, the posterior has density $\tilde\mu_1\mu_2$ with $\tilde\mu_1$ smooth and positive, so it satisfies the conditions of Definition \ref{fc}. For almost every history, the likelihood $\t\mapsto\prod_{s\leq t}\tilde Q^{a_s}_\t(y_s)$ is continuous and positive, hence bounded above and away from zero on $\Theta$. The KL-minimizing distributions under $a$ coincide pointwise (part (d)), which is what Lemma \ref{bridge} uses. The flow payoffs $\t\mapsto\bar u(b,\t)$ are continuous by assumption. Lemma \ref{dp} is proved in Section \ref{supp:dp} under the conditions of Definition \ref{fc}. In Lemma \ref{lambdaprops}, the bound $d$ is replaced by the finite bound of Proposition \ref{volumeequiv} when $\mu_2$ vanishes somewhere.
\end{remark}
 
\section{Resolution of Singularities}\label{supp:resolution}
 
\begin{theorem}[{\citealp{hironaka}; \citealp[Theorem~2.3]{watanabe}}]\label{ht}
    Let $f$ be a nonconstant real analytic function on a neighborhood of $0\in\R^d$ with $f(0)=0$. There are an open set $W\ni 0$, a $d$-dimensional real analytic manifold $\U$, and a proper real analytic map $g:\U\to W$ that is an analytic isomorphism from $\U\setminus\{f\circ g=0\}$ onto $W\setminus\{f=0\}$, such that every point of $\{f\circ g=0\}$ has local coordinates $u$ centered at it in which $f(g(u))=Su^k$ and $|\det Dg(u)|=b(u)|u^h|$, where $S\in\{-1,1\}$, $k,h\in\Z_+^d$, and $b$ is real analytic and nonvanishing.
\end{theorem}
 
\begin{lemma}\label{split}
    If $f_1,\hdots,f_m$ are real analytic near $0\in\R^d$ and $\prod_{i=1}^mf_i(u)=u^c$ for some $c\in\Z_+^d$, then near $0$, $f_i(u)=\eta_i(u)u^{c_i}$ with $c_i\in\Z_+^d$, $\sum_ic_i=c$, and $\eta_i$ real analytic and nonvanishing.
\end{lemma}
\begin{proof}
    The case $m=2$ is \citet[Theorem~2.7]{watanabe}. For $m\geq 3$, apply the case $m=2$ to $f_1$ and $f_2\cdots f_m$, absorb the unit of $f_2\cdots f_m$ into $f_2$, and apply the induction hypothesis.
\end{proof}
 
Fix $a\in A$ and suppress it, let $\e>0$ be as in Assumption \ref{fc2}(iv), and let $\Theta_\e:=\{\t\in\Theta:K(\t)\leq\e\}$. Let $\F'$ be the set of nonconstant functions among $K$, $\e-K$, $\mu_2$, $q_1,\hdots,q_k$.
 
\begin{theorem}[Resolution]\label{partitiontheta2}
    Assume the conditions of Definition \ref{fc}. There are an open set $\Theta^{(R)}_\e$ with $\Theta_\e\subseteq\Theta^{(R)}_\e\subseteq\Theta^{(R)}$, a $d$-dimensional real analytic manifold $\M^{(R)}$, a proper real analytic map $g:\M^{(R)}\to\Theta^{(R)}_\e$ that is an analytic isomorphism off a closed Lebesgue-null set $\L$, and finitely many charts $M_\g$, $\g\in\Gamma$, each identified with a box $[0,b_\g]^d$ with $b_\g\leq 1$, that cover $\M:=g^{-1}(\Theta_\e)$, such that for every $\g\in\Gamma$:
    \begin{enumerate}[(a)]
        \item if $\eta_{q_i,\g}$ for $i\leq k$ and $\eta_{\e-K,\g}$ in (b) are positive, then either $M_\g\cap g^{-1}(\Theta^*)\neq\emptyset$ and $K(g(u))=u^{2k_\g}$ on $M_\g$ for some $k_\g\in(\frac{1}{2}\Z_+)^d\setminus\{0\}$, or $K\circ g\geq\e_\g$ on $M_\g$ for some $\e_\g>0$;
        \item every $F\in\F'\setminus\{K\}$ satisfies $F(g(u))=\eta_{F,\g}(u)u^{c_{F,\g}}$ on $M_\g$, with $c_{F,\g}\in\Z_+^d$ and $\eta_{F,\g}$ real analytic and nonvanishing;
        \item $\mu_0(g(u))|\det Dg(u)|=\psi_\g(u)u^{h_\g}$ on $M_\g$, with $h_\g\in\Z_+^d$ and $\psi_\g$ smooth with $\psi_\g\geq c>0$;
        \item for every measurable $H\geq 0$ on $\Theta_\e$, $\int_{\Theta_\e}H\mu_0d\t\leq\sum_{\g\in\Gamma}\int_{M_\g\cap\M}H(g(u))\psi_\g(u)u^{h_\g}du$, and for every chart $\g$ and measurable $S\subseteq M_\g$ with $g(S)\subseteq\Theta_\e$, $\int_{g(S)}H\mu_0d\t=\int_SH(g(u))\psi_\g(u)u^{h_\g}du$.
    \end{enumerate}
\end{theorem}
\begin{proof}
    The function $K$ is real analytic on $\Theta^{(R)}$, as the expectation of the $L^2(P_*)$-valued analytic map $\t\mapsto f(\cdot,\t)$. The product $\Phi$ of the members of $\F'$ is real analytic, not identically zero, and vanishes on $\Theta^*$. Applying Theorem \ref{ht} to $\Phi$ at each of its zeros and gluing over the compact set $\Theta_\e$ \citep[Remark~2.6]{watanabe} gives an open set $\Theta^{(R)}_\e$ with $\Theta_\e\subseteq \Theta^{(R)}_\e\subseteq \Theta^{(R)}$ and a proper analytic map $g: \M^{(R)}\to \Theta^{(R)}_\e$ that is an isomorphism off $\L:=\{\Phi\circ g=0\}$, near each point of which $\Phi\circ g$ is a signed monomial and $|\det Dg|$ is a nonvanishing analytic function times a monomial. By Lemma \ref{split}, every $F\in\F'$ is then a nonvanishing analytic function times a monomial in the same coordinates. Let $p\in \M$. If $K(g(p))>0$, then $K\circ g$ is positive near $p$. If $K(g(p)) = 0$, $K(g(u))=\eta(u)u^\kappa$ near $p$ with $\kappa \neq 0$, and the nonvanishing $\eta$ has a constant sign $S\in \{-1,1\}$ near $p$. The change of coordinates $u_j\mapsto u_j|\eta(u)|^{1/|\kappa|}$, with $|\kappa|:=\sum_j\kappa_j$, is an analytic isomorphism near $p$ that turns $K\circ g$ into $Su^\kappa$ and preserves the forms of the other functions and of the Jacobian. Cover the compact set $\M$ by finitely many coordinate boxes $[-b,b]^d$ with $b\leq 1$, each around a point $p\in \M$ and small enough that the above holds, and split each box into its $2^d$ closed orthants, each identified with $[0,b]^d$ by reflecting coordinates. Reflections preserve the forms of the other functions and of the Jacobian, and turn $Su^\kappa$ into $s_\g u^\kappa$ with $s_\g\in \{-1,1\}$. For (a), suppose that the units $\eta_{q_i,\g}$, $i\leq k$, and $\eta_{\e-K,\g}$ are positive. Then $q_i\circ g\geq 0$ and $\e-K\circ g\geq 0$ on $M_\g$, so $g(M_\g)\subseteq\Theta_\e$, and $K\circ g\geq 0$ on $M_\g$ because $K\geq 0$ on $\Theta$. If $M_\g$ does not intersect $g^{-1}(\Theta^*)$, then $K\circ g$ is continuous and positive on the compact set $M_\g$. Otherwise, $M_\g$ is an orthant of a box around a point $p$ with $K(g(p)) = 0$, and since $s_\g u^\kappa = K\circ g \geq 0$ on $M_\g$, while $u^\kappa > 0$ on its interior, $s_\g =1.$ This gives (a) with $k_\g := \kappa/2$. For (c), $\mu_0(g(u))|\det Dg(u)|=\mu_1(g(u))\eta_{\mu_2}(u)|b(u)|u^{c_{\mu_2}+h}$, and $\eta_{\mu_2}>0$ because $\mu_2\geq 0$, so the coefficient is smooth and, after shrinking the boxes, bounded below. Part (d) is the change of variables formula, since $g$ is injective off $\L$ and $g(\L)$ is Lebesgue-null.
\end{proof}
 
Let $\Gamma^+$ be the set of charts on which the units $\eta_{q_i,\g}$, $i\leq k$, and $\eta_{\e-K,\g}$ are positive, so that part (a) applies to them, and $\Gamma_*$ the set of charts that intersect $g^{-1}(\Theta^*)$.
 
\begin{lemma}\label{image}
    If $\g\in\Gamma^+$, then $g(M_\g)\subseteq\Theta_\e$. If $\g\notin\Gamma^+$, then $M_\g\cap g^{-1}(\Theta_\e)$ is Lebesgue-null.
\end{lemma}
\begin{proof}
    For $\g\in\Gamma^+$, each $q_i(g(u))$ and $\e-K(g(u))$ is a positive function times a monomial, hence nonnegative on $[0,b_\g]^d$, so $g(u)\in\Theta_\e$. If one of these units is negative, the corresponding function is negative off the null set where its monomial vanishes.
\end{proof}
 
\section{The Learning Coefficient and the Resolution}\label{supp:volume}
 
\begin{lemma}[Monomial integrals]\label{monomialvolume}
    Let $k\in (\frac{1}{2}\Z_+)^n\setminus \{0\}$, $h\in\Z_+^n$, and $b>0$. Let $\lambda_\circ:=\min_{j:k_j>0}(h_j+1)/2k_j$ and let $m_\circ$ be the number of $j$ that attain this minimum.
    \begin{enumerate}[(a)]
        \item There are $0<c\leq C$ and $s_0>0$ such that, for all $s\in(0,s_0]$,
        \[
        cs^{\lambda_\circ}\left(\log\frac{1}{s}\right)^{m_\circ-1}\leq \int_{[0,b]^n}\1\{u^{2k}\leq s\}u^hdu\leq Cs^{\lambda_\circ}\left(\log\frac{1}{s}\right)^{m_\circ-1}.
        \]
        \item There is $C'$ such that $\int_{[0,b]^n}e^{-tu^{2k}}u^hdu\leq C't^{-\lambda_\circ}(\log t)^{m_\circ-1}$ for all $t\geq e$.
    \end{enumerate}
    The same bounds hold, with different constants, if some coordinates range over intervals $[\a_j,\b_j]$ with $0<\a_j<\b_j$, with $\lambda_\circ$ and $m_\circ$ computed over the remaining coordinates.
\end{lemma}
\begin{proof}
    Coordinates with $k_j=0$ contribute constant factors. For the others, substitute $u_j^{2k_j}=b^{2k_j}e^{-y_j}$, so that $u_j^{h_j}du_j$ becomes a constant times $e^{-\lambda_jy_j}dy_j$ with $\lambda_j:=(h_j+1)/2k_j$. The integral in (a) is then a constant times $I(L):=\int_{y\geq 0, \sum_jy_j\geq L}e^{-\sum_j\lambda_jy_j}dy$ with $L:=\log\frac{1}{s}+\text{constant}$. Let $P$ be the set of $j$ with $\lambda_j=\lambda_\circ$, and $\sigma:=\sum_{j\in P}y_j$, whose density under Lebesgue measure on $[0,\infty)^P$ is $\sigma^{m_\circ-1}/(m_\circ-1)!$. Restricting to $y_j\leq 1$ for $j\notin P$ and $\sigma\geq L$ gives $I(L)\geq c\int_L^\infty e^{-\lambda_\circ\sigma}\sigma^{m_\circ-1}d\sigma\geq c'e^{-\lambda_\circ L}L^{m_\circ-1}$. Conversely, $\int_{\sigma\geq x}e^{-\lambda_\circ\sigma}\sigma^{m_\circ-1}d\sigma\leq Ce^{-\lambda_\circ x^+}(1+x^+)^{m_\circ-1}$, so with $x:=L-\sum_{j\notin P}y_j$,
    \[
    I(L)\leq C(1+L)^{m_\circ-1}e^{-\lambda_\circ L}\int_{[0,\infty)^{P^c}}e^{-\sum_{j\notin P}(\lambda_j-\lambda_\circ)y_j}dy<\infty,
    \]
    since $\lambda_j>\lambda_\circ$ for $j\notin P$. This gives (a). For (b), let $F:= u^{2k}$ and $dm := u^h du$ on $[0,b]^n$, so that $\int e^{-tF}dm=\int_0^\infty te^{-ts}m(F\leq s) ds$. Bounding $m(F\leq s)$ by (a) for $s\leq s_0$ and by a constant otherwise, and substituting $s=x/t$, gives (b). Coordinates restricted to $[\a_j,\b_j]$ contribute factors bounded above and away from zero, which only changes the constants.
\end{proof}
 
For a point $u$ of a chart $M_\g$ with $\g\in\Gamma_*\cap\Gamma^+$, let $Z_\g(u):=\{j:u_j=0,k_{\g,j}>0\}$, $\lambda_\g(u):=\min_{j\in Z_\g(u)}(h_{\g,j}+1)/2k_{\g,j}$, and let $m_\g(u)$ be the number of minimizers, with $\lambda_\g(u):=\infty$ and $m_\g(u):=0$ if $Z_\g(u)=\emptyset$. Let $\lambda(u)$ be the minimum of $\lambda_\g(u)$ over the charts $\g\in\Gamma^+$ containing $u$, and $m(u)$ the largest $m_\g(u)$ among the charts attaining it. For $\t\in\Theta^*$, let
\[
\lambda^g(\t):=\min_{u\in g^{-1}(\t)}\lambda(u),\quad m^g(\t):=\max\{m(u):u\in g^{-1}(\t), \lambda(u)=\lambda^g(\t)\}.
\]
Some point of $g^{-1}(\t)$ lies in a chart of $\Gamma^+$, since points of the interior of $\Theta_\e$ near $\t$, off the null set where some $q_i$ vanishes or $g$ fails to be invertible, have preimages in charts of $\Gamma^+$, and a limit point of these preimages lies in $g^{-1}(\t)$ because the charts are compact. Hence $\lambda^g(\t)$ is finite.
 
\begin{proposition}\label{volumeequiv}
    Assume the conditions of Definition \ref{fc} and let $\t\in\Theta^*$. There is $r_0>0$ such that, for every $r\in(0,r_0]$, there are $0<c\leq C$ and $s_0>0$ with
    \[
    c\e^{\lambda^g(\t)}\left(\log\frac{1}{\e}\right)^{m^g(\t)-1}\leq\mu_0\left(\Theta\cap B(\t,r)\cap\{K\leq\e\}\right)\leq C\e^{\lambda^g(\t)}\left(\log\frac{1}{\e}\right)^{m^g(\t)-1}
    \]
    for all $\e\in (0,s_0]$. Thus, $\lambda(\t)=\lambda^g(\t)$, which therefore does not depend on the resolution. The learning coefficient takes finitely many values on $\Theta^*$, all strictly positive.
\end{proposition}
\begin{proof}
    Let $G:=g^{-1}(\t)$, which is compact since $g$ is proper, and take $\e$ below the constant of Theorem \ref{partitiontheta2}, so that $\Theta\cap\{K\leq\e\}\subseteq\Theta_\e$. Every chart containing a point of $G$ lies in $\Gamma_*$.
 
    For the upper bound, for $u_0\in G$ and a chart $\g\in\Gamma^+$ containing $u_0$, choose $\d>0$ such that $u_j\geq u_{0,j}/2$ whenever $u_{0,j}>0$, on $B_\g(u_0):=\{u\in M_\g:\|u-u_0\|_\infty<\d\}$. By Lemma \ref{monomialvolume}(a) and $\psi_\g\leq C$,
    \[
    \int_{B_\g(u_0)}\1\{u^{2k_\g}\leq\e\}\psi_\g(u)u^{h_\g}du\leq C\e^{\lambda_\g(u_0)}\left(\log\frac{1}{\e}\right)^{m_\g(u_0)-1}\leq C'\e^{\lambda^g(\t)}\left(\log\frac{1}{\e}\right)^{m^g(\t)-1},
    \]
    because $\lambda_\g(u_0)\geq\lambda(u_0)\geq\lambda^g(\t)$, and $m_\g(u_0)\leq m^g(\t)$ when equality holds throughout. The charts are closed and finitely many, so $u_0$ has an open neighborhood $W(u_0)$ that intersects only the charts $M_\g$ containing $u_0$, and each only inside $B_\g(u_0)$. Finitely many of these neighborhoods cover $G$. Since $g$ is proper, there is $r_0>0$ such that $g^{-1}(\Theta_\e\cap\bar B(\t,r_0))$ is contained in their union. For $r\leq r_0$, Theorem \ref{partitiontheta2}(d) with $H:=\1\{K\leq\e\}\1_{B(\t,r)}$ bounds $\mu_0(\Theta\cap B(\t,r)\cap\{K\leq\e\})$ by the sum of the finitely many integrals above, since charts outside $\Gamma^+$ do not contribute, by Lemma \ref{image}.
 
    For the lower bound, choose $u_0\in G$ and $\g\in\Gamma^+$ containing $u_0$ with $\lambda_\g(u_0)=\lambda^g(\t)$ and $m_\g(u_0)=m^g(\t)$. Given $r>0$, choose $\d>0$ such that $g$ maps $S:=\{u\in M_\g:\|u-u_0\|_\infty\leq\d\}$ into $B(\t,r)$, and $u_j\geq u_{0,j}/2$ on $S$ whenever $u_{0,j}>0$. By Lemma \ref{image}, $g(S)\subseteq\Theta_\e$, and $K\circ g=u^{2k_\g}$ on $S$, so Theorem \ref{partitiontheta2}(d), $\psi_\g\geq c$, and Lemma \ref{monomialvolume}(a) give
    \[
    \mu_0\left(\Theta\cap B(\t,r)\cap\{K\leq\e\}\right)\geq c\int_S\1\{u^{2k_\g}\leq\e\}u^{h_\g}du\geq c'\e^{\lambda^g(\t)}\left(\log\frac{1}{\e}\right)^{m^g(\t)-1}.
    \]
 
    The two bounds give $\lambda(\t)=\lambda^g(\t)$. The left side of the display does not depend on $g$, so neither does $\lambda^g$. Finally, $\lambda^g$ takes values in the finite set of ratios $(h_{\g,j}+1)/2k_{\g,j}$, which are positive.
\end{proof}
 
\section{Proof of Theorem \ref{main1}}\label{supp:main1}
 
Let $\t_0\in\Theta^*$, $f(y,\t):=\log(Q_{\t_0}(y)/Q_\t(y))$, and $K_t(\t):=t^{-1}\sum_{s\leq t}f(y_s,\t)$, so that the posterior of a set $S$ is $\mu_t(S)=Z_t(S)/Z_t(\Theta)$ with $Z_t(S):=\int_S e^{-tK_t(\t)}\mu_0(\t)d\t$.
 
\begin{theorem}[Standard form]\label{sf}
    Assume the conditions of Definition \ref{fc}. For every $\g\in\Gamma_*\cap \Gamma^+$, there is a function $a_\g(y,u)$, real analytic in $u$ as an $L^2(P_*)$-valued map on a neighborhood of $M_\g$, such that $f(y,g(u))=a_\g(y,u)u^{\lceil k_\g\rceil }$ and $\E_{P_*}[a_\g(y,u)]=u^{\lfloor k_\g\rfloor }$ for $u\in M_\g$. Consequently,
    \[
    K_t(g(u))=u^{2k_\g}-t^{-1/2}u^{\lceil k_\g\rceil }\xi^\g_t(u),\quad\xi^\g_t(u):=t^{-1/2}\sum_{s\leq t}\left(u^{\lfloor k_\g\rfloor }-a_\g(y_s,u)\right).
    \]
\end{theorem}
\begin{proof}
    This is the misspecified version of \citet[Main Theorem~6.1]{watanabe}, with renormalizability. Write $k:=k_\g$. Since $g$ is analytic, $f(\cdot,g(u))=\sum_\vartheta a_\vartheta u^\vartheta$ with $a_\vartheta\in L^2(P_*)$ and $\sum_\vartheta\|a_\vartheta\|_{L^2(P_*)}r^{|\vartheta|}<\infty$ for some $r>b_\g$. Since $\g\in \Gamma_*\cap \Gamma^+$, Theorem \ref{partitiontheta2}(a) gives $K(g(u)) = u^{2k}$ on $M_\g$ with $k\neq 0$, so $K(g(0)) = 0$, and $g(M_\g) \subseteq \Theta_\e$ by Lemma \ref{image}, so $g(0) \in \Theta^*\subseteq V$. Hence, for $u\in M_\g$ near 0, renormalizability applies at $g(u)$ and gives $\|f(\cdot,g(u))\|_{L^2(P_*)}^2\leq K(g(u))/C=u^{2k}/C$. Fix $i$ with $k_i>0$ and write $f(\cdot,g(u))=\sum_n u_i^nG_n(u_{-i})$. Let $n<k_i$ with $G_m \equiv 0$ for all $m<n$. Then $u_i^{-n}f(\cdot,g(u))\to G_n(u_{-i})$ in $L^2(P_*)$ as $u_i\downarrow 0$, while its norm is at most $C^{-1/2}u_i^{k_i-n}u_{-i}^{k_{-i}}\to 0$. Hence $G_n=0$ on an open set, so $G_n\equiv 0$ by the identity theorem. By induction, $G_n\equiv 0$ for every $n<k_i$, i.e., every $n<\lceil k_i\rceil $. Therefore $a_\vartheta=0$ unless $\vartheta\geq \lceil k\rceil $, and $a_\g(\cdot,u):=\sum_{\vartheta\geq \lceil k\rceil }a_\vartheta u^{\vartheta-\lceil k\rceil }$ satisfies $f(\cdot,g(u))=a_\g(\cdot,u)u^{\lceil k \rceil}$. Taking expectations, $u^{2k}=K(g(u))=u^{\lceil k\rceil }\E_{P_*}[a_\g(y,u)]$, and $2k-\lceil k\rceil = \lfloor k\rfloor$, so $\E_{P_*}[a_\g(y,u)]=u^{\lfloor k\rfloor}$ where $u^{\lceil k\rceil}\neq 0$, and everywhere by continuity.
\end{proof}
 
\begin{lemma}\label{lilbound}
    Assume the conditions of Definition \ref{fc} and fix $\e'>0$. There are an almost surely finite random variable $\Lambda$ and an almost surely finite random time $T_0$ such that, for all $t\geq T_0$,
    \[
    \max_{\g\in\Gamma_*\cap \Gamma^+}\|\xi^\g_t\|_\infty\leq\Lambda\sqrt{2\log\log t},\quad \sup_{\t\in\Theta:K(\t)\geq\e'}|\psi_t(\t)|\leq\Lambda\sqrt{2\log\log t},
    \]
    where $\psi_t(\t):=\sqrt{t}(K(\t)-K_t(\t))/\sqrt{K(\t)}$.
\end{lemma}
\begin{proof}
    Fix $\g\in\Gamma_*\cap \Gamma^+$. By the proof of Theorem \ref{sf}, $X(u):=u^{\lfloor k_\g\rfloor}-a_\g(y,u)$ is a mean-zero random element of $C(M_\g)$ with $\|X\|_\infty\leq F(y)$ and $|X(u)-X(u')|\leq M(y)\|u-u'\|_1$, where $F,M\in L^2(P_*)$ are obtained from the absolutely convergent expansion on a box slightly larger than $M_\g$. By the central limit theorem in $C(M_\g)$ for Lipschitz processes with a square-integrable Lipschitz constant \citep{jainmarcus}, the normalized partial sums are tight, and since $\E_{P_*}[F^2]<\infty$, the compact law of the iterated logarithm \citep{ledouxtalagrand} gives $\limsup_t\|\xi^\g_t\|_\infty/\sqrt{2\log\log t}<\infty$ almost surely. The same argument applies to $(K(\t)-f(y,\t))/\sqrt{K(\t)}$ on the compact set $\{K\geq\e'\}\cap\Theta$, which is covered by finitely many polydiscs on which $f(y,\cdot)$ has an absolutely convergent expansion with $L^2(P_*)$ coefficients (Lemma \ref{regularity}), and on which $1/\sqrt{K}$ is bounded and Lipschitz. Let $\Lambda$ be one plus the max of the finitely many lim sups.
\end{proof}
 
\begin{proof}[Proof of Theorem \ref{main1}]
    Recall that $\underline{\lambda}=\min_{\Theta^*}\lambda(\cdot)$. By Proposition \ref{volumeequiv}, $\underline{\lambda}=\min_{\Theta^*}\lambda^g$ and $\Theta^{**}=\{\t\in\Theta^*:\lambda^g(\t)=\underline{\lambda}\}$. Let $\e_{\min} := \min\{\e, \{\e_\g\}_{\g\in \Gamma^+\setminus \Gamma_*}\}$, and apply Lemma \ref{lilbound} with $\e'=\e_{\min}$. All statements below hold almost surely for all large $t$. By the AM-GM inequality, $tK_t(\t)\geq\frac{1}{2}(tK(\t)-\psi_t(\t)^2)$ when $K(\t)>0$, and $tK_t(g(u))\geq\frac{1}{2}(tu^{2k_\g}-\xi^\g_t(u)^2)$ on charts in $\Gamma_*\cap \Gamma^+$ since $u^{\lceil k_\g \rceil}\leq u^{k_\g}$ on $M_\g$ because $b_\g \leq 1$.
 
    Let $\g\in\Gamma_*\cap\Gamma^+$. If $k_{\g,j}>0$ and $u\in M_\g$ has $u_j=0$, then $g(u)\in\Theta^*$ by Lemma \ref{image}, so $\underline{\lambda}\leq\lambda^g(g(u))\leq\lambda_\g(u)\leq(h_{\g,j}+1)/2k_{\g,j}$. Hence every ratio $(h_{\g,j}+1)/2k_{\g,j}$ with $k_{\g,j}>0$ is at least $\underline{\lambda}$, and Lemma \ref{monomialvolume}(b) with $t/2$ instead of $t$ gives
    \begin{equation}\label{eq:chartupper}
    \int_{M_\g}e^{-tK_t(g(u))}\psi_\g(u)u^{h_\g}du\leq e^{\|\xi^\g_t\|_\infty^2/2}\int_{M_\g}e^{-\frac{t}{2}u^{2k_\g}}\psi_\g(u)u^{h_\g}du\leq C(\log t)^{d+\Lambda^2}t^{-\underline{\lambda}}.
    \end{equation}
 
    Fix $\bar\t\in\Theta^{**}$, and let $u_0$, $\g_0$, and the box $S$ be as in the lower bound of Proposition \ref{volumeequiv} for $\bar\t$ and some radius $r$. On $S\cap\{u^{2k_{\g_0}}\leq 1/t\}$, Theorem \ref{sf} and $u^{\lceil k_{\g_0}\rceil} \leq u^{k_{\g_0}}$ give $tK_t(g(u))\leq 1+X_t$ with $X_t:=\|\xi^{\g_0}_t\|_\infty$. By Theorem \ref{partitiontheta2}(d) and Lemma \ref{monomialvolume}(a),
    \begin{equation}\label{eq:evidence}
    Z_t(\Theta\cap B(\bar\t,r))\geq ce^{-1-X_t}\int_S \1\{u^{2k_{\g_0}} \leq 1/t\} u^{h_{\g_0}}du \geq c'e^{-X_t}t^{-\underline{\lambda}}.
    \end{equation}
    By Lemma \ref{lilbound}, $e^{-X_t}\geq(\log t)^{-1}$, so $Z_t(\Theta)\geq t^{-\underline{\lambda}-\eta}$ for every $\eta>0$.
 
    Now, we prove part (i). Let $U\supseteq\Theta^{**}$ be open and $V:=\Theta\setminus U$. By Theorem \ref{partitiontheta2}(d), $Z_t(V)$ is at most $\int_{V\cap\{K>\e\}}e^{-tK_t}\mu_0 d\t$ plus the sum over charts of $J_\g:=\int_{M_\g\cap g^{-1}(V\cap\Theta_\e)}e^{-tK_t(g(u))}\psi_\g(u)u^{h_\g}du$. The first term is at most $e^{-t\e/2}(\log t)^{\Lambda^2}$ by AM-GM and Lemma \ref{lilbound}. So is $J_\g$ for $\g\in \Gamma^+\setminus\Gamma_*$, with $\e_\g$ in place of $\e$, and $J_\g=0$ for $\g\notin \Gamma^+$ by Lemma \ref{image}. Let $\g\in\Gamma_*\cap\Gamma^+$ and $J^*:=\{j:k_{\g,j}>0, (h_{\g,j}+1)/2k_{\g,j}=\underline{\lambda}\}$. If $u\in M_\g$ has $u_j=0$ for some $j\in J^*$, the argument above gives $\lambda^g(g(u))=\underline{\lambda}$, so $g(u)\in\Theta^{**}$. Hence the compact sets $\{u\in M_\g:u_j=0\text{ for some }j\in J^*\}$ and $M_\g\cap g^{-1}(V)$ are disjoint, and $u_j\geq\d_\g>0$ for $j\in J^*$ on $M_\g\cap g^{-1}(V)$. On this set, $u^{2k_\g}$ is at least a positive constant times the monomial in the coordinates outside $J^*$, whose ratios are greater than $\underline{\lambda}$. Lemma \ref{monomialvolume}(b), applied as in \eqref{eq:chartupper} to those coordinates, gives $J_\g\leq C(\log t)^{d+\Lambda^2}t^{-\lambda'_\g}$ for some $\lambda'_\g>\underline{\lambda}$. If every $j$ with $k_{\g,j}>0$ lies in $J^*$, then $K(g(u))=u^{2k_\g}\geq\d_\g^{2|k_\g|}$ on $M_\g\cap g^{-1}(V)$, so $J_\g$ decays exponentially in $t$ and the same bound holds with $\lambda'_\g:=\underline{\lambda}+1$. Let $\lambda'$ be the smallest $\lambda'_\g$. Then, with $\eta:=(\lambda'-\underline{\lambda})/4$,
    \[
    \mu_t(V)=\frac{Z_t(V)}{Z_t(\Theta)}\leq t^{\underline{\lambda}+\eta}\left(C(\log t)^{d+\Lambda^2}t^{-\lambda'}+Ce^{-t\e_{\min}/2}(\log t)^{\Lambda^2}\right)\leq t^{-(\lambda'-\underline{\lambda})/2},
    \]
    which proves (i) with $\e(U):=(\lambda'-\underline{\lambda})/2$.
 
    For part (ii), suppose $F$ satisfies the hypothesis and $\bar\t\in\Theta^{**}\setminus F$. Choose $r>0$ with $\bar B(\bar\t,r)\cap F=\emptyset$, and let $U:=\Theta\setminus\bar B(\bar\t,r)$, an open set containing $F$. Summing \eqref{eq:chartupper} over the charts, and bounding the remaining terms as in part (i), gives $Z_t(\Theta)\leq C(\log t)^{d+\Lambda^2}t^{-\underline{\lambda}}$. Combined with \eqref{eq:evidence}, $\mu_t(\Theta\setminus U)\geq c(\log t)^{-d-\Lambda^2-1}$ for all large $t$ almost surely, which contradicts $\mu_t(\Theta\setminus U)\leq t^{-\e}$.
\end{proof}
 
\section{Narratives}\label{supp:narr}
 
\begin{proof}[Proof of Proposition \ref{narrdecomp}]
    (a) By Proposition \ref{volumeequiv}, $\lambda$ takes finitely many values $\ell_1<\cdots<\ell_N$ on $\Theta^*$, and $\lambda(\t)=\min_{u\in g^{-1}(\t)}\lambda(u)$. On each chart, $\lambda_\g(u)$ depends on $u$ only through the set of its vanishing coordinates, so $B_{\leq\ell}:=\{u\in g^{-1}(\Theta^*):\lambda(u)\leq\ell\}$ is semianalytic and relatively compact, and $L_{\ell_k}=g(B_{\leq\ell_k})\setminus g(B_{\leq\ell_{k-1}})$, with $B_{\leq\ell_0}:=\emptyset$. Images of relatively compact semianalytic sets under proper real analytic maps are subanalytic, and so are differences of subanalytic sets, so each $L_{\ell_k}$ is subanalytic. Every subanalytic subset of $\Theta^{(R)}$ can be partitioned into connected real analytic submanifolds of $\Theta^{(R)}$ such that every point of $\Theta^{(R)}$ has a neighborhood meeting only finitely many of them. Partition each $L_{\ell_k}$. Since $L_{\ell_k}$ lies in the compact set $\Theta$, which is covered by finitely many such neighborhoods, the partition is finite. Each piece is a connected real analytic submanifold contained in $\Theta^*$ on which $\lambda$ equals $\ell_k$, so it is a narrative. Hence $\Theta^*$ is a finite disjoint union of narratives. Among all such unions, choose one with the fewest elements. If a union of two or more of its elements were a narrative, merging them would give a union with fewer elements. The last claim of (a) holds because $\lambda$ is constant on each $I_i$ and $\Theta^{**}=L_{\ell_1}$.
 
    (b) Let $\{I_1,\hdots,I_n\}$ satisfy (a), and let $C$ be a connected component of some $L_\ell$. Each $I_i$ is connected and lies in a single level set of $\lambda$, so each $I_i$ is either contained in $C$ or disjoint from $C$. The $I_i$ cover $\Theta^*$, so $C$ is the union of the $I_i$ contained in $C$. Among the narratives contained in $C$, choose $J$ containing as many of the $I_i$ as possible. The narrative $J$ contains at least one $I_i$, since each $I_i\subseteq C$ is a narrative. Suppose some $I_i\subseteq C$ is not contained in $J$. Since $C$ is connected, the union of the $I_i$ contained in $J$ and the union of the other $I_i$ contained in $C$ cannot both be closed relative to $C$, so some $I_j\subseteq C$ not contained in $J$ satisfies $\overline{I_j}\cap\overline{J}\neq\emptyset$. By the hypothesis of (b), some narrative $N$ contains $J\cup I_j$. Since $N$ is connected and $\lambda=\ell$ on $N$, $N\subseteq C$, and $N$ contains more of the $I_i$ than $J$. This contradicts. Hence $J$ contains every $I_i\subseteq C$, so $J=C$ and $C$ is a narrative. By (a), $C$ therefore contains exactly one $I_i$. Hence $\{I_1,\hdots,I_n\}$ is the set of connected components of the sets $L_\ell$, which does not depend on the choice in (a).
\end{proof}
 
\begin{proposition}[Existence of robustness profiles]\label{profileexist}
    Fix $a\in A$ and suppress it. Let $I$ be a narrative with complexity $c$.
    \begin{enumerate}[(a)]
        \item If $k>0$ and $(k,\hdots,k)$ satisfies Definition \ref{robustdef}(a)(i) for some coincidences of $I$, then $(k,\hdots,k)$ is a lower robustness profile of $I$. In particular, $(\frac{1}{2},\hdots,\frac{1}{2})$ is a lower robustness profile of $I$. If $K\geq 0$ on an open set $W\subseteq \Theta^{(R)}$ with $W\cap I\neq \emptyset$, then $(1,\hdots, 1)$ is a lower robustness profile.
        \item An upper robustness profile of $I$ exists if and only if $\emptyset\neq W\cap\Theta^*\subseteq I$ for some open $W$. In this case, $(k,\hdots,k)$ is an upper robustness profile of $I$ for some $k>0$.
    \end{enumerate}
\end{proposition}
\begin{proof}
    (a) Let $\t_1$ be a point of the interior of $I$ relative to its manifold where the coincidences $\phi$ and the bound of Definition \ref{robustdef}(a)(i) are defined, and let $\chi=(\phi,\psi)$ be the local coordinates of Appendix \ref{app:twosided}. Since $\Theta$ is convex with nonempty interior, $\Theta$ contains a ball $B(p,\eta)$, and convexity gives $B(x+s(p-x),s\eta)\subseteq\Theta$ for every $x\in\Theta$ and $s\in(0,1]$. Taking $s:=\varrho/(2D)$, where $D$ is the diameter of $\Theta$, shows that for every $x\in I$ some ball of radius $\eta\varrho/(2D)$ is contained in $\Theta\cap B(x,\varrho)$. Since $I$ is a manifold of dimension $d-c$ near $\t_1$, the set $I\cap B(\t_1,r/2)$ contains points $x_1,\hdots,x_M$ with $M\geq C_1\varrho^{-(d-c)}$ and $\|x_i-x_j\|_2\geq 2\varrho$ for all $i\neq j$. The corresponding balls are disjoint, lie in $\Theta\cap B(\t_1,r)$ for small $\varrho$, and consist of points within $\varrho$ of $I$, so
    \[
    \Leb(\{\t'\in\Theta\cap B(\t_1,r):\dist(\t',I)<\varrho\})\geq C_1\varrho^{-(d-c)}v_d\left(\frac{\eta\varrho}{2D}\right)^d=C\varrho^c,
    \]
    where $v_d$ is the volume of the unit ball in $\R^d$. If $\dist(\t',I)<\varrho$, then $|\phi_\ell(\t')|\leq L\varrho$ for a Lipschitz constant $L$ of $\phi$. Taking $\varrho:=\d^{1/k}/L$, the set above lies in $\Theta\cap N^r_\d(k,\hdots,k)$, while $\Leb(N^r_\d(k,\hdots,k))\leq C'\d^{c/k}$, so the ratio in Definition \ref{robustdef}(a)(ii) is bounded below. For the last claims, let $\t'$ be near $\t_1$ and $\t'' := \chi^{-1}(0,\psi(\t')) \in I$, so that $\|\t'-\t''\|_2 \leq L'\|\phi(\t')\|_2$. Since $K$ is real analytic, it is Lipschitz near $\t_1$, and $K=0$ on $I$, so $K(\t') \leq C\|\t'-\t''\|_2\leq CL'\sum_\ell |\phi_\ell(\t')|$, which is Definition \ref{robustdef}(a)(i) for $(\frac{1}{2},\hdots, \frac{1}{2})$. If $K\geq 0$ on an open set $W$ intersecting $I$, take $\t_1\in W$. Every point of $W\cap I$ minimizes $K$ on $W$, so $\nabla K = 0$ on $W\cap I$, and Taylor's theorem gives $K(\t')\leq C\|\t'-\t''\|_2^2\leq CL'^2\|\phi(\t')\|_2^2$, which is Definition \ref{robustdef}(a)(i) for $(1,\hdots, 1)$.
    
    (b) If $\bar\rho$ is an upper robustness profile at $\t_0$, then $K>0$ wherever $\phi\neq 0$ near $\t_0$, so $\Theta^*\subseteq I$ near $\t_0$. Conversely, suppose $\emptyset\neq W\cap\Theta^*\subseteq I$, fix $\t\in W\cap\Theta^*$, and let $F:=\bar B(\t,r)\cap\Theta^*$ for small $r$. The functions $K$ and $\dist(\cdot,F)$ are continuous and subanalytic on $\bar B(\t,r)\cap\Theta$ and vanish exactly on $F$, so the \Lstroke{}ojasiewicz inequality gives $K\geq c_0\dist(\cdot,F)^{\a'}$ there for some $c_0,\a'>0$. Since $F\subseteq I$, $\phi$ vanishes on $F$, so $\max_\ell|\phi_\ell(\t')|\leq L\dist(\t',F)$ near $\t$, and $(\a'/2,\hdots,\a'/2)$ is an upper robustness profile.
\end{proof}
 
\section{Proof of Lemma \ref{dp}}\label{supp:dp}
 
For a bounded weakly continuous $\Phi:\Delta(\Theta)\to\R$, let $B_\Phi(a,\mu):=(1-\b)\bar u(a,\mu)+\b\int_Y\Phi(\mu_{a,y})dQ^a_\mu(y)$ and $(T\Phi)(\mu):=\max_aB_\Phi(a,\mu)$.
 
First, $T\Phi$ is weakly continuous. If $\mu_n\to\mu$ weakly, the densities $Q^a_{\mu_n}(y)=\int_\Theta Q^a_\t(y)d\mu_n(\t)$ converge pointwise, since $\t\mapsto Q^a_\t(y)$ is continuous on the compact set $\Theta$ (Lemma \ref{regularity}), and hence in $L^1$ by Scheff\'e's lemma. For $Q^a_\mu$-almost every $y$, the posteriors $(\mu_n)_{a,y}$ converge weakly to $\mu_{a,y}$, so dominated convergence gives $B_\Phi(a,\mu_n)\to B_\Phi(a,\mu)$, and $\bar u(a,\cdot)$ is weakly continuous by Lemma \ref{regularity}. Second, $T$ is a $\b$-contraction in the sup norm, so it has a unique fixed point $W$ with $\|W\|_\infty\leq\|\bar u\|_\infty$, and $B(a,\cdot):=B_W(a,\cdot)$ and $G_a$ are weakly continuous. This proves (b).
 
For (a), let $P^s_{\sigma,\mu}:=\int_\Theta P_{\sigma,\t}d\mu(\t)$ be the agent's subjective law over histories under a policy $\sigma$, where $P_{\sigma,\t}$ draws $y_{t+1}$ from $Q^{a_{t+1}}_\t$. Under $P^s_{\sigma,\mu}$, the conditional law of $\t$ given the history is $\mu_t$, so the conditional law of $y_{t+1}$ is $Q^{a_{t+1}}_{\mu_t}$. Applying $W(\mu_t)\geq B(a_{t+1},\mu_t)$ at every $t$ and telescoping gives
\[
W(\mu)=V(\sigma,\mu)+\sum_{t\geq 0}\b^t\E_{\sigma,\mu}\left[G_{a_{t+1}}(\mu_t)\right]\geq V(\sigma,\mu),
\]
with equality for the measurable selection $\pi^W(\mu):=\min\arg\max_aB(a,\mu)$, which is optimal at every $\mu$. Since $P_{\sigma,\t}$ does not depend on $\mu$, $V(\sigma,\mu)=\int_\Theta V(\sigma,\d_\t)d\mu(\t)$.
 
For (c), the same identity, started at $t$, shows that $\pi$ is optimal for $\mu_0$ if and only if $G_{\pi(\mu_t)}(\mu_t)=0$ for all $t$, $P^s_{\pi,\mu_0}$-almost surely. The objective law $\P^{\mu_0}_\pi$ is absolutely continuous with respect to $P^s_{\pi,\mu_0}$ on the history up to every $t$. The subjective density $\int_\Theta\prod_{s\leq t}Q^{a_s}_\t(y_s)\mu_0(\t)d\t$ vanishes only where the objective density $\prod_{s\leq t}P^{a_s}_*(y_s)$ does, by common support and Lemma \ref{regularity}(a). Hence the event in (c) also has $\P^{\mu_0}_\pi$-probability one.

\section{Multiple KL-Minimizing Distributions}\label{supp:weakid}
\renewcommand{\theassumption}{S.\arabic{assumption}}\setcounter{assumption}{0}
\renewcommand{\theHassumption}{supp.assumption.\arabic{assumption}}
 
This section drops weak identification (Assumption \ref{fc2}(i)), so that finitely many distinct subjective distributions may minimize the KL divergence. Theorems \ref{main1} and \ref{main2} extend class by class, where a class collects the parameters that induce the same KL-minimizing distribution. As in the main text, the outcome space $Y$ is finite. We maintain Assumption \ref{fc1} and replace Assumption \ref{fc2} with the following.
 
\begin{assumption}\label{fc2multi}
    For each $a\in A$:
    \begin{enumerate}[(i)]
        \item (Finitely many KL-minimizing distributions): The set $\{Q^a_\t:\t\in\Theta^*(a)\}$ is finite. Write $Q^a_1,\hdots,Q^a_{n_a}$ for its distinct elements, and for $i\leq n_a$ let
        \[
        \Theta^*_i(a):=\{\t\in\Theta^*(a):Q^a_\t=Q^a_i\},\quad f^a_i(y,\t):=\log\frac{Q^a_i(y)}{Q^a_\t(y)}.
        \]
        We call $\Theta^*_i(a)$ the $i$-th \textbf{class} of $\Theta^*(a)$.
        \item (Richness): $K^a\not\equiv 0$ on $\Theta^{(R)}$.
        \item (Analyticity): Assumption \ref{fc2}(iii) holds with $f^a_1$ in place of $f^a$.
        \item (Renormalizability within classes): There is $\e>0$ such that, for every $i\leq n_a$, there are an open set $V_{a,i}\supseteq\Theta^*_i(a)$ and a constant $C>0$ with $\E_{P^a_*}[f^a_i(y,\t)]\geq C\E_{P^a_*}[f^a_i(y,\t)^2]$ for all $\t\in V_{a,i}\cap\{\t\in \Theta:K^a\leq\e\}$.
    \end{enumerate}
\end{assumption}
 
Under weak identification, $n_a=1$ and Assumption \ref{fc2multi} reduces to Assumption \ref{fc2}. Every $Q^a_i$ attains the minimal KL divergence, so $\E_{P^a_*}[f^a_i(y,\t)]=K^a(\t)$ for every $i$. The local learning coefficient is defined as in Definition \ref{concentrationset} and does not involve the subjective distributions, so only the concentration set changes.
 
\begin{defn}[Class concentration sets]\label{concmulti}
    For $a\in A$ and $i\leq n_a$, let $\underline{\lambda}_i(a):=\min_{\t\in\Theta^*_i(a)}\lambda(\t\mid a)$ and
    \[
    \Theta^{**}_i(a):=\{\t\in\Theta^*_i(a):\lambda(\t\mid a)=\underline{\lambda}_i(a)\}.
    \]
    The \textbf{concentration set} is $\Theta^{**}(a):=\bigcup_{i=1}^{n_a}\Theta^{**}_i(a)$.
\end{defn}
 
Within each class, the concentration set consists of the parameters with the smallest learning coefficient. The minimal coefficients $\underline{\lambda}_i(a)$ may differ across classes, and every class contributes to $\Theta^{**}(a)$ regardless.
 
\begin{theorem}[Exogenous convergence without weak identification]\label{main1multi}
    Consider the exogenous learning framework under Assumptions \ref{fc1} and \ref{fc2multi}. For every open set $U\supseteq\Theta^{**}$, there is $\e(U)>0$ such that almost surely $\mu_t(\Theta\setminus U)\leq t^{-\e(U)}$ for all large $t$.
\end{theorem}
 
\begin{corollary}\label{tightmulti}
    Consider the exogenous learning framework under Assumptions \ref{fc1} and \ref{fc2multi}, and let $N_\d(S):=\{\t\in\Theta:\dist(\t,S)<\d\}$ for $S\subseteq\Theta$.
    \begin{enumerate}[(a)]
        \item For every class $i$, almost surely there is a sequence $t_k\to\infty$ with $\mu_{t_k}(N_\d(\Theta^{**}_i))\to 1$ for every $\d>0$. If there are at least two classes, then for every class $i$ and all small $\d>0$, almost surely $\liminf_t\mu_t(N_\d(\Theta^{**}_i))=0$, and $\mu_t$ does not converge weakly.
        \item If $F\subseteq\Theta$ is closed and, for every open $U\supseteq F$, there is $\e>0$ such that almost surely $\mu_t(\Theta\setminus U)\leq t^{-\e}$ for all large $t$, then $F\supseteq\Theta^{**}$.
    \end{enumerate}
\end{corollary}
 
Theorem \ref{main1multi} and Corollary \ref{tightmulti}(b) restate Theorem \ref{main1} with the concentration set of Definition \ref{concmulti}. Corollary \ref{tightmulti}(a) shows that the posterior returns to every class infinitely often and, with two or more classes, never settles on one. An action that the agent plays forever must therefore be justified by a belief on each class.
 
\begin{defn}[Uniform singular equilibrium]\label{semulti}
    Under Assumptions \ref{fc1} and \ref{fc2multi}, an action $a\in A$ is a \textbf{uniform singular equilibrium} if for every class $i\leq n_a$ there is a belief $\mu_i\in\Delta(\Theta^{**}_i(a))$ with $a\in BR(\mu_i)$.
\end{defn}
 
\begin{theorem}[Endogenous convergence without weak identification]\label{main2multi}
    Under Assumptions \ref{fc1} and \ref{fc2multi}, every limit action is a uniform singular equilibrium.
\end{theorem}
 
Under weak identification, a uniform singular equilibrium is a singular equilibrium, and Theorem \ref{main2multi} is Theorem \ref{main2}. The remaining results extend as follows. Propositions \ref{invariance} and \ref{narrdecomp}, Theorem \ref{twosided}, and Proposition \ref{profileexist} concern only $\lambda$, which depends on $K^a$ and $\Theta$ alone, and hold without change. Every narrative lies in a single class, since the classes are disjoint and compact (Lemma \ref{classes}). Theorem \ref{main3} holds with $\Theta^{**}(a)$ as in Definition \ref{concmulti}, by the proof in Appendix \ref{app:main3} with Theorem \ref{main1multi} in place of Theorem \ref{main1}.
 
It is important to note that the learning coefficient does not select across classes. Across classes, the log-likelihood ratio between two KL-minimizing distributions is a random walk with mean-zero increments, of order $\sqrt{t}$, while differences in learning coefficients contribute terms of order $\log t$. Thus, Corollary \ref{tightmulti}(a) holds regardless of the ordering of $\underline{\lambda}_1(a),\hdots,\underline{\lambda}_{n_a}(a)$. We conjecture that the learning coefficient can select between classes if the agent updates her beliefs using fractional posteriors, but we do not pursue this here.
 
For the proofs, fix $a\in A$, suppress it, and write $n:=n_a$. Lemma \ref{lambdaprops} holds under Assumption \ref{fc2multi} with the same proof, applied class by class, so each $\Theta^{**}_i$ is nonempty and compact.
 
\begin{lemma}[Classes]\label{classes}
    The classes $\Theta^*_1,\hdots,\Theta^*_n$ are compact and pairwise disjoint. There are $\e_0\in(0,\e]$, $\d_1>0$, and polynomials $p_1,\hdots,p_n$ such that the sets $\Theta_i:=\{\t\in\Theta:K(\t)\leq\e_0, p_i(\t)\geq 0\}$ satisfy:
    \begin{enumerate}[(a)]
        \item the set $\{\t\in\Theta:K(\t)\leq\e_0\}$ is the disjoint union of $\Theta_1,\hdots,\Theta_n$, and $\Theta_i\cap\Theta^*=\Theta^*_i$;
        \item the set $\Theta_i$ contains $N_{\d_1}(\Theta^*_i)$, and every point of $\Theta^*_i$ is a limit of interior points of $\Theta_i$;
        \item the \textbf{restricted environment} $\EE_i$, with parameter space $\Theta_i$, prior $\mu_0/\mu_0(\Theta_i)$ on $\Theta_i$, and $f_i$ in place of $f$, satisfies the hypotheses of Theorem \ref{main1}, with convexity of the parameter space replaced by (b). The local learning coefficient of $\EE_i$ agrees with $\lambda$ on $\Theta^*_i$, and the concentration set of $\EE_i$ is $\Theta^{**}_i$.
    \end{enumerate}
\end{lemma}
\begin{proof}
    The map $\t\mapsto Q_\t$ is continuous on $\Theta$ (Lemma \ref{regularity}(c), whose proof does not use weak identification), and each class is the preimage of a single distribution under this map restricted to the compact set $\Theta^*$, so the classes are compact. The classes are disjoint because $Q_1,\hdots,Q_n$ are distinct. By Tietze's extension theorem and the Weierstrass approximation theorem, there are polynomials $p_i$ with $p_i\geq 1$ on $\Theta^*_i$ and $p_i\leq -1$ on $\Theta^*_j$ for $j\neq i$. The compact sets $\Theta\cap\{p_i\geq 0\}\cap\{p_j\geq 0\}$ for $i\neq j$ and $\Theta\setminus\bigcup_i\{p_i>0\}$ do not intersect $\Theta^*$, so $K\geq c$ on their union for some $c>0$. Take $\e_0<\min\{\e,c\}$. Every $\t\in\Theta$ with $K(\t)\leq\e_0$ then has $p_i(\t)>0$ for exactly one $i$ and $p_j(\t)<0$ for all $j\neq i$, which gives (a).
 
    For (b), the set $O_i:=\{\t\in\Theta:K(\t)<\e_0, p_i(\t)>0\}$ is relatively open in $\Theta$, lies in $\Theta_i$, and contains the compact set $\Theta^*_i$, so $O_i$ contains $N_{\d_1}(\Theta^*_i)$ for some $\d_1>0$, which can be taken equal for all $i$. Since $\Theta$ is convex with nonempty interior, every point of $\Theta^*_i$ is a limit of interior points of $\Theta$ that lie in $O_i$, and these are interior points of $\Theta_i$.
 
    For (c), $\EE_i$ satisfies weak identification because $\Theta_i\cap\Theta^*=\Theta^*_i$, renormalizability by Assumption \ref{fc2multi}(iv), and analyticity because $f_i-f_1=\log(Q_i/Q_1)$ depends only on $y$ and is bounded on the support of $P_*$ (Assumption \ref{fc1}(i)). The parameter space $\Theta_i$ is defined by the analytic inequalities $q_1\geq 0,\hdots,q_k\geq 0$, $\e_0-K\geq 0$, and $p_i\geq 0$, so Theorem \ref{partitiontheta2} and Lemma \ref{image} apply with $\e_0$ in place of $\e$ and with $p_i$ added to $\F'$, to the units in Theorem \ref{partitiontheta2}(a), and to the units that define $\Gamma^+$. The proof of Theorem \ref{main1} uses convexity of the parameter space only in Proposition \ref{volumeequiv}, to find interior points of the parameter space near each point of $\Theta^*$, which (b) provides. For $\t\in\Theta^*_i$ and $r<\d_1$, (b) gives $\Theta\cap B(\t,r)\subseteq\Theta_i$, so $\t$ has the same near-explanations, and hence the same learning coefficient, in $\EE_i$ as in the original environment. The concentration set of $\EE_i$ is therefore $\Theta^{**}_i$.
\end{proof}
 
\begin{lemma}\label{nodominated}
    Let $Y_*:=\{y\in Y:P_*(y)>0\}$, let $\ell(y):=(\log(Q_j(y)/Q_1(y)))_{j\leq n}\in\R^n$ for $y\in Y_*$, and let $\V$ be the linear span of $\{\ell(y):y\in Y_*\}$. For every $i\leq n$, some $v\in\V$ satisfies $v_i>v_j$ for all $j\neq i$.
\end{lemma}
\begin{proof}
    Suppose no such $v$ exists for some $i$. By Gordan's theorem of the alternative, applied to the linear map $v\mapsto(v_i-v_j)_{j\neq i}$ on $\V$, there exist weights $w_j\geq 0$, $j\neq i$, summing to one, with $\sum_{j\neq i}w_j(v_i-v_j)=0$ for all $v\in\V$. Taking $v=\ell(y)$ gives $\log Q_i(y)=\sum_{j\neq i}w_j\log Q_j(y)$ for all $y\in Y_*$. Every $Q_j$ is positive exactly on $Y_*$ (Assumption \ref{fc1}(i)), so H\"older's inequality gives
    \[
    1=\sum_{y\in Y_*}Q_i(y)=\sum_{y\in Y_*}\prod_{j\neq i}Q_j(y)^{w_j}\leq\prod_{j\neq i}\left(\sum_{y\in Y_*}Q_j(y)\right)^{w_j}=1.
    \]
    Equality in H\"older's inequality forces the $Q_j$ with $w_j>0$ to be proportional on $Y_*$, hence equal, and then equal to $Q_i$. This contradicts distinctness of $Q_1,\hdots,Q_n$.
\end{proof}
 
\begin{proof}[Proof of Theorem \ref{main1multi}]
    Let $\e_0$ and $\Theta_1,\hdots,\Theta_n$ be as in Lemma \ref{classes}, and fix $\t_0\in\Theta^*_1$. For $i\leq n$, let $L^{(i)}_t:=\sum_{s\leq t}\log(Q_i(y_s)/Q_1(y_s))$ and $Z^{(i)}_t(S):=\int_Se^{-\sum_{s\leq t}f_i(y_s,\t)}\mu_0(\t)d\t$ for Borel $S\subseteq\Theta$. Since $tK_t(\t)=\sum_{s\leq t}f_i(y_s,\t)-L^{(i)}_t$,
    \begin{equation}\label{eq:factor}
    Z_t(S)=e^{L^{(i)}_t}Z^{(i)}_t(S)
    \end{equation}
    for every Borel $S\subseteq\Theta$ and every $i\leq n$. In particular, $Z^{(i)}_t(S\cap\Theta_i)/Z^{(i)}_t(\Theta_i)=\mu_t(S\cap\Theta_i)/\mu_t(\Theta_i)$ is the posterior probability of $S\cap\Theta_i$ in $\EE_i$. Let $U\supseteq\Theta^{**}$ be open and $V:=\Theta\setminus U$. By Lemma \ref{classes}(a), $\mu_t(V)\leq\mu_t(\{K>\e_0\})+\sum_{i=1}^n\mu_t(V\cap\Theta_i)$.
 
    Since $Z_t(\Theta)\geq Z_t(\Theta_i)$, \eqref{eq:factor} gives $\mu_t(V\cap\Theta_i)\leq Z^{(i)}_t(V\cap\Theta_i)/Z^{(i)}_t(\Theta_i)$, the posterior probability of $\Theta_i\setminus U$ in $\EE_i$. The set $U\cap\Theta_i$ is relatively open in $\Theta_i$ and contains $\Theta^{**}_i$, the concentration set of $\EE_i$. Theorem \ref{main1}(i), applied in $\EE_i$ (Lemma \ref{classes}(c)), gives $\e_i>0$ with $\mu_t(V\cap\Theta_i)\leq t^{-\e_i}$ for all large $t$, almost surely.
 
    The second bound of Lemma \ref{lilbound}, with $\e'=\e_0$, uses neither weak identification nor renormalizability, so the AM-GM inequality gives $Z_t(\{K>\e_0\})\leq e^{-t\e_0/2}(\log t)^{\Lambda^2}$ for all large $t$, almost surely. By \eqref{eq:factor} with $i=1$ and \eqref{eq:evidence} applied in $\EE_1$, we have $Z_t(\Theta)\geq Z^{(1)}_t(\Theta_1)\geq t^{-\underline{\lambda}_1-1}$ for all large $t$. Hence $\mu_t(\{K>\e_0\})\leq t^{-1}$ for all large $t$, almost surely.
 
    With $\e(U):=\frac{1}{2}\min\{\e_1,\hdots,\e_n,1\}$, almost surely $\mu_t(V)\leq(n+1)t^{-2\e(U)}\leq t^{-\e(U)}$ for all large $t$.
\end{proof}
 
\begin{proof}[Proof of Corollary \ref{tightmulti}]
    Keep the notation of the previous proof, and let $\d_1$ be as in Lemma \ref{classes}(b). Fix $i$ and $\d\in(0,\d_1]$. A point of $\Theta_i$ within $\d$ of $\Theta^{**}_j$ for some $j\neq i$ would lie in $\Theta_j$ by Lemma \ref{classes}(b), and $\Theta_j$ is disjoint from $\Theta_i$. Hence $\Theta_i\setminus N_\d(\Theta^{**}_i)\subseteq\Theta\setminus N_\d(\Theta^{**})$, and Theorem \ref{main1multi}, applied to $U:=N_\d(\Theta^{**})$, gives $\e_\d>0$ with
    \begin{equation}\label{eq:reduce}
    \mu_t(N_\d(\Theta^{**}_i))\geq\mu_t(\Theta_i)-t^{-\e_\d}
    \end{equation}
    for all large $t$ almost surely. By \eqref{eq:reduce} with $\d=\d_1/m$ for every $m\in\N$, the first claim of (a) follows once the event $A_i:=\{\limsup_t\mu_t(\Theta_i)=1\}$ has probability one: along a sequence $t_k$ with $\mu_{t_k}(\Theta_i)\to 1$, \eqref{eq:reduce} gives $\mu_{t_k}(N_{\d_1/m}(\Theta^{**}_i))\to 1$ for every $m$, and hence $\mu_{t_k}(N_\d(\Theta^{**}_i))\to 1$ for every $\d>0$.
 
    Replacing finitely many observations $y_1,\hdots,y_k$ by other points of $Y_*$ multiplies the unnormalized posterior density at every date $t\geq k$ by the same function of $\t$, which is continuous and positive on the compact set $\Theta$ and therefore bounded above and away from zero. The replacement therefore changes $1-\mu_t(\Theta_i)$ by at most a constant factor, so membership in $A_i$ does not depend on any finite number of observations. By Kolmogorov's zero-one law, $\P(A_i)\in\{0,1\}$.
 
    It suffices to show $\P(A_i) >0 $. Write $\underline{\lambda}_j$ for the learning coefficient of $\EE_j$. The lower bound \eqref{eq:evidence} and the upper bound in the proof of Theorem \ref{main1}(ii), applied in each $\EE_j$, give an almost surely finite random variable $\bar C$ with $(\log t)^{-\bar C}t^{-\underline{\lambda}_j}\leq Z^{(j)}_t(\Theta_j)\leq(\log t)^{\bar C}t^{-\underline{\lambda}_j}$ for all $j$ and all large $t$, almost surely. Fix $\eta>0$. There are a constant $\bar c$ and $T_1\in \N$ such that, with probability at least $1-\eta$, these bounds hold with $\bar C\leq\bar c$, and $Z_t(\{K>\e_0\})\leq e^{-t\e_0/4}$, for all $t\geq T_1$. Call this event $\Omega_\eta$. Let $\kappa:=\max_j|\underline{\lambda}_i-\underline{\lambda}_j|+2$ and $H_t:=\{L^{(i)}_t-L^{(j)}_t\geq\kappa\log t\text{ for all }j\neq i\}$. On $\Omega_\eta\cap H_t$, \eqref{eq:factor} gives, for $j\neq i$,
    \[
    \frac{Z_t(\Theta_j)}{Z_t(\Theta_i)}=e^{L^{(j)}_t-L^{(i)}_t}\frac{Z^{(j)}_t(\Theta_j)}{Z^{(i)}_t(\Theta_i)}\leq\frac{(\log t)^{2\bar c}}{t^2},\quad\frac{Z_t(\{K>\e_0\})}{Z_t(\Theta_i)}\leq e^{-t\e_0/4}(\log t)^{\bar c}t^{\underline{\lambda}_i},
    \]
    where the second bound uses $L^{(i)}_t=L^{(i)}_t-L^{(1)}_t\geq 0$ on $H_t$. Since $Z_t(\Theta)=Z_t(\{K>\e_0\})+\sum_jZ_t(\Theta_j)$, for every $\d'\in(0,1)$ there is a date $T_2$ such that $\mu_t(\Theta_i)>1-\d'$ on $\Omega_\eta\cap H_t$ for all $t\geq\max\{T_1,T_2\}$. Hence $\P(\mu_t(\Theta_i)>1-\d')\geq\P(H_t)-\eta$ for these $t$.
 
    The vectors $\ell(y_s)$ of Lemma \ref{nodominated} are i.i.d., bounded, and have mean zero, since every $Q_j$ attains the minimal KL divergence. Let $\Sigma$ be their covariance matrix. A vector $u$ satisfies $u^\top\Sigma u=\Var_{P_*}(u^\top\ell(y))=0$ when $u$ is orthogonal to $\V$, so the range of $\Sigma$ is $\V$. By the central limit theorem, $(L^{(j)}_t)_{j\leq n}/\sqrt{t}$ converges in distribution to $G\sim N(0,\Sigma)$, whose law has a positive density on $\V$. Since $\kappa\log t/\sqrt{t}\to 0$, the portmanteau theorem, applied to the open sets $\{v:v_i-v_j>\eta'\text{ for all }j\neq i\}$ with $\eta'\downarrow 0$, gives $\liminf_t\P(H_t)\geq p_i:=\P(G_i>G_j\text{ for all }j\neq i)$. By Lemma \ref{nodominated}, the set $\{v\in\V:v_i>v_j\text{ for all }j\neq i\}$ is a nonempty open subset of $\V$, so $p_i>0$. As $\eta$ was arbitrary, $\liminf_t\P(\mu_t(\Theta_i)>1-\d')\geq p_i$ for every $\d'\in(0,1)$. By the reverse Fatou lemma, $\mu_t(\Theta_i)>1-\d'$ for infinitely many $t$ with probability at least $p_i$, and intersecting these events over $\d'=1/k$, $k\in\N$, gives $\P(A_i)\geq p_i>0$. Hence $\P(A_i)=1$, which proves the first claim of (a).
 
    Let $n\geq 2$, $\d\in(0,\d_1]$, and $j\neq i$. By Lemma \ref{classes}, the sets $N_\d(\Theta^{**}_i)\subseteq\Theta_i$ and $N_\d(\Theta^{**}_j)\subseteq\Theta_j$ are disjoint, so $\mu_t(N_\d(\Theta^{**}_i))\leq 1-\mu_t(N_\d(\Theta^{**}_j))$, and the first claim of (a) for $j$ gives $\liminf_t\mu_t(N_\d(\Theta^{**}_i))=0$ almost surely. If $\mu_t$ converged weakly to some $\nu$ on an outcome in $A_i\cap A_j$, the portmanteau theorem would give $\nu(\overline{N_\d(\Theta^{**}_i)})=1$ and $\nu(\overline{N_\d(\Theta^{**}_j)})=1$ for all small $\d$, so $\nu(\Theta^{**}_i)=\nu(\Theta^{**}_j)=1$, which contradicts $\Theta^{**}_i$ and $\Theta^{**}_j$ disjoint.
 
    To prove (b), suppose some $\bar\t\in\Theta^{**}_i$ lies outside $F$. Choose $r\in(0,\d_1)$ with $\bar B(\bar\t,r)\cap F=\emptyset$, so that $\Theta\cap\bar B(\bar\t,r)\subseteq\Theta_i$. By the proof of Theorem \ref{main1}(ii) applied in $\EE_i$, the posterior probability of $\bar B(\bar\t,r)$ in $\EE_i$ is at least $(\log t)^{-C'}$ for all large $t$, almost surely, for some almost surely finite $C'$. Hence $\mu_t(\bar B(\bar\t,r))\geq(\log t)^{-C'}\mu_t(\Theta_i)\geq\frac{1}{2}(\log t)^{-C'}$ for infinitely many $t$, by (a). The open set $U:=\Theta\setminus\bar B(\bar\t,r)$ contains $F$, so by hypothesis $\mu_t(\bar B(\bar\t,r))\leq t^{-\e}$ for all large $t$, almost surely, which contradicts. Hence $\Theta^{**}\subseteq F$.
\end{proof}
 
\begin{proof}[Proof of Theorem \ref{main2multi}]
    Let $\pi$ be an optimal belief-based policy with $\P^{\mu_0}_\pi(\exists T:a_t=a\text{ for all }t>T)>0$, and construct the history as in Lemma \ref{coupling} with $\mu=\mu_0$. As in Appendix \ref{app:main2}, there is $T$ with $\P(C_T)>0$, where $C_T:=\{a_t=a\text{ for all }t>T\}$. On $C_T$, $\mu^c_{T+s}=\hat\mu^T_s$ for all $s\geq 0$, and $G_a(\mu^c_t)=0$ for all $t\geq T$ almost surely, by Lemmas \ref{coupling} and \ref{dp}(c).
 
    The prior $\mu^c_T$ has a smooth and strictly positive density with respect to $\mu_0$, so the exogenous environment $(\Theta,\{Q^a_\t\}_{\t\in\Theta},\mu^c_T,P^a_*)$ satisfies Assumptions \ref{fc1} and \ref{fc2multi} with the same classes and class concentration sets, which do not depend on the prior. Conditional on $\G_T$, the process $(\hat\mu^T_s)_{s\geq 0}$ is the posterior of this environment under i.i.d. data, so by Corollary \ref{tightmulti}(a), almost surely, for every class $i$ there is a sequence $s_k\to\infty$ with $\hat\mu^T_{s_k}(N_\d(\Theta^{**}_i(a)))\to 1$ for every $\d>0$.
 
    Fix an outcome in $C_T$ at which all of these properties hold, and fix a class $i$. Since $\Delta(\Theta)$ is weakly compact, a subsequence of $(\mu^c_{T+s_k})_k$ converges weakly to some $\mu_i\in\Delta(\Theta)$. By the portmanteau theorem, $\mu_i(\overline{N_\d(\Theta^{**}_i(a))})=1$ for every $\d>0$, and since $\Theta^{**}_i(a)$ is closed, $\mu_i\in\Delta(\Theta^{**}_i(a))$. By Lemma \ref{dp}(b), $G_a(\mu_i)=0$. The proof of Lemma \ref{bridge} uses only that the parameters in the support of the belief induce the same distribution under $a$, and every parameter in $\Theta^*_i(a)$ induces $Q^a_i$, so $a\in BR(\mu_i)$. Since $i$ was arbitrary, $a$ is a uniform singular equilibrium.
\end{proof}

\section{Tightest Convergence}\label{supp:tightest}

Theorem \ref{main1} identifies $\Theta^{**}$ as the smallest closed set to which the posterior concentrates almost surely at a polynomial rate. In this section, we drop the rate. We work in the exogenous learning framework, use the notation of Sections \ref{supp:volume} and \ref{supp:main1}, and recall that $\underline{\lambda}=\min_{\Theta^*}\lambda(\cdot)$.

By Proposition \ref{volumeequiv}, the near-explanations of $\t\in\Theta^*$ have volume of order $\e^{\lambda(\t)}\left(\log\frac{1}{\e}\right)^{m^g(\t)-1}$, so $m(\t):=m^g(\t)$ does not depend on the resolution. We call $m(\t)$ the \textbf{multiplicity} of $\t$, as in \citet{watanabe}.

\begin{defn}[Concurrence set and tightest set]\label{tightestdef}
    Let $\bar m:=\max_{\t\in\Theta^{**}}m(\t)$. The \textbf{concurrence set} is $\Theta^{**}_m:=\{\t\in\Theta^{**}:m(\t)=\bar m\}$. For $\t\in\Theta^{**}$, let
    \[
    \sigma^2(\t):=\lim_{\rho\downarrow 0}\sup\left\{\frac{\Var_{P_*}(f(y,\t'))}{K(\t')}:\t'\in\Theta\cap B(\t,\rho), K(\t')>0\right\}
    \]
    be the local ratio of the variance of the log-likelihood ratio to its mean, which is finite by Assumption \ref{fc2}(iv). A closed set $F\subseteq\Theta$ \textbf{attracts the posterior} if $\mu_t(U)\to 1$ almost surely for every open $U\supseteq F$. The \textbf{tightest set} $\Theta^\dag$ is the intersection of all closed sets that attract the posterior.
\end{defn}

\begin{proposition}[Tightest set]\label{tightestprop}
    Consider the exogenous learning framework.
    \begin{enumerate}[(i)]
        \item $\Theta^\dag$ is nonempty and attracts the posterior.
        \item $\Theta^{**}_m\subseteq\Theta^\dag\subseteq\Theta^{**}$.
        \item $\Theta^\dag\subseteq\{\t\in\Theta^{**}:m(\t)+\frac{1}{2}\sigma^2(\t)\geq\bar m\}$. In particular, if $m(\t)+\frac{1}{2}\sigma^2(\t)<\bar m$ for every $\t\in\Theta^{**}\setminus\Theta^{**}_m$, then $\Theta^\dag=\Theta^{**}_m$.
    \end{enumerate}
\end{proposition}

Part (i) says that $\Theta^\dag$ is the smallest closed set to which the posterior concentrates almost surely. Since a change of prior changes the posterior probability of every set by at most a constant factor, $\Theta^\dag$ does not depend on the prior. Part (ii) says that the smallest set that attracts the posterior is bracketed by $\Theta^{**}_m$ and $\Theta^{**}$. Part (iii) says that, provided fluctuations of the data are small enough, the concurrence set is the smallest set that attracts the posterior. The concurrence set has a clear economic interpretation, which the following example illustrates.

\begin{example}[Concurrence set]
    Consider the exogenous setting of Example \ref{staticexo}, but suppose the agent observes only her output, which she models as ability times circumstances, $Q_{(q,c)}=N(qc,2)$, with $P_*=N(0,1)$ on $\Theta=[0,1]^2$, so she overstates the variance of the noise. Then $K=\frac{1}{4}q^2c^2$, and self-blame $\{q=0\}$ and the narrative $\{c=0\}$, that her circumstances are zero, each require one coincidence with robustness $1$. Hence $\Theta^{**}=\Theta^*$, and Theorem \ref{main1} does not select between the two narratives. At the origin, where both narratives hold, the near-explanations $\{qc\leq 2\sqrt{\e}\}$ have area of order $\sqrt{\e}\log\frac{1}{\e}$ rather than $\sqrt{\e}$, so $m=2$ there and $m=1$ elsewhere on $\Theta^*$. The log-likelihood ratio $f=\frac{1}{4}(q^2c^2-2qcy)$ has variance $\Var_{P_*}(f)=\frac{1}{4}q^2c^2=K$, so $\sigma^2=1$ and $m+\frac{1}{2}\sigma^2=\frac{3}{2}<2=\bar m$ off the origin. By Proposition \ref{tightestprop}(iii), $\Theta^\dag=\Theta^{**}_m=\{(0,0)\}$, and the concurrence set selects both narratives at once, that her ability and her circumstances are both zero.\footnote{In the correctly specified model $Q_{(q,c)}=N(qc,1)$, $\sigma^2=2$, so part (iii) does not apply, but the conclusion still holds. The likelihood depends on $(q,c)$ only through $x=qc$, whose posterior puts mass tending to one on $[0,(\log t)/\sqrt{t}]$ by the law of the iterated logarithm, and given $x$, the prior of $q$ has density proportional to $\mu_0(q,x/q)/q$ on $[x,1]$. Hence $\mu_t(\{q>\rho\})\leq C\log(1/\rho)/\log t$ for all large $t$, and likewise for $c$.} Intuitively, the concurrence set breaks ties between narratives, so that the agent comes to believe that both narratives hold simultaneously, at their intersection.
\end{example}

\begin{proof}[Proof of Proposition \ref{tightestprop}]
    (i) $\Theta^\dag$ is closed as an intersection of closed sets, and $\Theta^{**}$ attracts the posterior by Theorem \ref{main1}. Let $U\supseteq\Theta^\dag$ be open. Every $\t$ in the compact set $\Theta\setminus U$ lies outside some closed set $F_\t$ that attracts the posterior. Choose $\rho_\t\in(0,\dist(\t,F_\t))$. The open set $\R^d\setminus\bar B(\t,\rho_\t)$ contains $F_\t$, so $\mu_t(\bar B(\t,\rho_\t))\to 0$ almost surely. Finitely many of the balls $B(\t,\rho_\t)$ cover $\Theta\setminus U$, so $\mu_t(\Theta\setminus U)\to 0$ almost surely, and $\Theta^\dag$ attracts the posterior. If $\Theta^\dag$ were empty, $U=\emptyset$ would give $\mu_t(\emptyset)\to 1$, which contradicts.

    (ii) Since $\Theta^{**}$ is closed (Lemma \ref{lambdaprops}) and attracts the posterior, $\Theta^\dag\subseteq\Theta^{**}$. For the lower inclusion, fix $\bar\t\in\Theta^{**}_m$ and $r>0$. Let $u_0$, $\g_0$, and $S$ be as in the lower bound of Proposition \ref{volumeequiv} for $\bar\t$, so that $m_{\g_0}(u_0)=\bar m$. Then Lemma \ref{monomialvolume}(a) keeps the factor $(\log t)^{\bar m-1}$ in \eqref{eq:evidence}, and
    \begin{equation}\label{eq:conclower}
    Z_t(\Theta\cap B(\bar\t,r))\geq ce^{-X_t}t^{-\underline{\lambda}}(\log t)^{\bar m-1}.
    \end{equation}
    Let $\g\in\Gamma_*\cap\Gamma^+$. Since $g(0)\in\Theta$ by Lemma \ref{image} and $K(g(0))=0$, $g(0)\in\Theta^*$, so every ratio $(h_{\g,j}+1)/2k_{\g,j}$ with $k_{\g,j}>0$ is at least $\lambda(g(0))\geq\underline{\lambda}$, and if $\underline{\lambda}$ is attained, the number of $j$ attaining it is $m_\g(0)\leq m(g(0))\leq\bar m$. Hence Lemma \ref{monomialvolume}(b) gives \eqref{eq:chartupper} with $e^{Y_t^2/2}(\log t)^{\bar m-1}$ in place of $(\log t)^{d+\Lambda^2}$, where $Y_t:=\max_{\g\in\Gamma_*\cap\Gamma^+}\|\xi^\g_t\|_\infty$. Bounding the remaining terms as in the proof of Theorem \ref{main1}(i),
    \[
    Z_t(\Theta)\leq C\left(e^{Y_t^2/2}t^{-\underline{\lambda}}(\log t)^{\bar m-1}+e^{(Y_t'^2-t\e_{\min})/2}\right),\quad Y_t':=\sup_{\t\in\Theta:K(\t)\geq\e_{\min}}|\psi_t(\t)|.
    \]
    By the central limit theorem in the proof of Lemma \ref{lilbound}, the processes $\xi^\g_t$ and $\psi_t$ are tight, so for every $\eta>0$ there is $M$ with $\P(Y_t\leq M, Y_t'\leq M)\geq 1-\eta$ for all $t$. Since $X_t\leq Y_t$, the two displays give $\mu_t(B(\bar\t,r))\geq c_M:=ce^{-M-M^2/2}/(2C)$ on this event, for all large $t$. By the reverse Fatou lemma, $\mu_t(B(\bar\t,r))\geq c_M$ for infinitely many $t$ with probability at least $1-\eta$. If $\bar\t\notin\Theta^\dag$, then for $r$ with $\bar B(\bar\t,r)\cap\Theta^\dag=\emptyset$, part (i) gives $\mu_t(B(\bar\t,r))\to 0$ almost surely, a contradiction.

    (iii) Let $\t\in\Theta^{**}$ with $m(\t)+\frac{1}{2}\sigma^2(\t)<\bar m$, and choose $\d,\eta\in(0,1)$ with $m(\t)+\d+(\sigma^2(\t)+\d)/(2(1-\eta))<\bar m$. By the proof of the upper bound of Proposition \ref{volumeequiv}, for small $r$, the preimage of $\Theta\cap B(\t,r)$ in the charts of $\Gamma^+$ is contained in finitely many compact boxes $S_i\subseteq M_{\g_i}$ around points $u_i\in g^{-1}(\t)$, with $\g_i\in\Gamma_*\cap\Gamma^+$, $\lambda_{\g_i}(u_i)\geq\underline{\lambda}$, and $m_{\g_i}(u_i)\leq m(\t)$ if equality holds. Shrinking the boxes, $g(S_i)\subseteq B(\t,\rho)$ for a given $\rho>0$. Fix $i$, write $k:=k_{\g_i}$, and let $\tilde\xi_t(u):=u^{\lceil k\rceil-k}\xi^{\g_i}_t(u)$, so that, by Theorem \ref{sf} and the AM-GM inequality,
    \[
    tK_t(g(u))=tu^{2k}-\sqrt{t}u^k\tilde\xi_t(u)\geq\eta tu^{2k}-\frac{\tilde\xi_t(u)^2}{4(1-\eta)}.
    \]
    Where $u^{2k}>0$, the summands of $\tilde\xi_t$ have variance $\Var_{P_*}(f(y,g(u)))/K(g(u))$, so for small $\rho$, their variance is at most $\sigma^2(\t)+\d/2$ on $S_i$, by continuity. Multiplication by $u^{\lceil k\rceil-k}$ is a bounded linear operator on $C(S_i)$, so $\tilde\xi_t$ also satisfies the compact law of the iterated logarithm, and $\limsup_t\|\tilde\xi_t\|_\infty/\sqrt{2\log\log t}$ equals the largest standard deviation of its summands on $S_i$ \citep{ledouxtalagrand}. Hence, almost surely for all large $t$, $e^{\|\tilde\xi_t\|_\infty^2/(4(1-\eta))}\leq(\log t)^{(\sigma^2(\t)+\d)/(2(1-\eta))}$, and Theorem \ref{partitiontheta2}(d) and Lemma \ref{monomialvolume}(b) give
    \[
    Z_t(\Theta\cap B(\t,r))\leq Ct^{-\underline{\lambda}}(\log t)^{m(\t)-1+(\sigma^2(\t)+\d)/(2(1-\eta))}.
    \]
    By \eqref{eq:conclower} and Lemma \ref{lilbound}, $Z_t(\Theta)\geq ct^{-\underline{\lambda}}(\log t)^{\bar m-1-\d}$ for all large $t$, almost surely. By the choice of $\d$ and $\eta$, $\mu_t(B(\t,r))\to 0$ almost surely. Then the closed set $\Theta^\dag\setminus B(\t,r)$ attracts the posterior, since $U\cup B(\t,r)$ is open and contains $\Theta^\dag$ whenever $U\supseteq\Theta^\dag\setminus B(\t,r)$ is open. Hence $\t\notin\Theta^\dag$. The last claim follows from (ii).
\end{proof}
 
\section{Joint Coincidences}\label{supp:joint}
 
Definition \ref{robustdef} treats the coincidences of a narrative one at a time. Each term $|\phi_\ell|^{2\rho_\ell}$ involves a single coincidence, so a violation of any one coincidence increases the KL divergence regardless of whether the others hold. We now replace each term by a \textit{test} $\prod_{\ell=1}^c|\phi_\ell|^{2\tau_\ell}$, which may involve several coincidences, and show that Theorem \ref{twosided} extends with $\frac{1}{2}\sum_{\ell=1}^c1/\rho_\ell$ replaced by the value of a linear program.
 
\begin{defn}[Joint robustness]\label{jointrobustdef}
    Fix $a\in A$ and suppress it. Let $I$ be a narrative with coincidences $\phi_1(\t)=0,\hdots,\phi_c(\t)=0$ at $\t_0\in I$. A \textbf{test} is a nonzero vector $\tau=(\tau_1,\hdots,\tau_c)\geq 0$, and we write $|\phi(\t)|^{2\tau}:=\prod_{\ell=1}^c|\phi_\ell(\t)|^{2\tau_\ell}$. For a nonempty finite set $\T$ of tests, a vector $\rho\in(0,\infty]^c$ is \textbf{compatible} with $\T$ if $\sum_{\ell=1}^c\tau_\ell/\rho_\ell\geq 1$ for every $\tau\in\T$, where $1/\infty:=0$ and $\d^{1/\infty}:=1$ in $N^r_\d(\rho)$. Let $v(\T):=\min\{\sum_{\ell=1}^c1/\rho_\ell:\rho\text{ is compatible with }\T\}$, which is attained and positive.
    \begin{enumerate}[(a)]
        \item $\T$ is a \textbf{lower test set} of $I$ if, for some such $\t_0$ and $\phi$, (i) $K\leq C\sum_{\tau\in\T}|\phi|^{2\tau}$ on $\Theta$ near $\t_0$, and (ii) for some $\rho$ compatible with $\T$ that attains $v(\T)$, $\liminf_{\d\downarrow 0}\Leb(\Theta\cap N^r_\d(\rho))/\Leb(N^r_\d(\rho))>0$ for all small $r>0$.
        \item $\T$ is an \textbf{upper test set} of $I$ if, for some such $\t_0$ and $\phi$, $K\geq c_0\sum_{\tau\in\T}|\phi|^{2\tau}$ on $\Theta$ near $\t_0$.
    \end{enumerate}
    If $\T$ is both a lower and an upper test set with respect to the same coincidences, then $I$ is \textbf{jointly robust} with tests $\T$.
\end{defn}
 
A vector $\rho$ is a lower (upper) robustness profile of $I$ if and only if $\{\rho_\ell e_\ell:\ell\leq c\}$ is a lower (upper) test set of $I$, and $v(\{\rho_\ell e_\ell:\ell\leq c\})=\sum_\ell1/\rho_\ell$, so Definition \ref{jointrobustdef} subsumes Definition \ref{robustdef}.
 
\begin{theorem}\label{jointtwosided}
    Fix $a\in A$ and suppress it. Let $I$ be a narrative. For any lower test set $\underline{\T}$ and any upper test set $\bar\T$ of $I$, $\frac{1}{2}v(\bar\T)\leq\lambda(I)\leq\frac{1}{2}v(\underline{\T})$. In particular, if $I$ is jointly robust with tests $\T$, then $\lambda(I)=\frac{1}{2}v(\T)$.
\end{theorem}
\begin{proof}
    Use the coordinates $\chi=(\phi,\psi)$ of Appendix \ref{app:twosided}. For the upper bound, let $\rho$ be as in Definition \ref{jointrobustdef}(a)(ii) for $\underline{\T}$. On $\Theta\cap N^r_\d(\rho)$, every test satisfies $|\phi|^{2\tau}\leq\d^{2\sum_\ell\tau_\ell/\rho_\ell}\leq\d^2$, so $K\leq C|\underline{\T}|\d^2$, and $\Leb(N^r_\d(\rho))\geq c\d^{v(\underline{\T})}$. The argument of Appendix \ref{app:twosided} then gives $\lambda(I)\leq\frac{1}{2}v(\underline{\T})$. For the lower bound, write $v:=v(\bar\T)$. By linear programming duality, there are $y_\tau\geq 0$ with $\sum_\tau y_\tau=v$ and $\sum_\tau y_\tau\tau_\ell\leq 1$ for every $\ell$. With weights $p_\tau:=y_\tau/v$, the weighted AM-GM inequality gives, for $|\phi_\ell|\leq 1$,
    \[
    \sum_{\tau\in\bar\T}|\phi|^{2\tau}\geq\prod_{\tau\in\bar\T}\left(|\phi|^{2\tau}\right)^{p_\tau}=\prod_{\ell=1}^c|\phi_\ell|^{2\sum_\tau p_\tau\tau_\ell}\geq\prod_{\ell=1}^c|\phi_\ell|^{2/v}.
    \]
    Hence the near-explanations of $\t_0$ lie in $\{\prod_\ell|\phi_\ell|\leq(\e/c_0)^{v/2}\}$, whose volume is at most $C\e^{v/2}\left(\log\frac{1}{\e}\right)^{c-1}$ by Lemma \ref{monomialvolume}(a) with $k=(1,\hdots,1)$ and $h=0$. Taking log, dividing by $\log\e$ and letting $\e\downarrow 0$ gives $\lambda(I)\geq v/2$.
\end{proof}
 
\begin{example}[Joint coincidences]\label{jointex}
    Consider the exogenous setting of Example \ref{staticexo}, but suppose the agent observes a single outcome, her recognition, which she models as $Q_{(q,c)}=N(q(c-d)(q-c),1)$, with $P_*=N(0,1)$, on $\Theta=[0,1]^2$. Then $K(q,c)=\frac{1}{2}q^2(c-d)^2(q-c)^2$ and $\Theta^*=\{q=0\}\cup\{c=d\}\cup\{q=c\}$, three lines on each of which $\lambda=1/2$ away from the intersections. At $(d,d)$, with coincidences $\phi_1:=c-d$ and $\phi_2:=q-c$, the factor $q^2$ is bounded away from zero, so $K\asymp|\phi_1|^2|\phi_2|^2$. The two coincidences are tested only jointly, with $\T=\{(1,1)\}$. A vector $\rho$ is compatible with $\T$ if and only if $1/\rho_1+1/\rho_2\geq 1$, so $v(\T)=1$ and Theorem \ref{jointtwosided} gives $\lambda((d,d))=1/2$. Theorem \ref{twosided} gives only $\lambda\leq 1/2$, since $(d,d)$ is arbitrarily close to other narratives and has no upper robustness profile. In Example \ref{staticexo}, the two outcomes test the coincidences of circumstance-blame separately, so $\lambda=1$. Since recognition depends on the product, violating either coincidence alone leaves the fit unchanged.
\end{example}

\section{Competing Narratives}\label{supp:competing}

\citet{ba} studies an agent who holds two models, uses one at a time, and switches to the other when its Bayes factor exceeds a threshold. \citet{ba} characterizes the models that persist by their \textit{asymptotic accuracy} in predicting the equilibrium distribution of outcomes and the \textit{tightness} of their priors around such equilibria. In the main analysis of \citet{ba}, every model has finitely many parameters and every prior has full support, so a correctly specified model places positive prior probability on the parameters that induce the true distribution. With a continuum of parameters, these parameters form a Lebesgue-null set by Assumption \ref{fc2}(ii), and prior tightness is zero for every model. We show that, of two correctly specified models in the exogenous learning framework, the model with the smaller learning coefficient persists and the other does not, whatever the priors.

\begin{defn}[Competing models]\label{competingdef}
    Consider the exogenous learning framework. A \textbf{model} is a triple $\M_i=(\Theta_i,\{Q^i_\t\}_{\t\in\Theta_i},\mu^i_0)$ such that $(\Theta_i,\{Q^i_\t\}_{\t\in\Theta_i},\mu^i_0,P_*)$ is a correctly specified environment that satisfies Assumptions \ref{fc1} and \ref{fc2}, and $\underline{\lambda}_i$ is its learning coefficient $\underline{\lambda}$ of Definition \ref{concentrationset}. The \textbf{marginal likelihood} of $\M_i$ and the \textbf{Bayes factor} of $\M_j$ against $\M_i$ after $y^t$ are
    \[
    \ell_t(i):=\int_{\Theta_i}\prod_{s\leq t}Q^i_\t(y_s)\mu^i_0(\t)d\t,\quad B_t(j\mid i):=\frac{\ell_t(j)}{\ell_t(i)}.
    \]
    Given models $\M_i$ and $\M_j$ and a switching threshold $\a\geq 1$, the agent starts with $\iota_0\in \{i, j\}$ and uses model $\iota_t\in\{i,j\}$ at date $t$, where
    \begin{equation}\label{eq:switching}
    \iota_{t+1}:=\begin{cases}
    j & \text{if }\iota_t=i\text{ and }B_t(j\mid i)>\a,\\
    i & \text{if }\iota_t=j\text{ and }B_t(j\mid i)<1/\a,\\
    \iota_t & \text{otherwise.}
    \end{cases}
    \end{equation}
    Model $\M_i$ \textbf{persists} against $\M_j$ if $\P(\exists T:\iota_t=i\text{ for all }t\geq T)>0$.
\end{defn}

With $\a>1$, the switching rule \eqref{eq:switching} and persistence are those of \citet{ba}. A misspecified model never persists against a correctly specified one, since the logarithm of the Bayes factor then grows linearly in $t$.

\begin{proposition}[Persistence]\label{persistence}
    Let $\M_i$ and $\M_j$ be models with $\underline{\lambda}_i\neq\underline{\lambda}_j$, and let $\a\geq 1$. Model $\M_i$ persists against $\M_j$ if and only if $\underline{\lambda}_i<\underline{\lambda}_j$. In this case, almost surely, $B_t(j\mid i)\leq t^{-\d}$ for all large $t$, for every $\d<\underline{\lambda}_j-\underline{\lambda}_i$, and the agent uses $\M_i$ at all large $t$, regardless of initial model.
\end{proposition}

The hypothesis $\underline{\lambda}_i\neq\underline{\lambda}_j$ cannot be dropped. If $\M_i$ and $\M_j$ differ only in their priors and have a single KL-minimizing parameter $\t^*$, then $B_t(j\mid i)$, the posterior mean of $\mu^j_0/\mu^i_0$ under $\M_i$, converges almost surely to $\mu^j_0(\t^*)/\mu^i_0(\t^*)$ by Theorem \ref{main1}, and the priors and the threshold determine persistence, as in \citet{ba}.

By Proposition \ref{persistence}, no model is \textit{globally robust} in the sense of \citet{ba}, that is, at a given prior, no model persists against every competing model at every prior of the competing model. Every model has a positive learning coefficient (Proposition \ref{volumeequiv}), and the learning coefficient of a competing model can be arbitrarily small.\footnote{For $k\in\N$ and $|g|\leq 1/2$ with $\E_{P_*}[g]=0$ and $P_*(g\neq 0)>0$, the model with parameter space $[0,1]$ and $Q_\t:=(1+\t^kg)P_*$ has $K\asymp\t^{2k}$ and learning coefficient $1/(2k)$.}

\begin{example}[Competing narratives]\label{competingex}
    Consider Example \ref{staticexo}, where $\Theta^*=I_1\cup I_2$ and $\lambda(I_1)=1/2<1=\lambda(I_2)$ (Example \ref{resex}). The agent holds two models with the subjective distributions of Example \ref{staticexo}, the \textit{low-ability} model $\M_1$ with $\Theta_1=[0,\bar q]\times[0,1]$ and the \textit{high-ability} model $\M_2$ with $\Theta_2=[\bar q,1]\times[0,1]$, where $\bar q\in(0,d)$. Since $\bar q<d$, $\Theta_1\cap\Theta^*=I_1$ and $\Theta_2\cap\Theta^*=I_2$, so the low-ability model explains the data only by self-blame, and the high-ability model only by circumstance-blame. Every point of $\Theta$ within distance $\bar q$ of $I_1$ lies in $\Theta_1$, and $(d,d)$ is interior to $\Theta_2$, so the near-explanations of each narrative are as in Example \ref{resex}, and $\underline{\lambda}_1=1/2<1=\underline{\lambda}_2$. By Proposition \ref{persistence}, at every switching threshold and whichever model she starts with, the agent almost surely uses the low-ability model at all large $t$, and $B_t(2\mid 1)\leq t^{-\d}$ for all large $t$, for every $\d<1/2$. She selects the same narrative as the agent of Example \ref{resex}, whose model contains both narratives. With the subjective distributions of Example \ref{robustex}, $\underline{\lambda}_2=1/k$, so the conclusion reverses for $k\geq 3$, and Proposition \ref{persistence} is silent for $k=2$.
\end{example}

Let $K^i_t$ and $Z^i_t$ be as in Section \ref{supp:main1} for the environment of $\M_i$, with $\t_0\in\Theta^*_i$. Since $\M_i$ is correctly specified, $Q^i_{\t_0}=P_*$ and $\ell_t(i)=Z^i_t(\Theta_i)\prod_{s\leq t}P_*(y_s)$, so
\begin{equation}\label{eq:bfdecomp}
B_t(j\mid i)=\frac{Z^j_t(\Theta_j)}{Z^i_t(\Theta_i)}.
\end{equation}

\begin{lemma}[Bayes factor bounds]\label{bfbounds}
    Let $\M_i$ and $\M_j$ be models. There is an almost surely finite random variable $C$ such that, almost surely, for all large $t$,
    \[
    (\log t)^{-C}t^{\underline{\lambda}_i-\underline{\lambda}_j}\leq B_t(j\mid i)\leq(\log t)^Ct^{\underline{\lambda}_i-\underline{\lambda}_j}.
    \]
\end{lemma}
\begin{proof}
    Apply the proof of Theorem \ref{main1} to the environment of $\M_i$, with $\Theta_i\subseteq\R^{d_i}$. By \eqref{eq:evidence} and Lemma \ref{lilbound}, $Z^i_t(\Theta_i)\geq c_i(\log t)^{-1}t^{-\underline{\lambda}_i}$, and by the proof of part (ii), $Z^i_t(\Theta_i)\leq C_i(\log t)^{d_i+\Lambda_i^2}t^{-\underline{\lambda}_i}$, almost surely for all large $t$, where $c_i,C_i>0$ are constants and $\Lambda_i$ is the random variable of Lemma \ref{lilbound}. Hence $(\log t)^{-2}t^{-\underline{\lambda}_i}\leq Z^i_t(\Theta_i)\leq(\log t)^{d_i+\Lambda_i^2+1}t^{-\underline{\lambda}_i}$ once $\log t\geq\max\{C_i,1/c_i\}$. The same bounds hold for $\M_j$, and \eqref{eq:bfdecomp} gives the claim with $C:=d_i+d_j+\Lambda_i^2+\Lambda_j^2+3$.
\end{proof}

\begin{proof}[Proof of Proposition \ref{persistence}]
   First note that if $B_t(j\mid i)<1/\a$ for all $t\geq T$, then $\iota_t=i$ for all $t\geq T+1$, regardless of $\iota_0$. At any $t\geq T$, if $\iota_t=i$, then $B_t(j\mid i)<1/\a\leq\a$, so $\iota_{t+1}=i$, and if $\iota_t=j$, then $B_t(j\mid i)<1/\a$, so $\iota_{t+1}=i$. Also note that on the event $\{\iota_t=i\text{ for all }t\geq T\}$, $B_t(j\mid i)\leq\a$ for all $t\geq T$, since otherwise the rule would prescribe $\M_j$.

    Suppose $\underline{\lambda}_i<\underline{\lambda}_j$, and let $\d<\underline{\lambda}_j-\underline{\lambda}_i$. By Lemma \ref{bfbounds}, almost surely $B_t(j\mid i)\leq(\log t)^Ct^{-(\underline{\lambda}_j-\underline{\lambda}_i)}\leq t^{-\d}$ for all large $t$. Hence, almost surely, there is $T$ with $B_t(j\mid i)<1/\a$ for all $t\geq T$, so $\iota_t=i$ for all $t\geq T+1$ by the first observation, regardless of $\iota_0$, and $\M_i$ persists against $\M_j$. Conversely, suppose $\underline{\lambda}_i>\underline{\lambda}_j$. By Lemma \ref{bfbounds}, $B_t(j\mid i)\to\infty$ almost surely, so for every $T$, almost surely $B_t(j\mid i)>\a$ for some $t\geq T$. By the second observation, $\P(\iota_t=i\text{ for all }t\geq T)=0$ for every $T$, and $\M_i$ does not persist against $\M_j$.
\end{proof}
\end{document}